\documentclass[12pt,draftcls,onecolumn]{IEEEtran}

\usepackage{graphicx}
\usepackage{algorithmic}
\usepackage{algorithm}
\usepackage{subfigure}
\usepackage{color}
\usepackage{amsmath,amssymb}
\newtheorem{theorem}{Theorem}

\newtheorem{assumption}{Assumption}
\newtheorem{corollary}{Corollary} 
\newtheorem{proposition}{Proposition}
\newtheorem{lemma}{Lemma}
\newtheorem{definition}{Definition}

\newtheorem{remark}{Remark}

\usepackage{epstopdf}
\usepackage{color}

\IEEEoverridecommandlockouts
\newcommand{\real}{\mathbb R}

\title{Decision Making for Rapid Information Acquisition in the Reconnaissance of Random Fields}
\author{Dimitar Baronov and John Baillieul 
\thanks{This work was supported by  ODDR\&E MURI07 Program Grant Number FA9550-07-1-0528 and by the National Science Foundation ITR Program Grant Number DMI-0330171, all to Boston University.}
\thanks{D. Baronov and J. Baillieul are with the Department of Mechanical Engineering, Boston University, Boston, MA, 02215. {\tt\small \{baronov,johnb\}@bu.edu} }}

\date{}

\begin{document}
\maketitle
\begin{abstract}
Research into several aspects of robot-enabled reconnaissance of random fields is reported.  The work has two major components: the underlying theory of information acquisition in the exploration of unknown fields and the results of experiments on how humans use sensor-equipped robots to perform a simulated reconnaissance exercise.   

The theoretical framework reported herein extends work on robotic exploration that has been reported by ourselves and others.  Several new figures of merit for evaluating exploration strategies are proposed and compared.  Using concepts from differential topology and information theory, we develop the theoretical foundation of search strategies aimed at rapid discovery of topological features (locations of critical points and critical level sets) of {\em a priori} unknown differentiable random fields.  The theory enables study of efficient reconnaissance strategies in which the tradeoff between speed and accuracy can be understood.  The proposed approach to rapid discovery of topological features has led in a natural way to to the creation of parsimonious reconnaissance routines that do not rely on any prior knowledge of the environment. The design of topology-guided search protocols uses a mathematical framework that quantifies the relationship between what is discovered and what remains to be discovered. The quantification rests on an information theory inspired model whose properties allow us to treat search as a problem in optimal information acquisition. A central theme in this approach is that ``conservative'' and ``aggressive'' search strategies can be precisely defined, and search decisions regarding ``exploration'' vs.\ ``exploitation'' choices are informed by the rate at which the information metric is changing.

The paper goes on to describe a computer game that has been designed to simulate reconnaissance of unknown fields. Players carry out reconnaissance missions by choosing sequences of motion primitives from two families of control laws that enable mobile robots to either ascend/descend in gradient directions of the field or to map contours of constant field value.  The
strategies that emerge from the choices of motion sequences are classified in terms of the speed with which information is acquired, the fidelity with which the acquired information represents the entire field, and the extent to which all critical level sets have been approximated. The game
thus records each player's performance in acquiring information about both the topology and geometry of the unknown fields that have been randomly generated.

\end{abstract}


\section{Introduction}

The use of mobile point sensors to explore unknown  fields is of current interest across many characteristic length scales and in many domains of science and technology.  Applications ranging from environmental monitoring \cite{Leonard:2006gh}, where mobile sensors must traverse significant distances to map variations in thermal fields or concentrations of chemical species, to nano-scale imaging \cite{Baronov:2009fk}, where the main tool is the Scanning Probe
Microscope, which acquires an image by scanning a point probe over micrometer-scale samples.  Previous work has been focused on utilizing field characteristics (length and time scales) to design optimal reconnaissance strategies \cite{Leonard:2006gh, Demetriou:2009,Cortes:2007uq,Graham:2009vn}.  In these references, the reconnaissance agents are distributed and controlled such that they minimize the error in their estimate of the unknown field.  

The goal of the reconnaissance strategies described in the present paper is to map topological features of unknown random fields.  This feature-based reconnaissance accommodates unknown  field reconstruction aimed at capturing the field's geometry (level sets, curvatures, gradients, etc.) and the field's topology (critical points).  It exploits hierarchically organized feedback loops, wherein low level reactive motion primitives as described in  \cite{Baronov:uq,Baronov:2008fk,Baronov:2010aa,Zhang:2007ys,Cochran:2007fk,Mayhew:2007uq} allow the mobile sensor to track features, while at a high level, models of optimal data acquisition guide the system to efficiently assemble a global representation of important qualitative features of the field. To this end, we make a connection between differential topology and information  using the language of ergodic theory and the  notions of entropy that were first applied to ergodic theory by Kolmogorov.  Following essentially the same development that is presented in \cite{Parry} and \cite{Walters} for ergodic mappings, we shall define the {\em entropy} of a scalar field by means of a three stage process.  We begin by defining the {\em partition entropy} of finite partitions of compact sets.  Next, we define the entropy of a function on a compact domain with respect to a finite partition of its range.  Ultimately, this leads to a definition of the {\em entropy}  of the function itself.  One of the main elements in our information-theoretic approach to the study of unknown fields is the relationship between the entropy of the field and a certain finite partition of the domain called the {\em topology-induced partition}.  This sought-after connection between the differential topology of the field and its entropy is given in Theorem 3.

The relationship between topology and entropy guides our exploration of reconnaissance strategies.  Recalling previously published work \cite{Baronov:uq,Baronov:2008fk,Baronov:2010aa}, two families of sensor motion control laws are introduced in terms of which reconnaissance strategies can be implemented.  These families consist of (a) motions that climb or descend gradients of the unknown field and (b) motions that follow lines along which the value of the field is constant.  By appropriately switching between the control laws, data is acquired about the unknown field.  As more and more lines of constant value ({\em isolines}) become known, they provide the boundaries between cells in an increasingly fine set of partitions of the domain that we call {\em data-induced partitions}.  By keeping track of the rate of increase of data-induced partition entropy as isolines are mapped, we obtain essential insight into how effectively the reconnaissance strategy is discovering information about the unknown field.

While every isoline that is mapped increases the entropy of the data-induced partition, it is only certain isolines that yield information about the topology-induced partition. Specifically, it is shown that by mapping a new isoline in a cell in the data-induced partition with negative Euler characteristic, new information regarding the topology-induced partition may be obtained. Using this observation, a reconnaissance strategy is described in which our robotic motion primitives instantiate  the well known {\it exploration} versus {\it exploitation} paradigm. Specifically, the gradient ascend/descend motions are used to discover local maxima and minima in the field. The search for these constitutes the {\it exploration} phase of the reconnaissance. The localization of max's and min's enables the {\it exploitation} phase in which isolines can be mapped so as to provide an increasingly complete set of information regarding the topological characteristics of the field. One can think of the isoline mapping strategy that is designed using knowledge of the fields' extremum points as providing a steepest ascent along an {\it information gradient} aimed  at learning as much as possible about the topology induced partition. 

A decision criterion based on the rate at which an information metric is increasing under isoline mapping informs choices of whether to next map isolines or gradients in the proposed reconnaissance strategy. The criterion involves parameters similar to those in {\it simulated annealing}, and a comparison is made between {\it aggressive strategies} (which place high value on finding extremum points) and {\it conservative strategies} (emphasizing isoline mapping to fill in geometric detail regarding the field). Using Monte-Carlo simulations, reconnaissance strategies are compared in terms of how rapidly they acquire information about the topological characteristics of the unknown field. 

In the final part of the paper,   we turn  to the study of human performance in the kinds of reconnaissance missions we have abstractly characterized. Work reported in \cite{Gillner:1998} suggests that in exploring unknown environments humans assimilate information in the form of topological maps.   The same circle of ideas will play a role  in our determination of  styles emerging in the  human guided reconnaissance of unknown  fields.  A computer game that simulates a human guided robot-enabled search is described.  The game is designed so that a human mission director can have a robot map either gradient ascending/descending lines or map isolines of an unknown field.  A metric of bias toward the acquisition of topological information is proposed, and in terms of this metric, we assess the styles of twenty-seven subjects who played the game. All subjects exhibited some bias toward acquiring topological information, but the range of the bias was broad. The players who sought and discovered information about the topology-induced partition tended to be very parsimonious in terms of the numbers of isolines that were mapped. They also tended toward consistency in their performance from one game to the other.

\begin{figure}[htbp] 
   \centering
   \includegraphics[width=400pt]{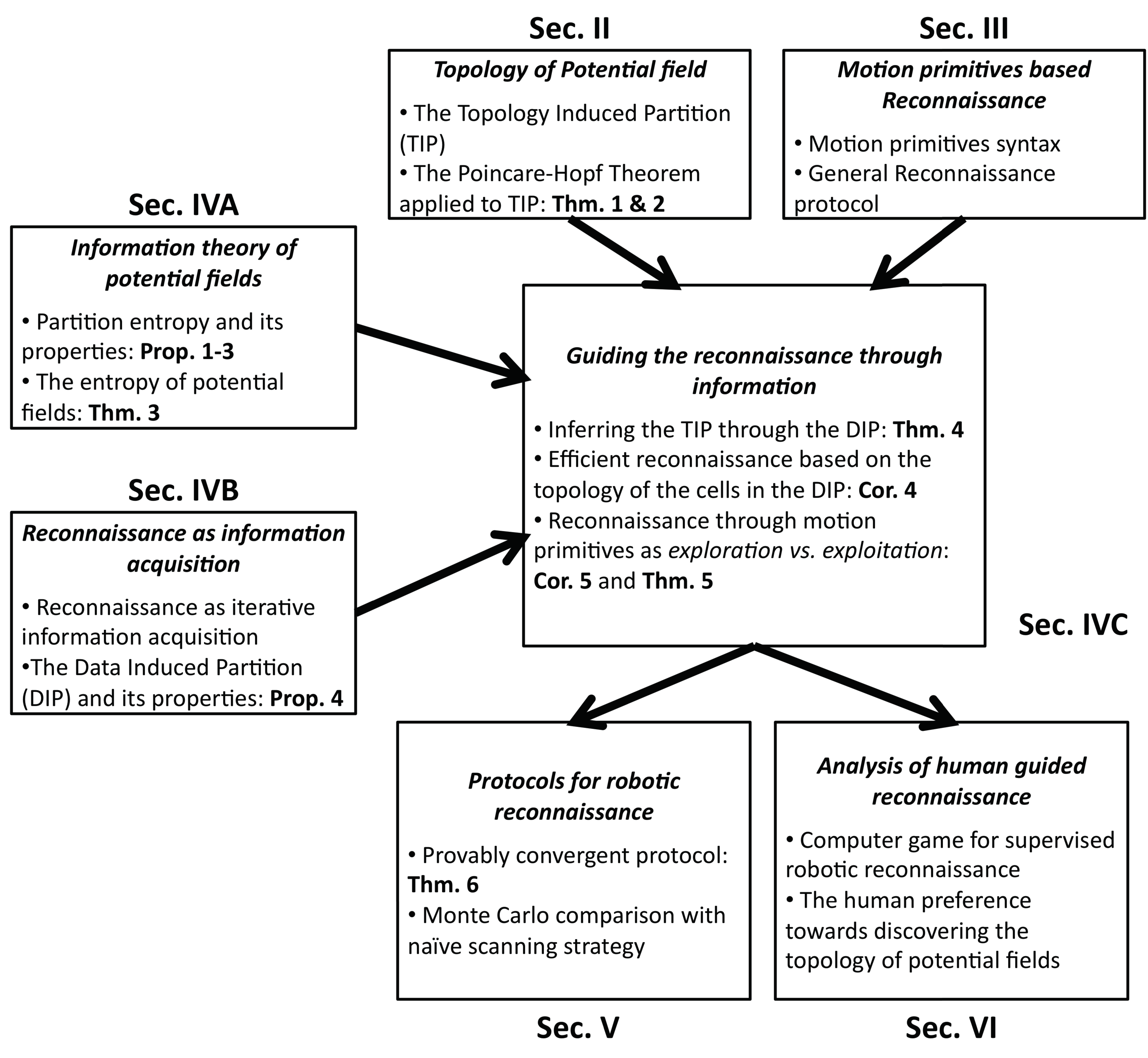} 
   \caption{A schematic of this paper organization and concepts progression.}
   \label{fig:mainConcept}
\end{figure}

The road map of this paper together with the progression of the main concepts is provided in Fig. \ref{fig:mainConcept}. In the next section, we present a mathematical abstraction of scalar  fields based on their critical point structure. In Section \ref{sec:motionPrimitives}, we show how reconnaissance protocols for uncovering the  field's topological structure can be based on the motion primitives developed in \cite{Baronov:uq,Baronov:2008fk,Baronov:2010aa}. This is further formalized in Section \ref{sec:information}, where the reconnaissance process is described  as the acquisition of information as measured by a certain Shannon-like entropy metric. In Section \ref{sec:machineRecon}, we propose protocols for unknown field reconnaissance and compare them by Monte-Carlo simulations. Human in-the-loop reconnaissance is treated  in Section \ref{sec:humanRecon}, where we describe an experiment in which human subjects play a computer game of simulated robotic reconnaissance. In Section \ref{sec:conclusion}, we offer some concluding remarks and discuss open problems.

The following notation will be  used throughout the paper. Notation that is contained within a single section is omitted for simplicity.  

\begin{itemize}
\item Information
\begin{itemize}
\item $H(\mathcal{V}_k)$---The entropy of $\mathcal{V}_k$
\item $H(\mathcal{M}|\mathcal{V}_k)$---The conditional entropy of $\mathcal{M}$ with respect of $\mathcal{V}_k$.
\item $\overline{H}(\mathcal{V}_k)$---The bound of $H(\mathcal{M}|\mathcal{V}_k)$ given $H(\mathcal{V}_k)$.
\end{itemize}
\item Reconnaissance
\begin{itemize}
\item $f(\mathbf{r})$---the scalar field for reconnaissance.
\item $X$---The reconnaissance domain.
\item $b^{iso}(\mathbf{r}_o)$---An isoline mapping motion program that starts from $\mathbf{r}_o$
\item $b^{grad}(\mathbf{r}_o)$---An gradient line mapping motion program that starts from $\mathbf{r}_o$
\item $B_k$---A sequence of $k$ motion programs.
\item $S(B_k)$---The set of all mapped level sets, including extremum points, by $B_k$.
\item $Q(B_k)$---The set of all mapped gradient lines by $B_k$.
\end{itemize}
\item Topology
\begin{itemize}
\item $\xi$---A level set of $f$.
\item $\zeta$---A gradient line of $f$.
\item Cr$(f,X)$---The critical level sets of $f$ within $X$.
\item Cr$^{0,2}(f,X)$---The critical level sets of $f$ within $X$ that are single points (maxima and minima).
\item Cr$^1(f,X)$---The critical level sets of $f$ within $X$ that are contours (corresponding to index $1$ critical points).
\item $\mathcal{M}(f,X)$---The topology induced partition of $f$ within $X$.
\item $\mathcal{V}_k$---A data induced partition after the execution of $k$ motion programs.
\item $V_k^i$---A cell in $\mathcal{V}_k$.
\item $\chi(V_k^i)$---The Euler characteristic of the set $V_k^i$.
\item $\mathcal{V}'_k$---A subset of $\mathcal{V}_k$ containing cells with Euler characteristics smaller or equal to $-1$. 
\end{itemize}
\end{itemize}

\section{Topology of a scalar field}
\label{sec:topology}
The following section presents a formal treatment of the topology of scalar functions defined on compact planar domains. In particular, we present a construction called the {\it topology induced partition} of a function $f:\real^2\to\real$, which is based on the  concept of {\it monotonic sets} that can be traced to the image processing  literature \cite{Biasotti:2000fk,Sole:2004uq,Shinagawa:1991uq}. Here, however, the description of the topology is defined in the context of the data  acquisition particulars of the search process. In \cite{Baillieul:2009kx}, we have introduced the related notion of a {\it monotone search sequence}. These concepts and their relationship to reconnaissance decisions will be discussed in what follows. 

The function under consideration is a scalar  field,  $f:X\to\real$, defined on a compact,  connected and simply-connected domain $X\subset \real^2$ called the {\it search domain}. To avoid pathological behavior which would not contribute to the current discussion, the following technical assumptions are imposed on $f$:
\begin{assumption}
\label{as:morseFunction}
The function $f:X\to\real$ is a Morse function.
\end{assumption}

This assumption implies that   its critical points are non-degenerate, and therefore isolated. 
\begin{assumption}
\label{as:boundary}
The boundary of the search domain,  $\partial X$, is a level contour of $f$,  $f|_{\partial X} = $const.  
\end{assumption}

Basic results in Morse theory \cite{Matsumoto:1997fk} allow us to describe the topological characteristics of functions that satisfy these assumptions.  Conditions under which a random field will be a Morse function are given in \cite{Adler2}. To make the analysis invariant to scaling, we assume that range of the unknown function always reflects the full dynamic range of the sensor being modeled. 

\subsection{The topology induced partition}
To investigate the topology of  scalar  fields through Morse theory,  we will consider these functions as surfaces in three dimensions. Then, as in \cite{Matsumoto:1997fk}, $f:X\to\real$ can be considered as the height function, and important topological characteristics  can be described through the changes in the number of connected components of
\begin{equation}
\label{eq:Omega}
X_c = \left\{\mathbf{r}\in X:f(\mathbf{r})\geq c\right\},
\end{equation}
as $c$ is allowed to vary over the range of $f$. 

Assume without loss of generality that the range is the unit interval $[0,1]$. Then as $c$ decreases from $1$ to $0$, we observe in Fig. \ref{fig:threeValues}  that the number of connected components of $X_c$ changes at the critical points of the height function---connected components appearing at the maxima (Fig. \ref{subfig:1extremum} and \ref{subfig:2extrema}), merging at the saddle points (Fig. \ref{subfig:saddle}) and disappearing at local minima (not shown).

\begin{figure}[htbp] 
   \centering
    \subfigure[]{\label{subfig:1extremum}\includegraphics[width=120pt]{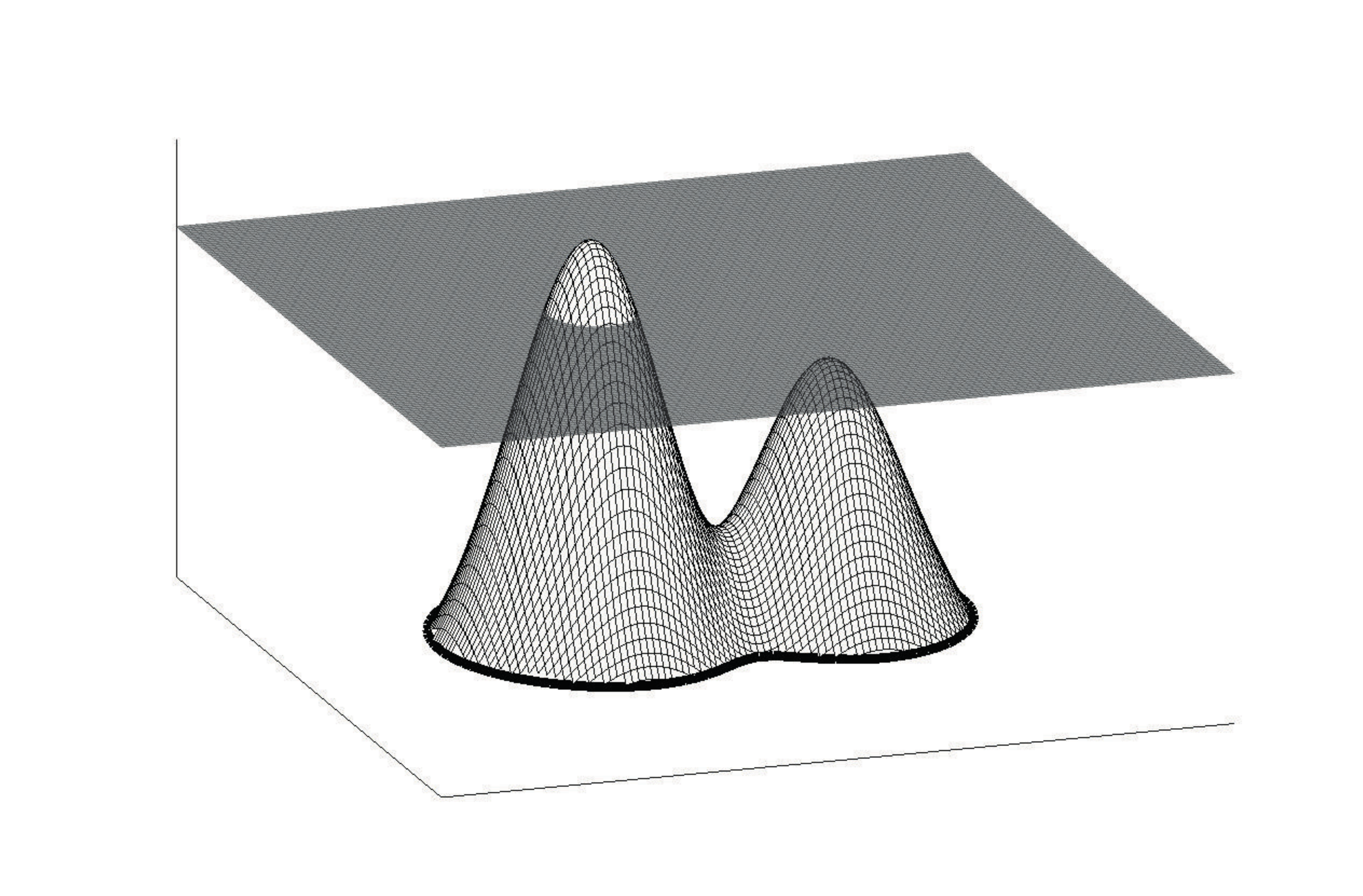}}
    \subfigure[]{\includegraphics[width=120pt]{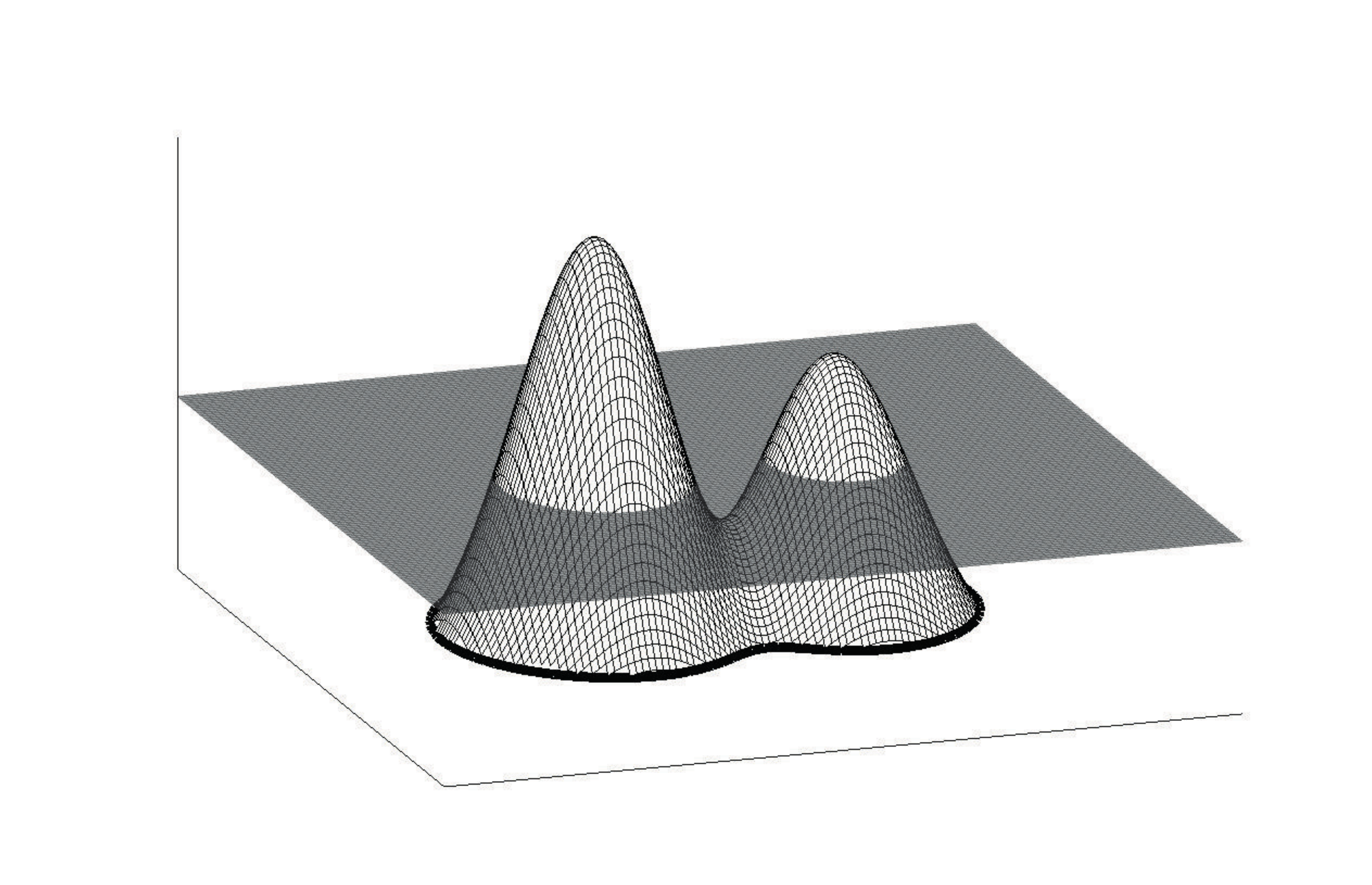}\label{subfig:2extrema}}
    \subfigure[]{\includegraphics[width=120pt]{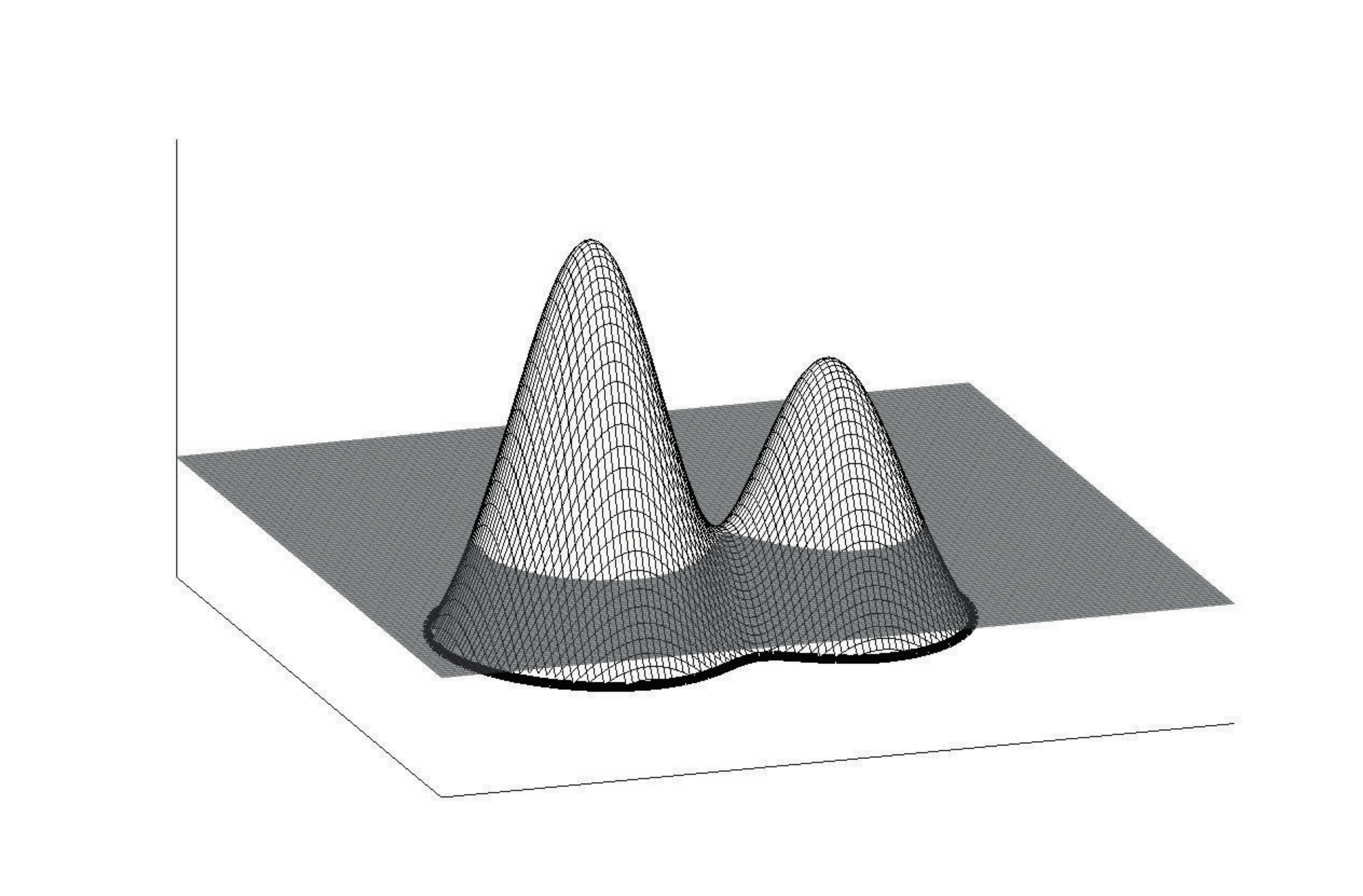}\label{subfig:saddle}}
   \caption{The height map for three different values}
   \label{fig:threeValues}
\end{figure}

Our objective is to study the differential topology of the  surface by decomposing it into component parts, each of which belongs to a homotopy class that is a member of a family of homotopy classes determined by the critical point structure of $f$. We obtain a corresponding decomposition of the domain $X$ into {\it diffeomorphic} components. {\it Diffeomorphisms} preserve homotopy classes and the Euler characteristic invariant. Recall that for a connected set, the Euler characteristic is defined as $\chi =${\it the number of vertices} $-${\it the number edges}$+${\it the number of faces} for any triangulation of the set \cite{Dodson:1997bk}. For a disk, the Euler characteristic is $1$, and it is easy to see that introducing a hole reduces the Euler  characteristic by $1$, since this effectively removes a face from the simplicial decomposition. Hence for an annulus, the Euler characteristic is 0, and other planar domains and their Euler characteristics are illustrated in Fig. \ref{fig:topology}.

\begin{figure}[htbp] 
   \centering
    \subfigure[$\chi = 1$]{\includegraphics[width=80pt]{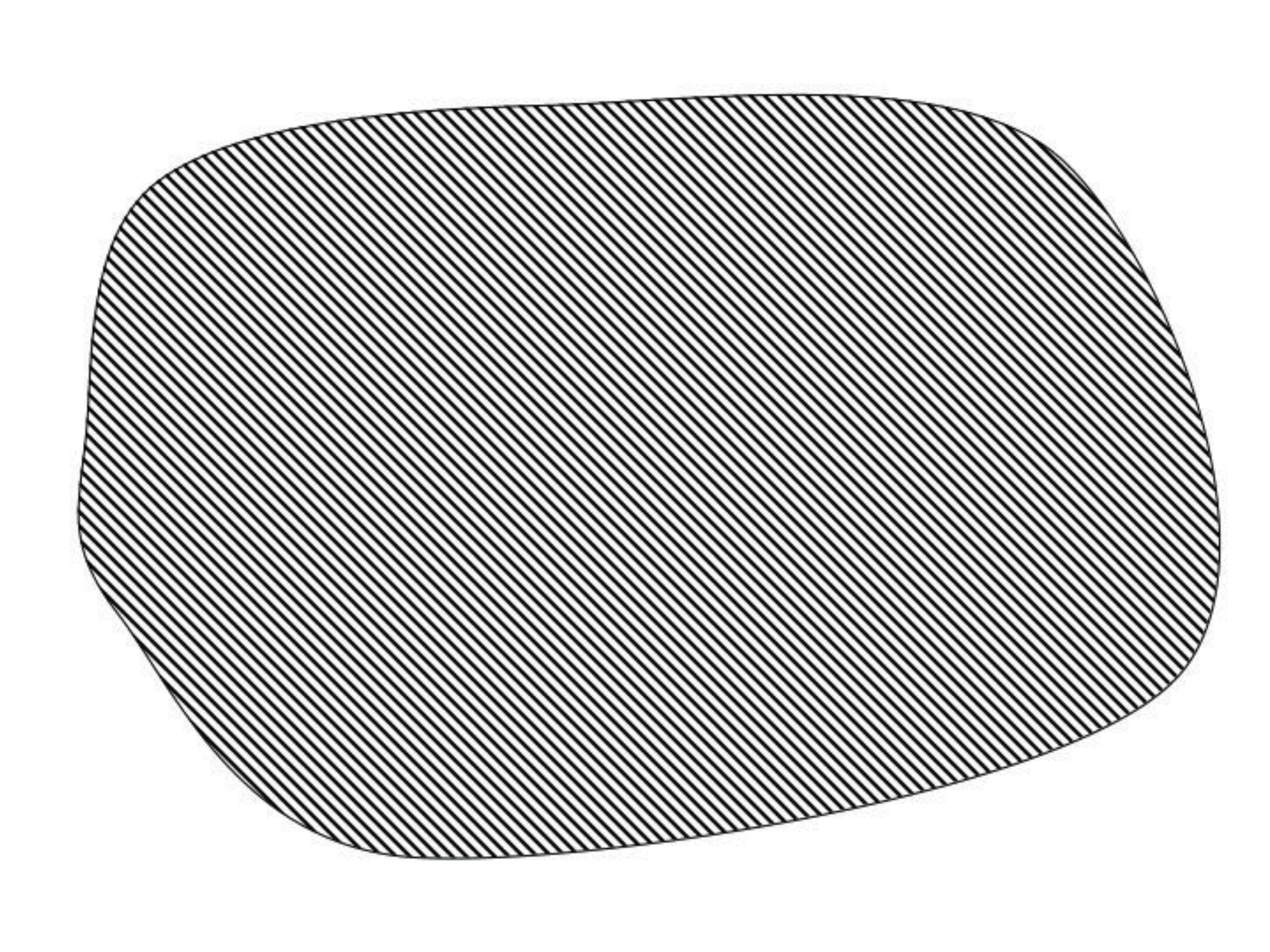}\label{subfig:euler0}}
    \subfigure[$\chi = 0$]{\includegraphics[width=80pt]{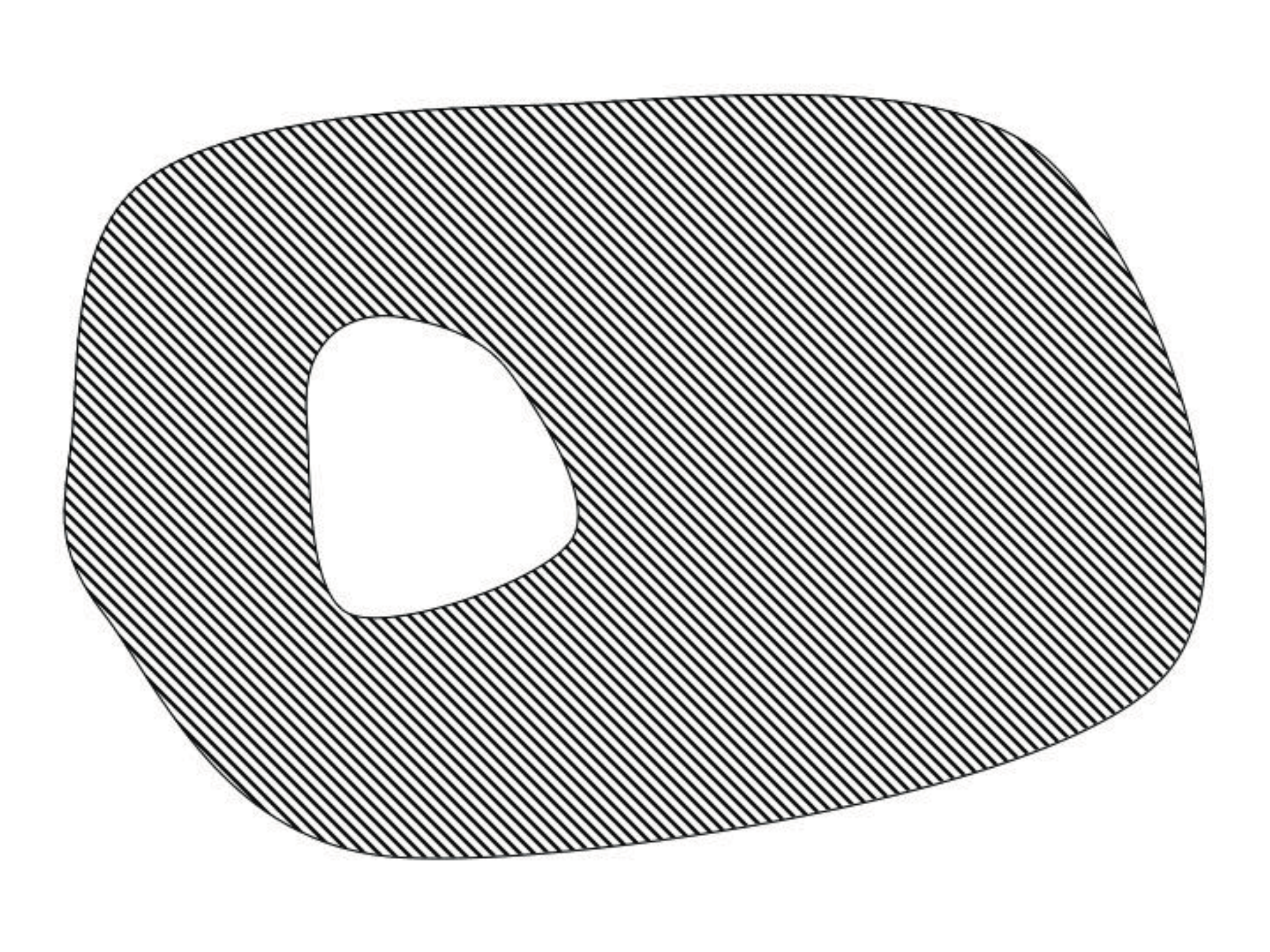}\label{subfig:euler1}}
    \subfigure[$\chi = -1$]{\includegraphics[width=80pt]{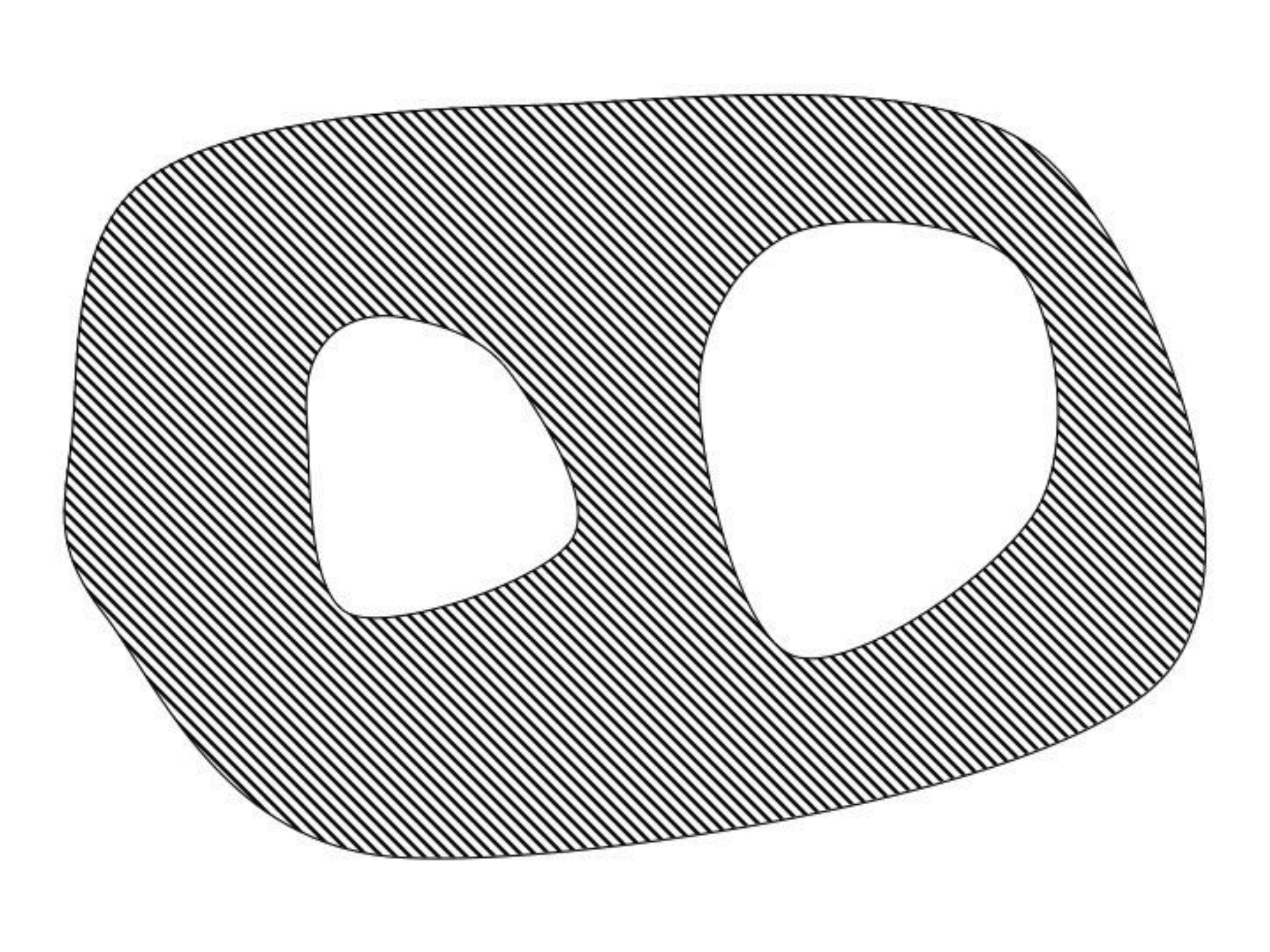}\label{subfig:euler2}}
   \caption{For any planar domain with sufficient regular boundary, the Euler characteristic is $1-g$ where $g$ is the number of holes.}
   \label{fig:topology}
\end{figure}

  What we show next is that there exists a unique decomposition of the  surface into critical level sets and annuli.

A well-known result from Morse theory is the following.
\begin{lemma}[\cite{Matsumoto:1997fk}]
Define for some arbitrary $c_1<c_2$ the sets    $X_{c_1}$ and $X_{c_2}$ associated with the height map according to \eqref{eq:Omega}. Assume that there is no critical value $c^*$ s.t. $c^*\in [c_1,c_2]$, where a critical value is defined as a value for which there exist a critical point $\mathbf{r}^*$, $f(\mathbf{r}^*)=c^*$. Then, $X_{c_1}$ is diffeomorphic to $X_{c_2}$. 
\end{lemma}

From this lemma, it follows that we may partition the surface into connected components having the same Euler characteristic between  the critical values of the function. 

\begin{definition}
For any domain $V\subset X$, we denote the set of connected components of $V$ by cc$(V)$. 
\end{definition}

\begin{definition}
A {\it critical level  set} of the function $f$ in the domain $X$ is a connected component $\xi(c^*)$, $f(\mathbf{r}^*) = c^*$, with $\mathbf{r}^*$ being a critical point, which satisfies  
\begin{eqnarray*}
\xi(c^*) &\in& \text{cc}(\{\mathbf{r}\in X : f(\mathbf{r})=c^*\})\\
\mathbf{r}^*&\in& \xi(c^*).
\end{eqnarray*} 
\end{definition}

The set of all {\it critical level sets} of $f$ in the domain $X$ will be  denoted by Cr$(f,X)$. In addition, let $\text{Cr}^{0,2}(f,X)$ denote critical level sets of  dimension  $0$ (extremum points---i.e. values of $f$ at local maxima and local minima), and let $\text{Cr}^1(f,X)$ denote critical level sets of  dimension $1$ (contours associated with saddle points), so Cr$(f,X)$ $= \text{Cr}^{0,2}(f,X)\cup \text{Cr}^1(f,X)$.

Finally, define a set of annular subsets of the domain $X$,
\[
\mathcal{M}(f,X) = \text{cc}\left(X\setminus\text{Cr}\left(f,X\right)\right).
\]
We call $\mathcal{M}$ the {\it topology induced partition} of the function $f$. We note that $\mathcal{M}$ is strictly speaking not a partition, since the critical level sets are removed, and hence the cells are open sets. Fig. \ref{fig:exampleTIP} shows the topology induced partition of a particular  field. The white areas are magnifications of critical level sets. 

\begin{figure}[htbp] 
   \centering
 \includegraphics[width=200pt]{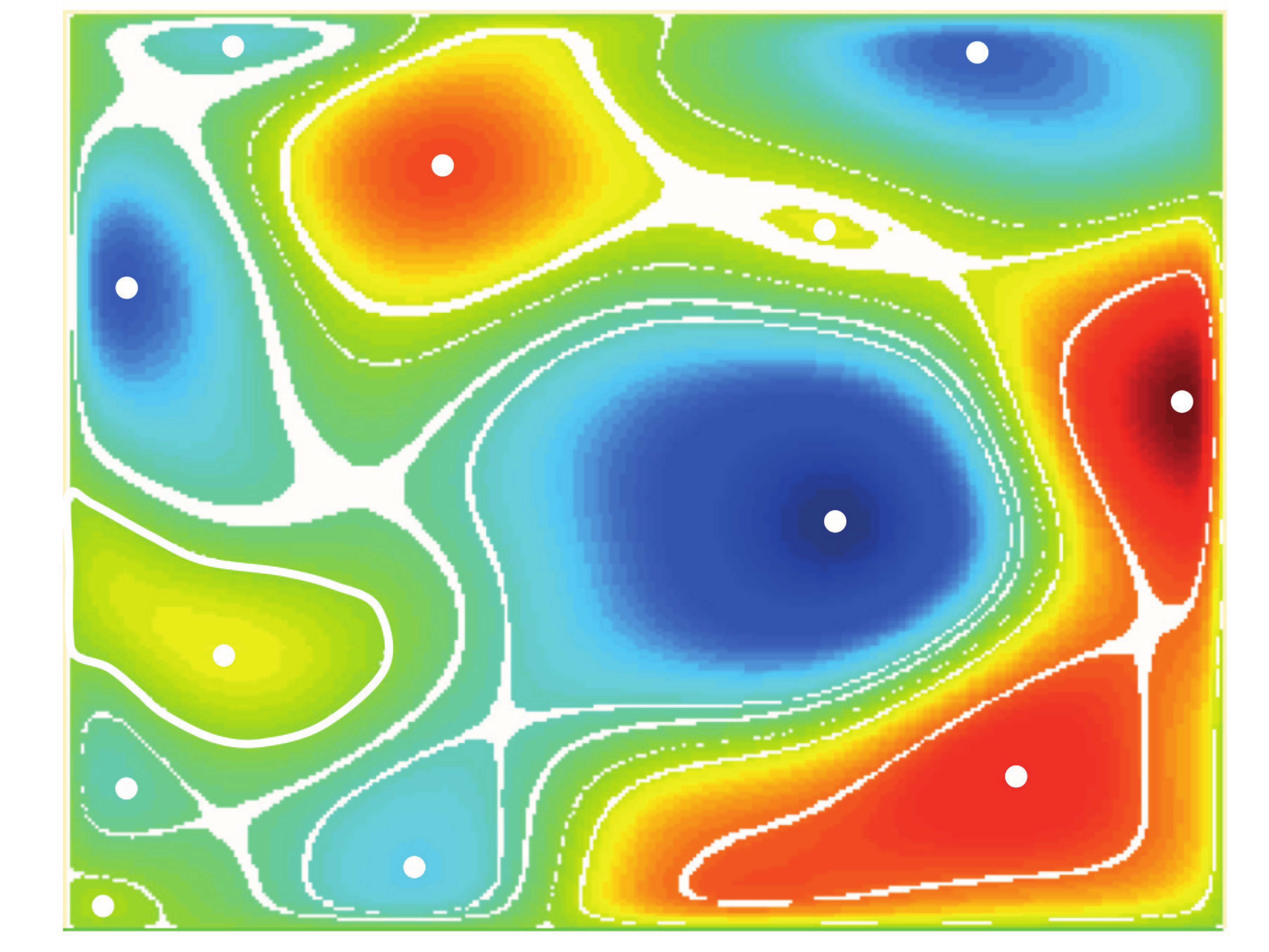} 
   \caption{The topology induced partition of a particular  field.}
   \label{fig:exampleTIP}
\end{figure}

\subsection{Properties of the topology induced partition}
To establish the properties of the topology induced partition, we use the Poincar\'e-Hopf theorem to relate the number of critical points to the Euler characteristic of the set in which this function is contained (see Def. 1 for the notation in the theorem).   As noted,  the  surface can be viewed as a manifold embedded in three dimensional space. Thus, the gradient of the  field, $\nabla f$, induces a corresponding vector field on the  surface. In view of these observations, we restate the  Poincar\'e-Hopf theorem in the following manner. (See \cite{Guillemin:1974fk}.)

\begin{theorem}[Poincar\'e-Hopf]
\label{thm:P-H}
Let $f:X\to\real$ be as above satisfying Assumptions \ref{as:morseFunction} and \ref{as:boundary} in a compact domain $X\subset\real^2$. Let $V\subseteq X$ be a connected set with a boundary $\partial V$ having the property that each connected component of $\partial V$ is a level set of $f$.   Then, if $m$ is the number of the extremum points (index $0$ and index $2$ critical points) of $f$ in $V$, and $n$ is the number of its saddle points (index $1$ critical points),
\begin{equation}
m - n  = \chi(V).
\end{equation}  
\end{theorem}

The next corollary  shows how the Poincar\'e-Hopf theorem can be restated in terms of the critical level sets, Cr$(f,X)$. 

\begin{corollary}
\label{thm:PHT}
Let the set  $V\subseteq X$ be  connected with the connected components of its boundary being level sets of a Morse function, $f:X\to\real$, and let $m$ and $n$ be respectively the cardinalities $m = \left|\left\{{\xi}\in\text{Cr}^{0,2}(f,X)|{\xi}\subset V\right\}\right|$ and 
$n = \left|\left\{{\xi}\in\text{Cr}^1(f,X)|{\xi}\subset V\right\}\right|$. Suppose  that with every critical set, ${\xi}\in \text{Cr}^1(f,X)$,  there is only a single associated saddle point. 
 Then $m-n = \chi(V)$.
\end{corollary}
\begin{proof}
Let the curves $\xi_0,\xi_1,\cdots,\xi_l$ be the connected components of $\partial V$,  having the same orientation such that in traversing $\xi_0$ in the positive sense, int$(V)$ lies inside $\xi_0$, while in traversing $\xi_j$, $j=1,\cdots,l$, in the same positive sense, int$(V)$ lies outside of $\xi_j$. That is to say that $V$ is a compact multiply connected domain whose outer boundary is $\xi_0$ and whose interior boundaries are $\xi_1,\cdots,\xi_l$. Then, $\chi(V)$ will be equal to $1-l$, since every  $\xi_j$, $j=1,\cdots,l$ will correspond to a hole in $V$. 

Denote by $m_i$, $n_i$  the numbers of even index and odd index critical points respectively in the interior of each boundary curve, int$(\xi_i)$. Each int$(\xi_i)$ is diffeomorphic to a disk, and we can apply Theorem \ref{thm:P-H}.  Moreover, taking into account that $\xi_0$ is the outer boundary, we can write:
\begin{eqnarray*}
m &=& m_0 - \sum_{i=1}^l m_i\\
n &=& n_0 - \sum_{i=1}^l n_i\\
\end{eqnarray*}
and therefore
\[
m -n =m_0-n_0- \sum_{i=1}^l (m_i-n_i) = 1 - l = \chi(V).
\]
\end{proof}

\begin{figure}[htbp] 
   \centering
\includegraphics[width=240pt]{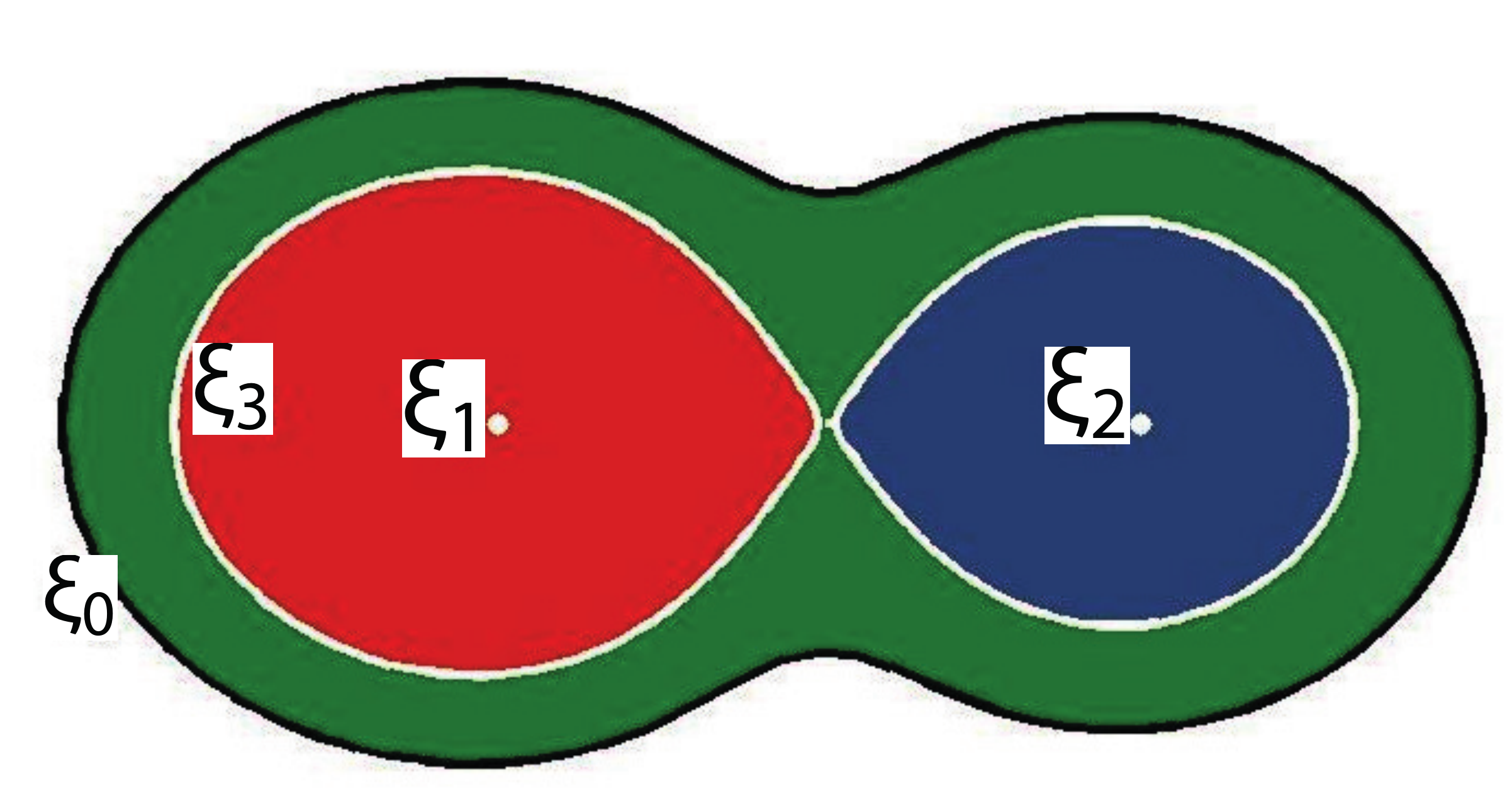}\includegraphics[width=100pt]{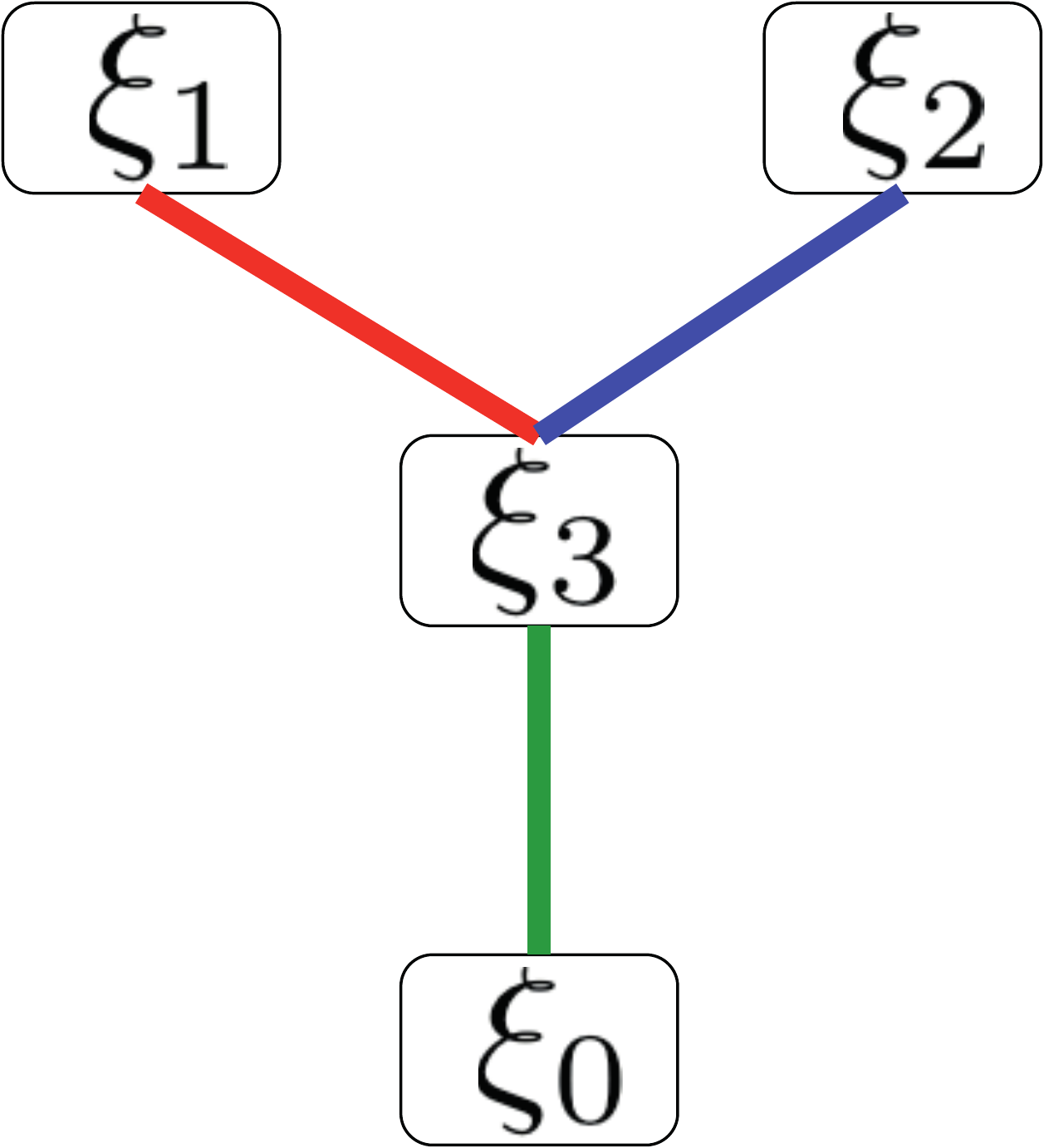}
 \caption{A particular  topology induced partition and its graph representation.\label{fig:graphR} Critical level sets correspond to vertices and cells correspond to edges.}
\end{figure}

To further elaborate the properties of the topology induced partition, we recall that each cell $M_i\in\mathcal{M}$ is an annular component, and therefore has two boundaries one of which may degenerate to a point. Thus, a cell can be viewed as connecting two level sets. As a result, we will be able to associate  the topology induced partition to a tree graph, $G(v,E)$, where the vertices correspond either to  critical level sets or to the boundary of the search domain, $\partial X$. The edges of the graph will be associated with the annular  cells of the topology induced partition. In this way, the leaf vertices will correspond to extrema or to the boundary, $\partial X$. The vertices with degree $3$ will be associated with the elements of Cr$^1(f,X)$.  Such graphs are referred to as {\it Reeb graph}s, \cite{Biasotti:2000fk}.  Fig. \ref{fig:graphR} illustrates the graph representation of a simple topology induced partition. 

The Reeb graph representation implies the following cardinality for the topology induced partition.
\begin{theorem}
\label{th:TIPcardinality}
Let $f:X\to[a,b]$ satisfy Assumptions \ref{as:morseFunction} and \ref{as:boundary}, let Cr$^{0,2}(f,X)\subset \text{Cr}(f,X)$ be the set of all  extremum points  of $f$ within $X$, and suppose that  for each element of Cr$^1(f,X)$ there is a single associated index $1$ critical  point.  Then, for the {\it topology induced partition}, $\mathcal{M}(f,X)$,  the following holds: 
\[
\left|\mathcal{M}(f,X)\right| = 2 \left | \text{Cr}^{0,2}(f,X) \right| -1.
\]  
\end{theorem}

\begin{proof}
Since $G(E,v)$ is a tree graph,  
\[
|E| = |v| -1  = |\text{Cr}(f,X)| + 1 - 1,
\]
where it is taken into account that the boundary is also a vertex. However, Corollary \ref{thm:PHT} gives the following relationship:
\[
|\text{Cr}^{0,2}(f,X)| - |\text{Cr}^1(f,X)| = 1.
\]
This yields
\[
|E| = \left|\mathcal{M}(f,X)\right| =2 \left | \text{Cr}^{0,2}(f,X) \right| -1.
\]
\end{proof}

The graph representation will be further utilized to illustrate different aspects of the proposed reconnaissance strategies.


\section{Reconnaissance with motion primitives}
\label{sec:motionPrimitives}

A motion primitive for unknown  field reconnaissance denotes a feedback control law that allows the vehicle to navigate the   field to track either a path along a level set, or a path of increasing or decreasing  intensity. Thus, the vehicle can map features of the  surface and accumulate information about its general topology, and about its topology induced partition in particular.

In the design of the reconnaissance strategy, we consider two motion programs,  $b^{iso}(\mathbf{r}_o)$  and $b^{grad}(\mathbf{r}_o)$, that map respectively isolines, $\xi$, and gradient paths, $\zeta$, passing through the point $\mathbf{r}_o$. One can ultimately think of a motion program as a construction that utilizes  feedback control laws to track a specific feature of the unknown  field. In   \cite{Baronov:uq,Baronov:2008fk,Baronov:2010aa}, we have shown how isoline and gradient tracking control laws can be designed that rely purely on the intensity measurements of the field. 

The $b^{iso}(\mathbf{r}_o)$ will produce the level set contour that passes through the point $\mathbf{r}_o$. The $b^{grad}(\mathbf{r}_o)$, on the other hand,  will produce a contour passing thorough $\mathbf{r}_o$ that is tangent to the gradient at each point, and whose end points either lie on the boundary, or at one of the extremum points of the function.

A reconnaissance protocol will be a rule for choosing points $\mathbf{r}_o$ and the motion programs that should be executed from them. Thus, a program $b^{iso}(\mathbf{r}_o)$ will consist of two steps: i) go to the point  $\mathbf{r}_o$, and ii) map the isoline passing through $\mathbf{r}_o$.  A sequence of  $k$ motion programs will be denoted by $B_k = \{b_1,b_2,\cdots,b_k\}$, where  $b_i\in\{b^{iso},b^{grad}\}$ ($i\in\{1,2,\dots,k\}$) and every program has its own originating point $\mathbf{r}_o^i$. In what follows, the reconnaissance  protocols that will be discussed will be restricted to choosing the originating points for every consecutive motion program from the paths traced by  a previously executed program.  In more detail, the different  features mapped as a result of a string $B_k$ will be denoted by
\begin{equation}
\label{eq:mappedIsoDefined}
S(B_k) = \{\xi_1,\xi_2,\cdots, \xi_m\},
\end{equation}
for the collection of mapped level sets and extremum points, and by 
\[
Q(B_k) = \{\zeta_1,\zeta_2,\cdots, \zeta_l\},
\]
for the mapped gradient paths.  Then, in terms of this notation, the proposed reconnaissance protocol is depicted in  Fig. \ref{fig:searchArch}. The originating points for the isoline following motion programs will be chosen from the set of mapped gradient paths, and vice versa, the originating points for the gradient following motion programs from the set of mapped level contours.

\begin{figure}[htbp] 
   \centering
   \includegraphics[width=200pt]{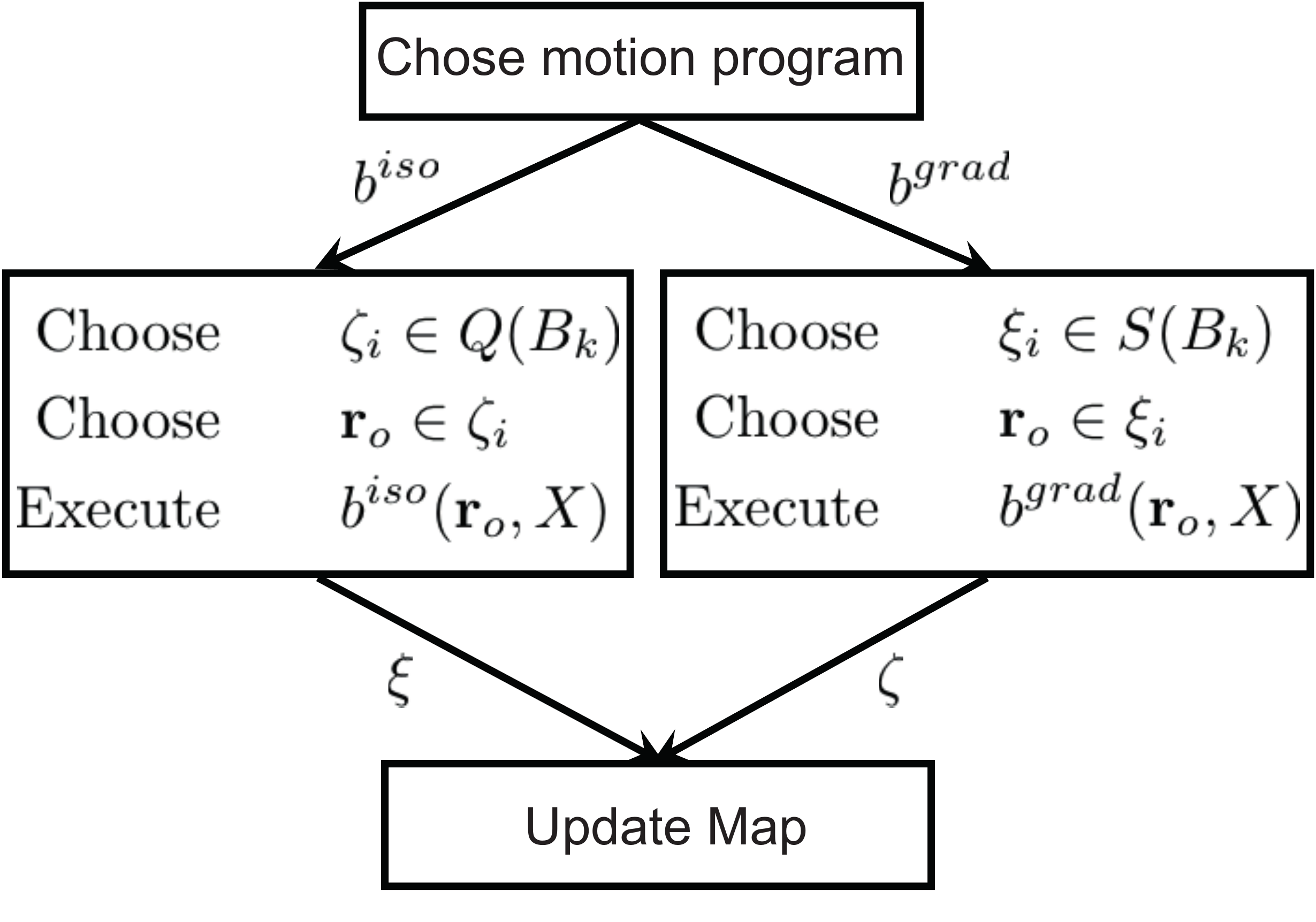} 
   \caption{The general protocol of applying the $b^{iso}$ and the $b^{grad}$ motion programs into search strategies. First,  a particular mapping program is  chosen (isoline or gradient line mapping). Then, if the program is isoline mapping, $\mathbf{r}_o$ is chosen from a previously mapped gradient line, and vice versa, if the program is gradient line mapping,  $\mathbf{r}_o$ is chosen from a previously mapped isoline. The mapped object is then added to the map. }
   \label{fig:searchArch}
\end{figure}

 We illustrate this  reconnaissance protocol by  a particular example. Consider a scalar  field with a topology induced partition such as the one illustrated in Fig. \ref{subfig:topExample}. The critical level sets are depicted as dashed curves,  but these are unknown at the initiation of the search. Initially the only level set that is known is the boundary of the domain. We assume that the agent choses two sufficiently separated random points on this boundary and maps the  gradient lines that originate from them with the $b^{grad}$ motion program (Fig. \ref{subfig:after2grads}).
 
\begin{figure}[htbp] 
   \centering
    \subfigure[]{\label{subfig:topExample}\includegraphics[width=100pt]{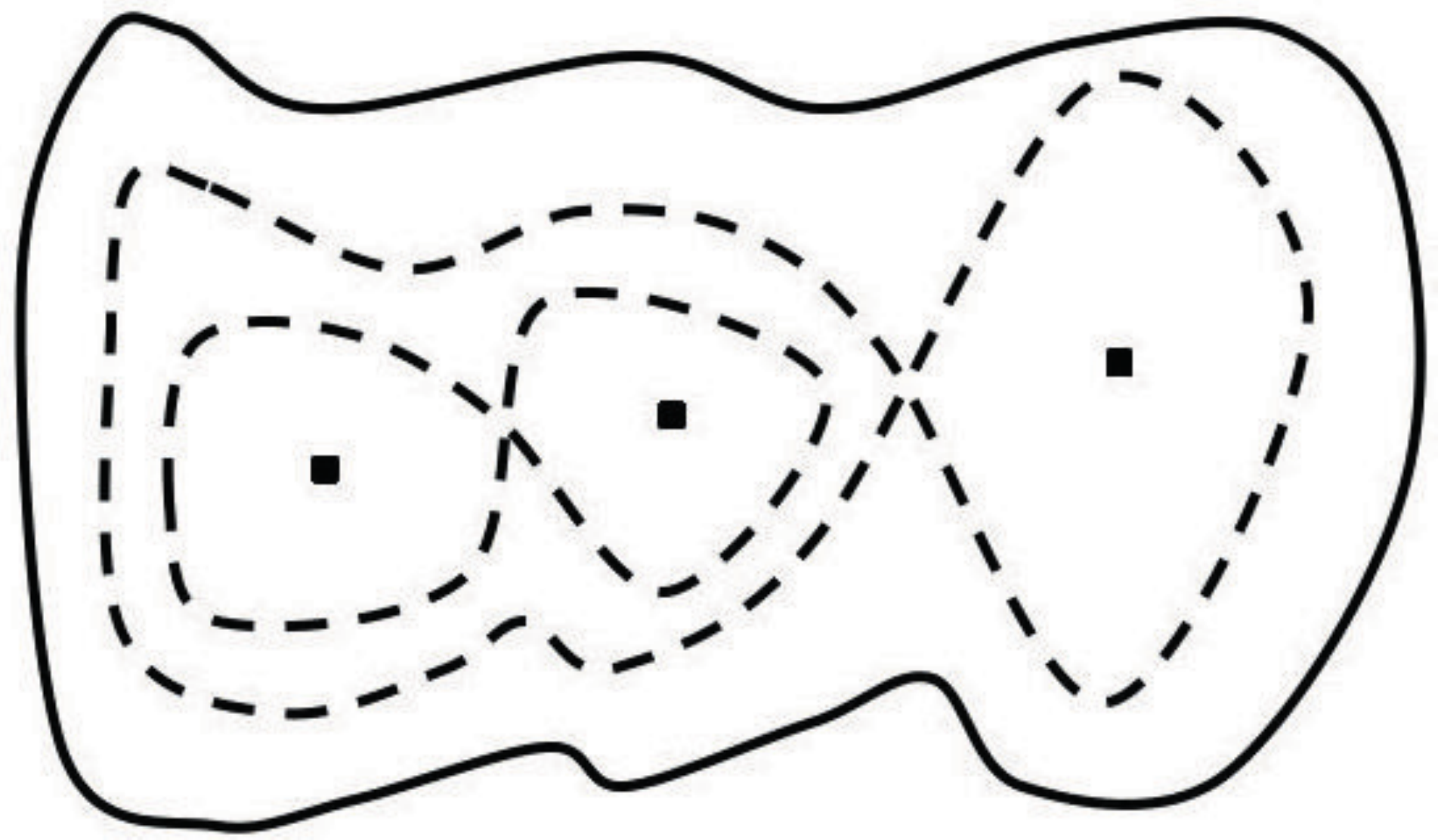}}
    \subfigure[]{\includegraphics[width=100pt]{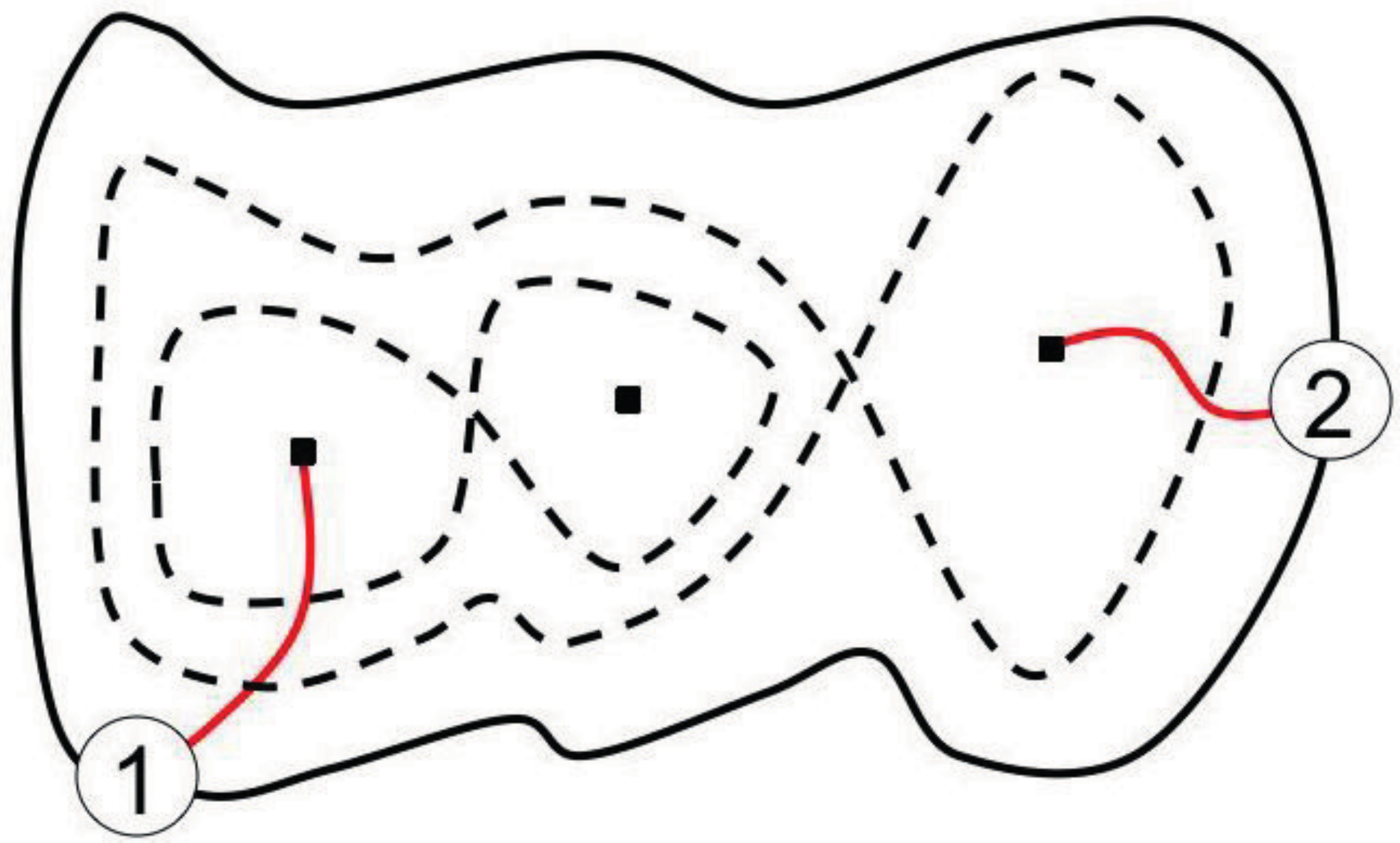}\label{subfig:after2grads}}
    \subfigure[]{\includegraphics[width=100pt]{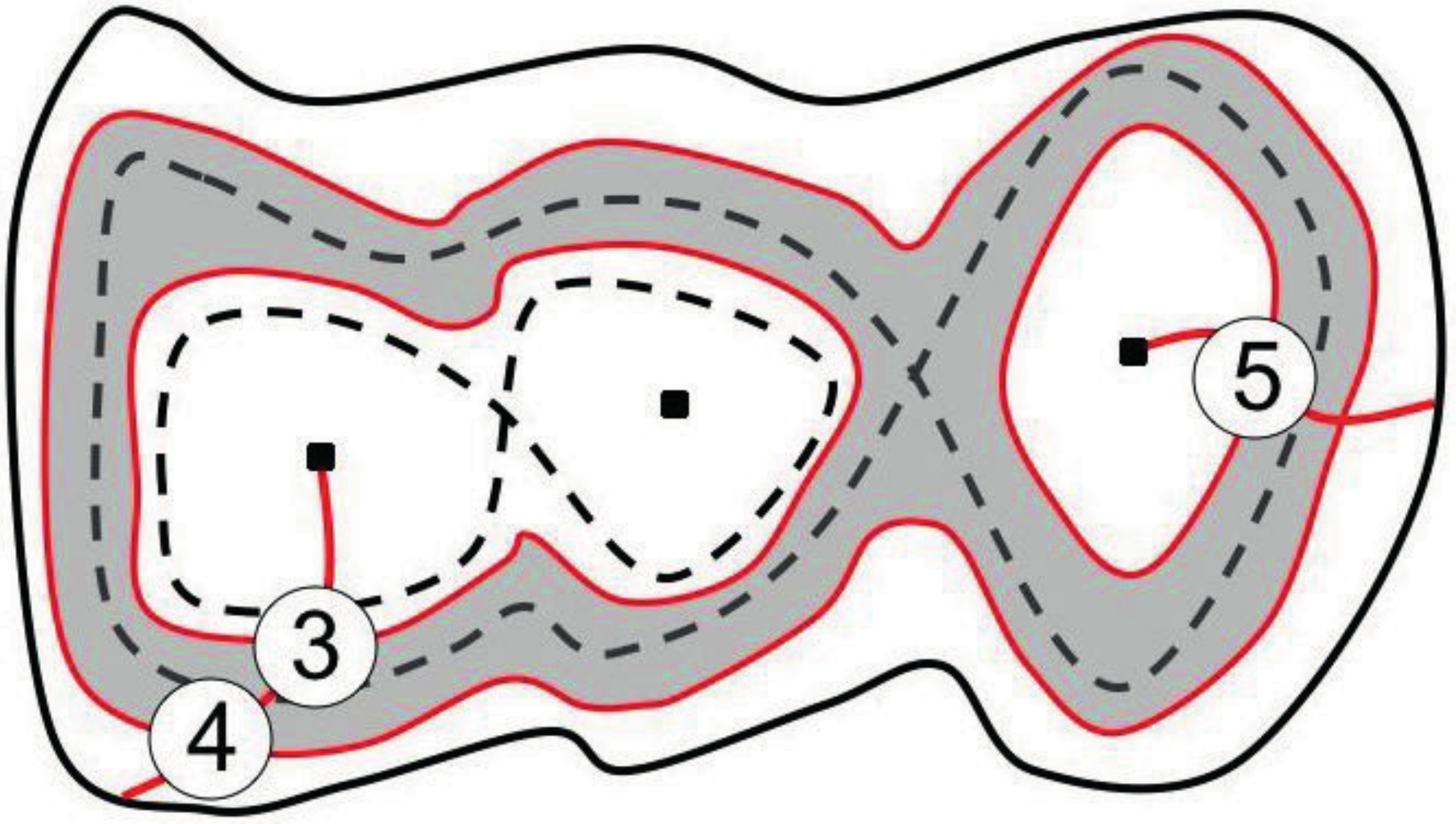}\label{subfig:after3isos}}
   \caption{Example application of the $b^{iso}$ and the $b^{grad}$ motion programs according to the proposed strategy.}
   \label{fig:exampleSearch}
\end{figure}

From a topological point of view, these actions can be described in terms of the graph description of the topology (Fig. \ref{subfig:emptyGraph}). That is, the two mapped gradient lines correspond to paths in this graph that span a connected subgraph, Fig. \ref{subfig:spannedGraph}. In effect the $b^{grad}$ primitives have revealed the existence of two maxima and the  saddled point that is implied by them (Thm. \ref{thm:P-H}). Now assume  that the robot proceeds and maps the isolines also depicted as red curves  in Fig. \ref{subfig:after3isos}, with originating points chosen, according to the proposed protocol, at random locations along the  previously mapped gradient lines. As a result, the robot has confined one of the saddle points within the shaded region of the figure, i.e. it has collected information about the topology induced partition of the function.
\begin{figure}[htbp] 
   \centering
    \subfigure[]{\label{subfig:emptyGraph}\includegraphics[height=100pt]{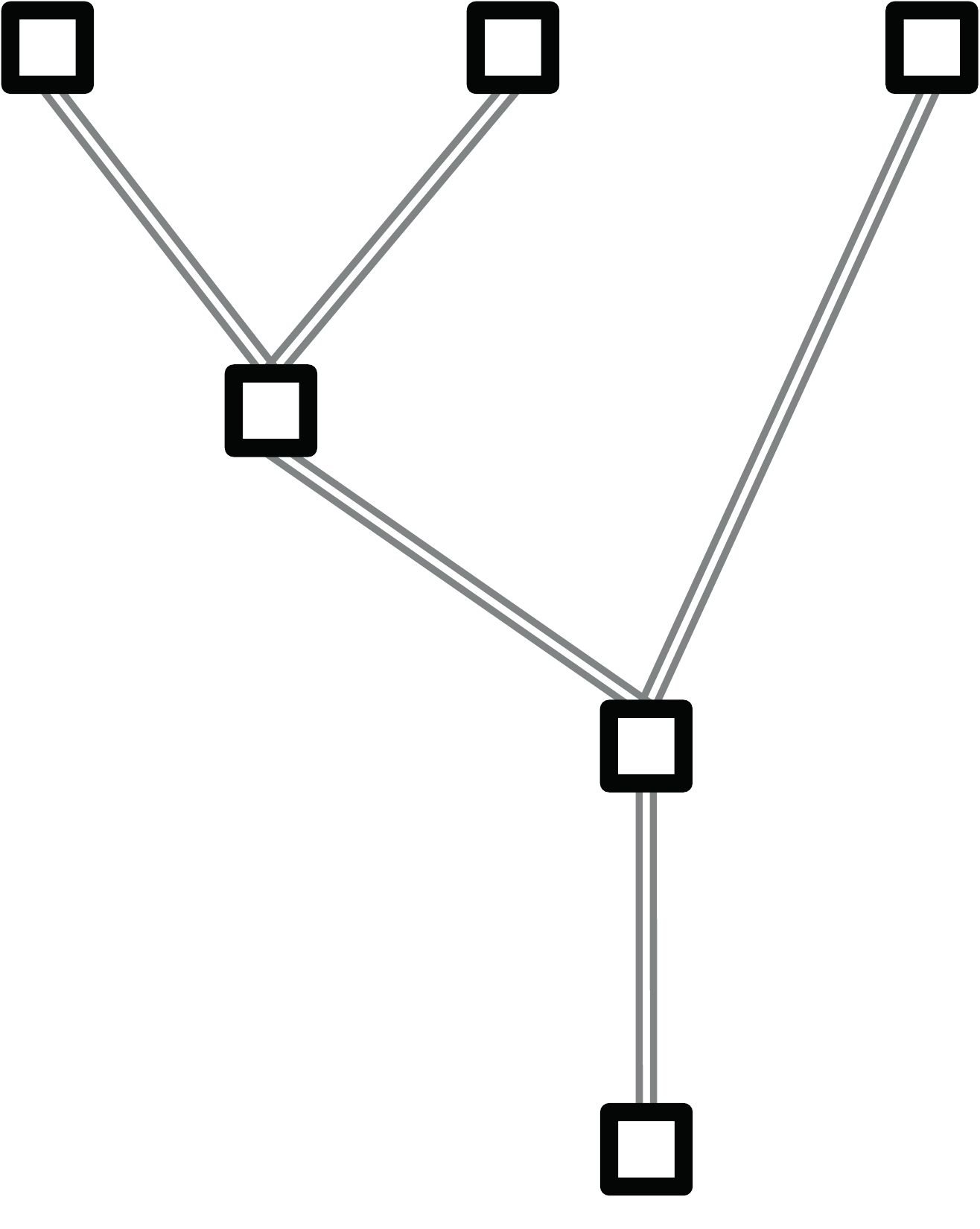}}~~~~~~~~~~
    \subfigure[]{\includegraphics[height=100pt]{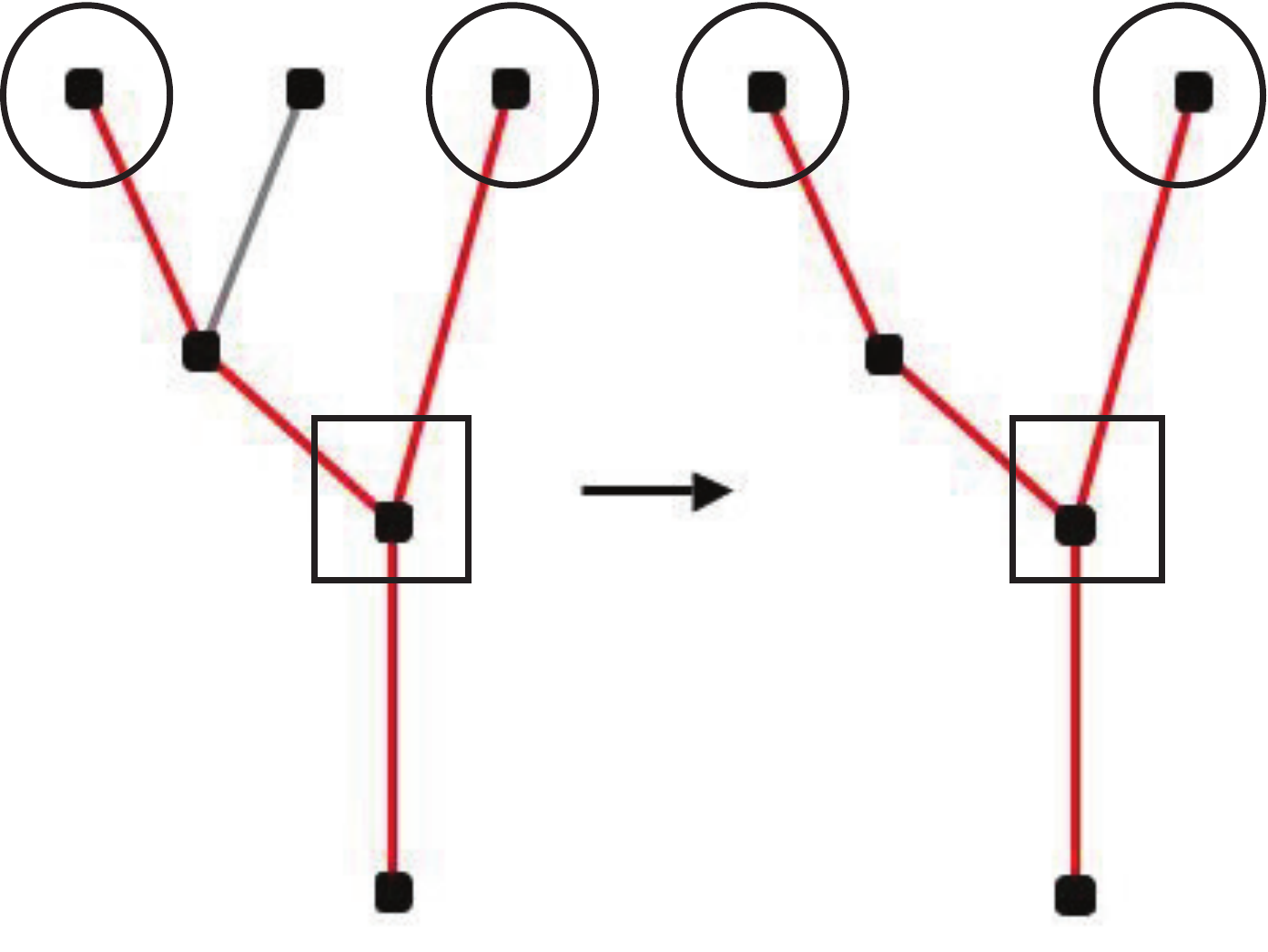}\label{subfig:spannedGraph}}
   \caption{The inference of a part of the topology from the mapped gradient paths. The two valence three nodes in (a) represent unknown index 1 critical level sets. In (b), the circled nodes correspond to mapped local maxima, and the boxed valence three node corresponds to the inferred index 1 critical point (depicted within the shaded region in Fig. 7(c).}
   \label{fig:graphSearch}
\end{figure}
Here, the term information is used loosely, but in the next section it is formalized, in order  to provide a framework within which the reconnaissance can be equated to information acquisition. Moreover, subsequently in this work, we  formalize and generalize the type of Poincar\'e-Hopf based inference that allowed us in the example to conclude the existence of a critical set associated with a saddle point. 

\section{Reconnaissance and information}
\label{sec:information}

In this section, we introduce the information theory of scalar  fields through first defining the concepts of  entropy and conditional entropy of partitions. The main result of this theory is summarized by Theorem \ref{thm:Main}, which establishes that the scalar  field information content is determined by its topology, which, on the other hand, is represented by the topology induced partition. In Section \ref{sec:infoRecon}, we show that pursuing the discovery of this partition through the data acquisition  protocols specified in Section \ref{sec:motionPrimitives} corresponds to iterative information acquisition. Section \ref{sec:infoGuide} builds upon the information formulation of reconnaissance  and the properties of the topology induced partition described in Section \ref{sec:topology} to state rules through which the reconnaissance can be effectively guided based on topology inference procedures and information metrics.

\subsection{Measuring the information capacity of a function or field} 
\label{sec:infoField}

 Information about an unknown field will be measured by the entropy (analogous to Shannon entropy) associated with certain partitions that reflect how important qualitative features of the  field vary over the search domain.  Work along these lines that studied transformations on measure spaces in terms of metrics on partitions of their domains was reported in \cite{Parry}.  Our work is inspired in part by this early effort, as well as by the concepts of information in ergodic theory as reported in \cite{Walters} and the numerous references cited therein and the seminal work on topological entropy by Adler, Konheim, and McAndrew, \cite{RLAdler}.

We begin by discussing the notion of partition entropy, roughly following the development in \cite{Walters}.  Suppose that $X$ is a compact metric space---e.g a compact, connected subset of $\mathbb{R}^n$ having nonempty interior.  Suppose that $X$ comes equipped with a measure $\mu$ together with a $\sigma$-algebra $\Sigma$ of measurable subsets.  Let $\alpha$ and $\beta$ be two at-most-countable partitions of $X$ such that all cells in $\alpha,\beta$ are members of $\Sigma$.  To make the link between partitions and probability and information theory, we think of the partitions $\alpha=\{A_1,\dots\}$ and $\beta=\{B_1\dots\}$ as describing the possible outcomes of experiments, with $\mu(A_j)$ being the probability of outcome $A_j$ in experiment $\alpha$ and $\mu(B_k)$ being the probability of outcome $B_k$ in experiment $\beta$.

To each partition, we associate a measure $H(\alpha)$ (resp.\ $H(\beta)$) that describes the amount of uncertainty about the outcome of the experiment.  We thus define the {\em partition entropies}

\[
\begin{array}{ll}
H(\alpha)=-\sum_{A_i\in\alpha}\mu(A_i)\log_2\mu(A_i),& {\rm and}\\[0.1in]
H(\beta)=-\sum_{B_j\in\beta}\mu(B_j)\log_2\mu(B_j).
\end{array}
\]
(We assume that $\mu(X)=1$.)\ \  In the same spirit, the {\em conditional entropy} of $\alpha$ conditioned on $\beta$ is given by

\begin{eqnarray}
\nonumber H(\alpha | \beta) & = & \sum_{B_j\in\beta} \mu(B_j)H(\alpha|B_j)\\[0.1in]
\nonumber &=&-\sum_{B_j\in\beta} \mu(B_j)\,\sum_{A_i\in\alpha}\frac{\mu(A_i\cap B_j)}{\mu(B_j)}\log_2 \frac{\mu(A_i\cap B_j)}{\mu(B_j)}\\[0.1in]
 &=&-\sum_{B_j\in\beta}\sum_{A_i\in\alpha}\mu(A_i\cap B_j)\log_2 \frac{\mu(A_i\cap B_j)}{\mu(B_j)}.
 \label{eq:infoKL}
\end{eqnarray}

To establish certain properties of the conditional entropy, we state the following definition.
\begin{definition}
Let $\beta$ and $\gamma$ be two partitions of $X$. We say that $\gamma$ is a {\it refinement} of $\beta$ (denoted by writing $\beta\leq\gamma$) if each element of $\beta$ is the finite union of elements of $\gamma$. The partition $\gamma$ is said to be a {\it proper refinement} of $\beta$ if there is an element $B\in\beta$ that is the union of no fewer than two elements of $\gamma$. 
\end{definition}

Then, the following bounds  on the conditional partition entropy (cf.  \cite{Cover:1938fk}) are included for completeness.
\begin{proposition}
\label{prop:KLbounds}
Let $\alpha$ and $\beta$ be two partitions of a given domain  $X$, then the metric   $H(\alpha|\beta)$ defined by \eqref{eq:infoKL} satisfies the inequality
\begin{equation}
\label{eq:infoIneq}
0\leq H(\alpha|\beta) \leq H(\alpha),
\end{equation}
with equality:
\[
H(\alpha|\beta) = 0
\]
 if and only if  $\beta$ is a refinement of $\alpha$ ($\beta\geq\alpha$).  
\end{proposition}

\begin{proof}
Starting with the right side of \eqref{eq:infoIneq}, the log-sum inequality \cite{Cover:1938fk} yields:
\begin{eqnarray*}
-\sum_{B_j\in\beta}{\mu(A_i\cap B_j)}\log_2\frac{\mu(A_i\cap B_j)}{\mu(B_j)} &\leq& -\left(\sum_{B_j\in\beta}{\mu(A_i\cap B_j)}\right) \log_2\frac{\sum_{B_j\in\beta}\mu(A_i\cap B_j)}{\sum_{B_j\in\beta}\mu(B_j)} \\
&=& -{\mu(A_i)}\log_2{\mu(A_i)},
\end{eqnarray*}
and therefore
\[
H(\alpha|\beta) \leq -\sum_{A_i\in\alpha} {\mu(A_i)}\log_2{\mu(A_i)} = H(\alpha).
\]
For the  left side, on the other hand,
\[
\frac{\mu(A_i\cap B_j)}{\mu(B_j)}\leq1,
\]
with equality satisfied if and only if $B_j\subseteq A_i$ (or more precisely $\mu(B_j\cap \bar{A}_i)=0$), and therefore,
\[
\log_2\frac{\mu(A_i\cap B_j)}{\mu(B_j)} \leq 0, 
\]
which when substituted back in \eqref{eq:infoKL}  concludes the proof. 
\end{proof}

\begin{proposition}
\label{prop:prop2}
Let $\alpha = \{A_1,\dots\}$, $\beta = \{B_1,\dots\}$, $\gamma = \{C_1,\dots\}$ be partitions of $X$ such that $\beta$ is a refinement of $\alpha$. Than $H(\alpha|\gamma)\leq H(\beta|\gamma)$. The inequality is strict if $\beta$ is a proper refinement of $\alpha$, provided that $\gamma$ is not  a refinement of $\beta$. 
\end{proposition}
\begin{proof}
Let $A\in\alpha$ be written as
\[
A = \bigcup_{j=1}^{n_A} B_j.
\]
We claim that
\begin{equation}
\label{eq:refineIneq}
-\mu(A\cap C) \log_2 \frac{\mu(A\cap C)}{\mu(C)} \leq -\sum_{j=1}^{n_A} \mu(B_j\cap C) \log_2 \frac{\mu(B_j\cap C)}{\mu(C)}
\end{equation}
with the inequality being strict if $n_A\geq 2$. The result is easily established for $n_A=2$ by writing $p_1 = \mu(B_1\cap C)/\mu(A\cap C)$, $p_2 = \mu(B_2\cap C)/\mu(A\cap C)$ and $q = p_1+p_2$. Then we have
\[
0\leq -\frac{p_1}{q}\log_2\frac{p_1}{q}-\frac{p_2}{q}\log_2\frac{p_2}{q}
\]
with the inequality being strict unless one of $p_1$ and $p_2$ is zero. This inequality is easily shown to be equivalent to \eqref{eq:refineIneq} for $n_A=2$,  where the inequality is strict unless $\mu(B_j\cap C)=0$. 

The proof of the case $n_A>2$ uses a simple inductive argument along the same lines. Here, for each $C\in X$ we have
\[
-\sum_{A_i\in\alpha }\mu(A_i\cap C) \log_2 \frac{\mu(A_i\cap C)}{\mu(C)} \leq -\sum_{j=1}^{n_A} \mu(B_j\cap C) \log_2 \frac{\mu(B_j\cap C)}{\mu(C)}
\]
with the inequality being strict if $\beta$ is a proper refinement of $\alpha$. The proposition follows by summing both sides over $C\in\gamma$. 
\end{proof}

\begin{corollary}
Let $\alpha = \{A_1,\dots\}$ and $\beta = \{B_1,\dots\}$ be partitions of $X$ such that $\beta$ is a refinement of $\alpha$. Then $H(\alpha)\leq H(\beta)$. The inequality is strict if $\beta$ is a proper refinement of $\alpha$. 
\end{corollary}
\begin{proof}
This is a special case of Proposition \ref{prop:prop2} in which $\gamma$ is the trivial partition $\gamma=\{X\}$. 
\end{proof}

\begin{proposition}
\label{prop:refinement1}
Let $\alpha = \{A_1,\dots\}$, $\beta = \{B_1,\dots\}$, $\gamma = \{C_1,\dots\}$ be partitions of $X$ such that $\gamma$ is a refinement of $\beta$. Then $H(\alpha|\beta)\geq H(\alpha|\gamma)$. 
\end{proposition}
\begin{proof}
For each $B_j\in\beta$, write $B_j = \bigcup_{i=1}^{n_j} C_i$ and note that
\[
\mu(B_j) = \sum_{i=1}^{n_j} \mu (C_i),~\text{and}
\]
\[
\mu(A_k\cap B_j) = \sum_{i=1}^{n_j} \mu (A_k\cap C_i), \forall A_k \in \alpha. 
\]
The log-sum inequality \cite{Cover:1938fk} may be rendered
\begin{eqnarray*}
\sum_{i=1}^{n_j} \mu (A_k\cap C_i) \log_2 \frac{\mu (A_k\cap C_i)}{\mu (C_i)} & \geq& \left (\sum_{i=1}^{n_j} \mu (A_k\cap C_i) \right)\log_2 \frac{\sum_{i=1}^{n_j}\mu (A_k\cap C_i)}{\sum_{i=1}^{n_j}\mu (C_i)} \\
&=& \mu (A_k\cap B_j) \log_2 \frac{\mu (A_k\cap B_j)}{\mu (B_j)}.
\end{eqnarray*}
Multiplying both sides of the inequality by $-1$ and summing over the cells in the partition yields the desired result. 
\end{proof}

\begin{remark}
Proposition \ref{prop:refinement1} shows that conditional entropy, $H(\alpha|\beta)$, is a non-increasing function of $\beta$ under partition refinement. Even when $\gamma$ is a proper refinement of $\beta$, however, it is not necessarily the case that $H(\alpha|\beta)>H(\alpha | \gamma)$. A simple example that illustrates this is given by taking $\alpha=\beta = \{X\}$ and letting $\gamma =\{C_1,C_2\}$ where $\mu(C_1)=\mu(C_2)$.  
\end{remark}

We seek to transfer these notions of entropy to quantify the entropy of continuous scalar functions defined on compact domains.  Let $X\subset\mathbb{R}^m$ be a compact, connected, simply connected domain, and let $f:\mathbb{R}^m\to \mathbb{R}$. Then $f(X)$ is a compact connected subset of $\mathbb{R}$ which we write as $[a,b]$.  At the outset, we fix a finite partition, $\Pi_m$, of this interval:
\begin{equation}
\label{eq:rangePartition}
a=x_0<x_1<\cdots<x_m=b.
\end{equation}
For each $x_j,\ j=1,\dots,m$, we denote the set of connected components of $f^{-1}([x_{j-1},x_j])$ by
\[
\text{cc}\,(f^{-1}([x_{j-1},x_j])\,).
\]
For any such partition, we obtain a corresponding {\it domain partition} 
\begin{equation}
\label{eq:equivDomainPart}
{\cal V}_{\Pi_m}=\bigcup_{j=1}^m\left\{\,\text{cc}\left(f^{-1}\left(\left[x_{j-1},x_j\right]\right)\right)\right\}
\end{equation}
of $X$.  We define the {\em entropy} of $f$ with respect to ${\cal V}_{\Pi_m}=\{V^1,\dots,V^N\}$ (or equivalently with respect to $\Pi_m$) as
\begin{equation}
H\left(f,{\cal V}_{\Pi_m}\right)= -\sum_{j=1}^N\frac{\mu(V^j)}{\mu(X)}\log_2 \frac{\mu(V^j)}{\mu(X)},
\label{eq:jb:complexity}
\end{equation}
where $\mu$ is Lebesgue measure on $\mathbb{R}^m$.  We shall also refer to (\ref{eq:jb:complexity}) as the {\em partition entropy} of $f$ with respect to ${\cal V}_{\Pi_m}$ (with respect to $\Pi_m$).

\begin{figure}[htbp] 
   \centering
   \includegraphics[width=400pt]{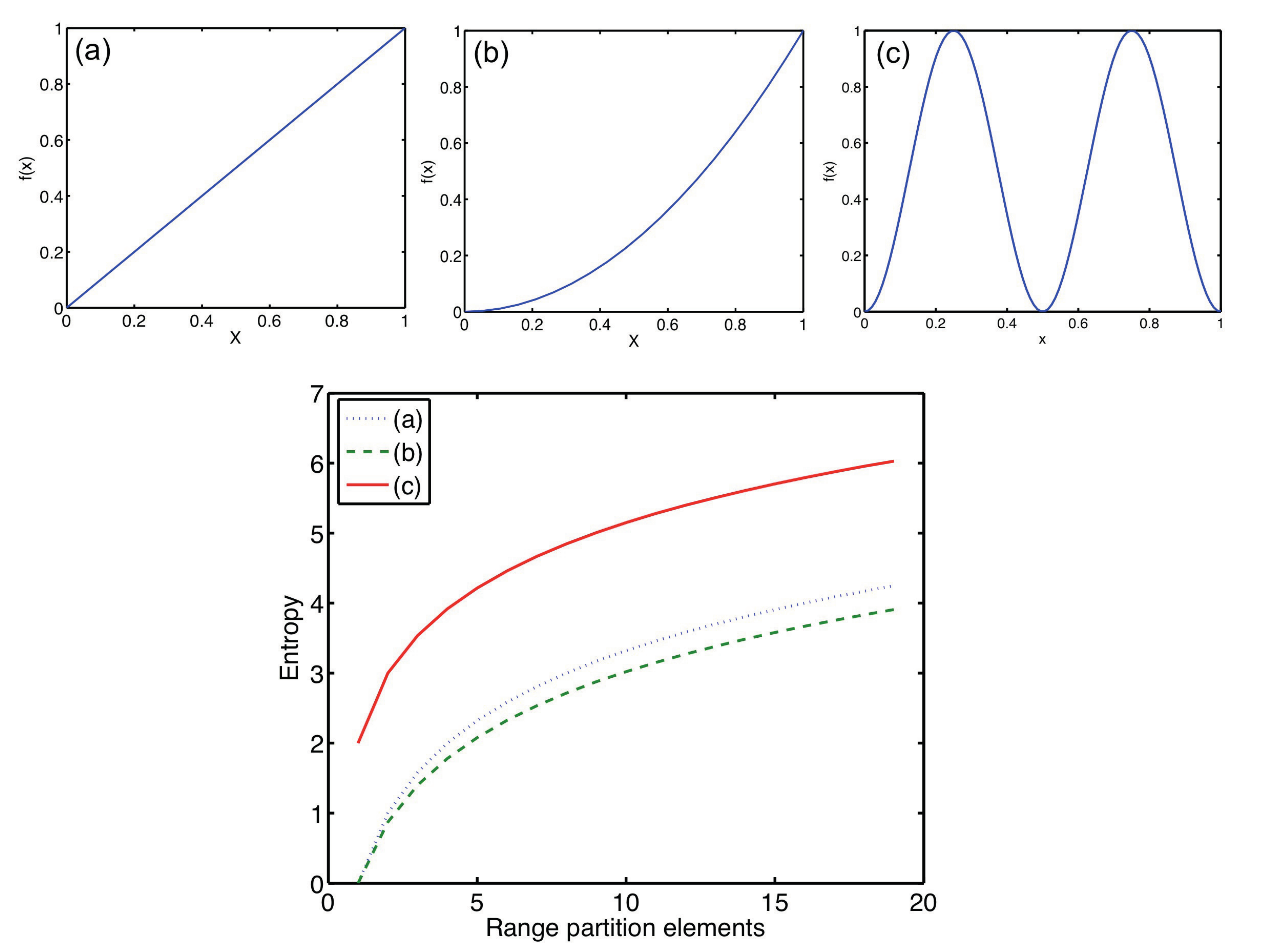} 
   \caption{From left to right, the functions mapping the unit interval onto itself are given by (a) $f_1(x)=x$, (b) $f_2(x)=x^{2}$, and (c) $f_3(x)=\sin^2{2\pi x}$. The bottom graph corresponds to the entropy, $H\left(f,\mathcal{V}_{\Pi_m}\right)$, for the three functions for uniform range partitions of different sizes, $m$}
   \label{fig:compEntPart}
\end{figure}

It is easy to write down expressions for the entropy of simple functions as the one illustrated in  Fig.\ \ref{fig:compEntPart}.  In all cases, the range is the interval $[0,1]$, and we partition this range into $m$ equal subintervals of length $1/m$.  We evaluate the entropy (\ref{eq:jb:complexity}) for $m$ between $1$ and $20$. The entropy grows in all cases as the number $m$ of cells in the range partition increases (See the bottom graph of Fig. \ref{fig:compEntPart}).  However, it grows with different rates, with function (a) growing with $\log_2 m$, (b) growing slower since it contracts lengths near $0$ and expands them near $1$, and (c) growing with the highest rate since its topology partitions the domain into finer sets.   We shall thus be less interested in the specific value of the entropy than in the relative sizes for different classes of functions---and in how these scale as $m$ becomes large. At this point, we can quantify the relationship between the topology induced partition and the entropy of the function.

\begin{theorem}
\label{thm:Main}
Let $\mathcal{V}_{m}$ be a domain partition of $X$  corresponding to a range division of $f:X\to\real$ into $m$ equal components. Let $\mathcal{M}$ be the topology induced partition of the same function and define $\delta_i$ for each cell $M_i\in\mathcal{M}$ as:
\[
\delta_i = \frac{\sup_{\mathbf{r}\in M_i}  f(\mathbf{r}) - \inf_{\mathbf{r}\in M_i}  f(\mathbf{r})}{\sup_{\mathbf{r}\in X}  f(\mathbf{r}) - \inf_{\mathbf{r}\in X}  f(\mathbf{r})}.
\]
Then the following holds:
\begin{equation}
\label{eq:increaseEntLimit}
\lim _ {m\to\infty}\left(H(f,\mathcal{V}_{m})-\log_2 m\right) \leq H(\mathcal{M}) + \sum_{i=1}^n \frac{\mu(M_i)}{\mu(X)} \log_2 \delta_i,
\end{equation}
where $n=|\mathcal{M}|$ is the cardinality of the (finite) topology induced partition. 
\end{theorem} 
\begin{proof}
For each $M_i\in\mathcal{M}$, let
\[
\mathcal{V}^i = \{V^j\in \mathcal{V}_{m}:V^j\subset M_i\},
\]
and let 
\[
\mathcal{V}^* = \{V^j\in \mathcal{V}_{m}:V^j\text{~is~not~confined~in~any~}M_i\}.
\]
Then, we may write
\begin{equation}
\label{eq:twoSeparateEnts}
H(f,\mathcal{V}_{m}) = -\sum_{M_i\in\mathcal{M}}\sum_{V^{ij}\in\mathcal{V}^i} \frac{\mu({V}^{ij})}{\mu(X)} \log_2 \frac{\mu({V}^{ij})}{\mu(X)} -\sum_{V^{k}\in\mathcal{V}^*} \frac{\mu({V}^{k})}{\mu(X)} \log_2 \frac{\mu({V}^{k})}{\mu(X)}. 
\end{equation}

Given any $\epsilon>0$, we can find an $m_{\epsilon}$ sufficiently large that for all $m>m_{\epsilon}$
\begin{equation}
\label{eq:complexEpsilon}
-\sum_{M_i\in\mathcal{M}}\sum_{V^{ij}\in\mathcal{V}^i} \frac{\mu({V}^{ij})}{\mu(X)} \log_2 \frac{\mu({V}^{ij})}{\mu(X)}<-\sum_{M_{i}\in\mathcal{M}} \frac{\mu({M}_{i})}{\mu(X)} \log_2 \frac{\mu({M}_{i})}{\mu(X)\lceil \delta_i m \rceil} + \frac{\epsilon}{2}
\end{equation}
and
\begin{equation}
\label{eq:simpleEpsilon}
-\sum_{V^{k}\in\mathcal{V}^*} \frac{\mu({V}^{k})}{\mu(X)} \log_2 \frac{\mu({V}^{k})}{\mu(X)}  < \frac{\epsilon}{2},
\end{equation}
where $\lceil \cdot \rceil$ denotes the ceiling function.
The right hand side of \eqref{eq:complexEpsilon} is:
\[
H(\mathcal{M}) +  \sum_{i=1}^n \frac{\mu(M_i)}{\mu(X)} \log_2 \lceil \delta_i m \rceil +  \frac{\epsilon}{2} = H(\mathcal{M}) +  \sum_{i=1}^n \frac{\mu(M_i)}{\mu(X)} \log_2\frac{ \lceil \delta_i m \rceil}{m} +\log_2m+  \frac{\epsilon}{2}.
\] 
Then combining these relationships back in \eqref{eq:twoSeparateEnts} yields
\[
H(f,\mathcal{V}_{m})-\log_2 m < H(\mathcal{M}) +  \sum_{i=1}^n \frac{\mu(M_i)}{\mu(X)} \log_2\frac{ \lceil \delta_i m \rceil}{m} +  {\epsilon}
\]
As $m\to\infty$, $\frac{ \lceil \delta_i m \rceil}{m}\to\delta_i$, and since $\epsilon$ can be chosen arbitrarily, the result follows. 
\end{proof}

\begin{remark}
Under the assumption that $X$ is compact and $f:X\to\real$ is Morse function, the right-hand-side of \eqref{eq:increaseEntLimit} is well defined and the limit on the left-hand-side is always finite. 
\end{remark}

\begin{remark}
\label{remark:possibleNegative}
There exist fields for which either one or both sides of the inequality \eqref{eq:increaseEntLimit} are negative. 
\end{remark}

\begin{remark}
It is not difficult to prove that if $\Pi_m$ is any partition \eqref{eq:rangePartition} of the range of the function $f$ having corresponding domain partition $\mathcal{V}_{\Pi_m}$ \eqref{eq:equivDomainPart}, then
\[
\lim_{m\to\infty} \left(H(f,\mathcal{V}_{\Pi_m})-\log_2 m\right) = \lim_{m\to\infty} \left(H(f,\mathcal{V}_{m})-\log_2 m\right),
\]
provided 
\[
\max_{1\leq j \leq m} |x_j - x_{j-1}|\to0.
\]
\end{remark}

The last remark motivates the following.

\begin{definition}
Given a compact set $X$ and a Morse function $f:X\to\real$, and given a set of partitions of the range of $f$ into $m$ equal subintervals, the {\it entropy} of $f$ is defined by:
\[
\lim_{m\to\infty} \left(H(f,\mathcal{V}_m)-\log_2m\right). 
\] 
\end{definition}

\begin{remark}
For any partition $\Pi_m$  \eqref{eq:rangePartition} of the range of $f$ having corresponding domain partition $\mathcal{V}_{\Pi_m}$ \eqref{eq:equivDomainPart}, we can think of the {\it conditional entropy} of the domain partition conditioned on the range partition being given by the expression 
\begin{equation}
\label{eq:rangeDomainCE}
-\sum_{j=1}^m\sum_{V^{i}\in\mathcal{V}_{\Pi_m}} \mu\left(V^i\cap f^{-1}([x_{j-1},x_j])\right)\log_2\frac{\mu\left(V^i\cap f^{-1}([x_{j-1},x_j])\right)}{x_j-x_{j-1}}.
\end{equation}
If $\Pi_m$ is the uniform partition with $x_j-x_{j-1} = \frac{b-a}{m}$ for $j = 1,2,\dots,m$, then, up to an additive constant that is independent of $m$, this conditional entropy coincides with
\[
H(f,\mathcal{V}_{m})-\log_2 m.
\]
\end{remark}
Strictly speaking of course, \eqref{eq:rangeDomainCE} is not an actual conditional entropy since $\mathcal{V}_{\Pi_m}$ and $\Pi_m$ are partitions of different spaces that do not have a common measure. Moreover, as indicated by Remark \ref{remark:possibleNegative}, the expression  \eqref{eq:rangeDomainCE}  may take negative values, which means that it cannot be viewed as a conditional entropy in the usual sense. Nevertheless, it is useful to think of the quantity \eqref{eq:rangeDomainCE}  as the limiting value under partition refinement of the conditional entropies of domain partitions corresponding to increasingly fine sequences of range partitions. (Another way to view \eqref{eq:rangeDomainCE} is as the capacity of the function $f$ to act as an information channel between the range, $[a,b]$, and the domain, $X$.)

Going back to Sec. II, the bound of the entropy of a particular function, as its range partition is refined, depends completely on its critical level sets and their scalar value. Therefore, the reconnaissance strategies, that we will analyze, will be focused on the discovery of the topology induced partition. 

\subsection{Reconnaissance as information acquisition}
\label{sec:infoRecon}
We note that the reconnaissance cannot directly discover the boundaries of the cells in the topology induced partition that are associated with saddle points. For this to happen, the robot should choose an originating point, $\mathbf{r}_o$, for an isoline mapping primitive that lies on a particular critical level set $\xi^*\in$ Cr$^1(f,X)$, Cr$^1(f,X)$ being the set of critical level sets associated with saddle points. Since these are zero-measure sets, their discovery in the course of random mapping of isolines can occur with zero probability.  Therefore, we pursue   a procedure by which  the set  of mapped level sets, $S(B_k)$, can be used to make inferences about the topology induced partition and more specifically about Cr$^1(f,X)$. To achieve this, by analogy with the domain partition associated with a given range partition, we describe the collected data as a partition induced by the mapped level sets (contours and extrema), $\mathcal{V}\left(S\left(B_k\right)\right) :=\mathcal{V}_k$ given by
\begin{equation}
\label{eq:dataInducedP}
\mathcal{V}_k:= \text{cc}\left(X\setminus S\left(B_k\right)\right).
\end{equation}
The elements of this {\it data induced partition} will be connected components that do not contain mapped level sets, and we will denote them by $\mathcal{V}_k = \left\{V_k^1,V_k^2,\cdots,V_k^N\right\}$. 

The evolution of $\mathcal{V}_k$ under the search process corresponds to  iterative partition refinement. That is, given the mapping of a level set $\xi$ in the set $V_{k-1}^j\in\mathcal{V}_{k-1}$,  the updated data induced partition will be given by:
\begin{equation}
\label{eq:dataPartition}
\mathcal{V}_{k} = \left(\mathcal{V}_{k-1}\setminus \left\{V_{k-1}^j\right\}\right) \cup \text{cc}\left(V_{k-1}^j\setminus \xi\right),
\end{equation}
where again cc$(\cdot)$ signifies the set of connected components. (Fig. \ref{fig:updateExample} shows an example update of the data induced partition under the mapping of a particular isoline.) We also remind that $S(B_k)$ will contain the mapped extremum points (cf \eqref{eq:mappedIsoDefined}), which are degenerate level contours. Therefore, the data induced partition will be updated also by the mapping of gradient lines. That is, the discovered extremum points are subtracted from $X$, which does not create new cells in $\mathcal{V}_k$, but instead alters cells' topology. 

\begin{figure}[htbp] 
   \centering
    \subfigure[The data induced partition before the update $\mathcal{V}_1=\{V_1^1,V_2^1\}$]{\includegraphics[width=140pt]{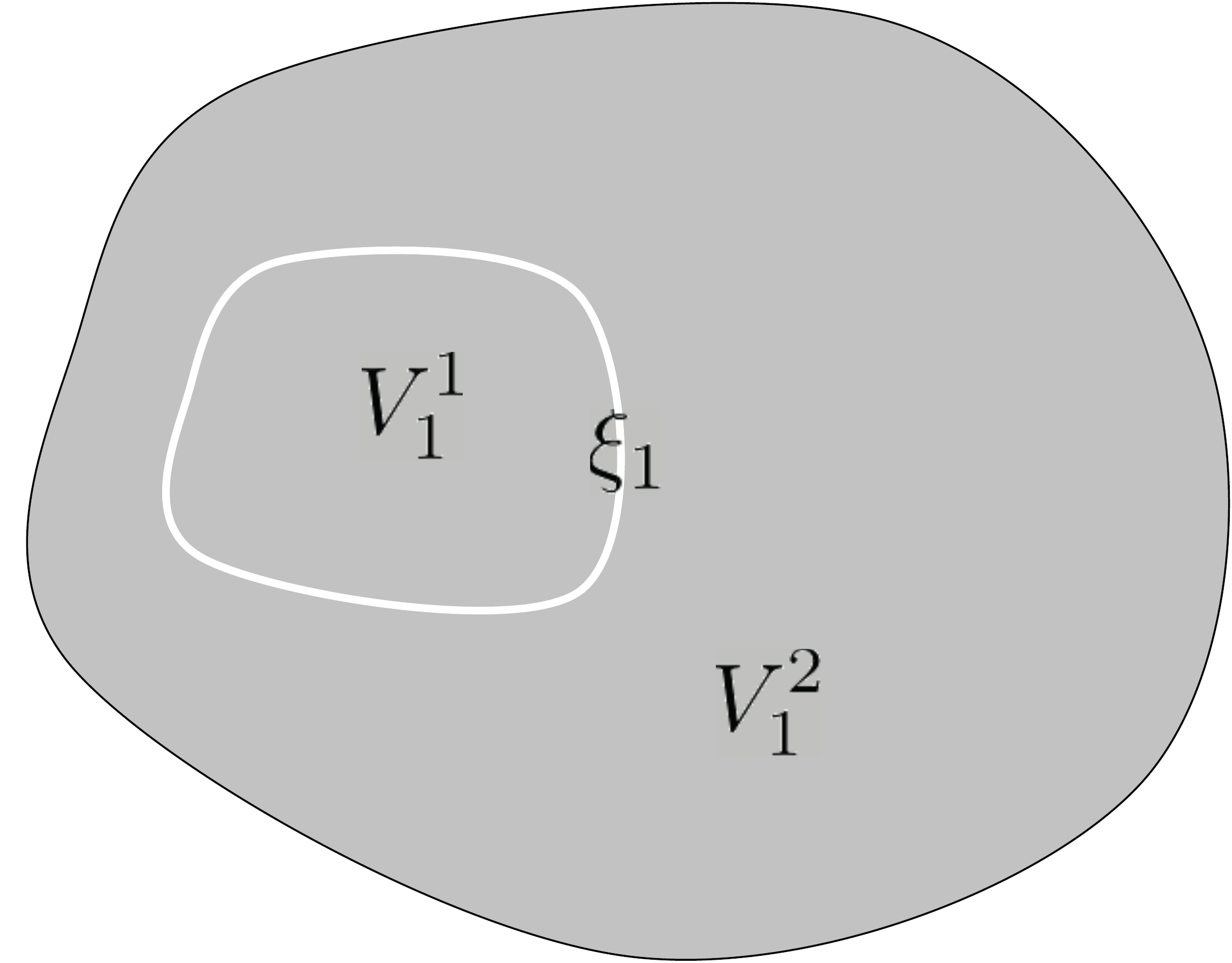}}~~~~~~~~
    \subfigure[and after the mapping of $\xi_2$, yielding  $V_2^1=V_1^1$, and $\{V_2^2,V_2^3\} =$cc$(V_1^2\setminus \xi_2)$  ]{\includegraphics[width=140pt]{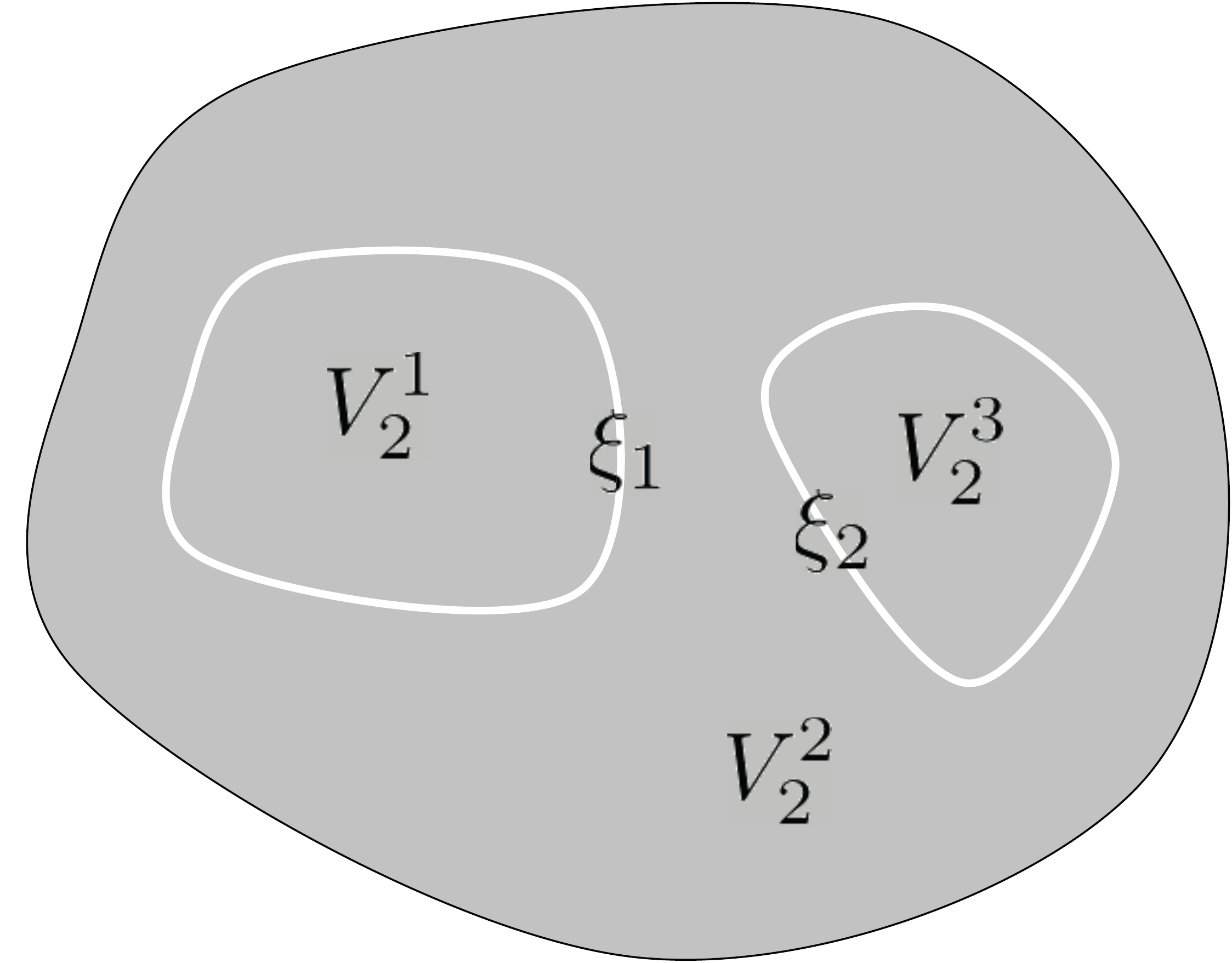}}
   \caption{The data induced partition before and after updating to account for the mapping of an isoline $\xi_2$. The updating consists of relabeling $V_{1}^{1}$ to become $V_{2}^{1}$ and replacing $V_{1}^{2}$ with $\{V_{2}^{2},V_{2}^{3}\}$.} 
   \label{fig:updateExample}
\end{figure}

The problem of quantifying the information content of $S(B_k)$ becomes equivalent to quantifying the relationship between the {\it data induced partition} and the {\it topology induced partition}.  For this purpose, we will consider the conditional entropy metric \eqref{eq:infoKL} and define the {\it conditional entropy}, $H\left(\mathcal{M}|\mathcal{V}_k\right) $, as the following function of the {\it data induced partition} and the {\it topology induced partition}. 
\begin{equation} 
\label{eq:infoKLPart}
H\left(\mathcal{M}|\mathcal{V}_k\right) = -\sum_{M_i\in\mathcal{M}}\sum_{V^j_k\in\mathcal{V}_k}\frac{\mu\left(M_i\cap V^j_k\right)}{\mu(X)}\log_2\frac{\mu\left(M_i\cap V^j_k\right)}{\mu(V^j_k)}. 
\end{equation}    

An important property of how the conditional entropy evolves under a motion primitive based reconnaissance is stated in the next proposition.

\begin{proposition}
\label{prop:decreasing}
Let $\mathcal{M}$ and $\mathcal{V}_k$ be respectively the {\it topology induced partition} and a {\it data induced partition}, and let $\mathcal{V}_k$ evolve according to \eqref{eq:dataPartition}. Then  the {\it conditional entropy}, $H(\mathcal{M}|\mathcal{V}_k)$, is a non-increasing function of $k$. 
\end{proposition}

\begin{proof}
This result follows directly from Proposition \ref{prop:refinement1}, because $\mathcal{V}_k$ is a refinement of  $\mathcal{V}_{k-1}$. 
\end{proof}

By employing  Proposition  \ref{prop:KLbounds} and Proposition \ref{prop:decreasing}, the {\it conditional entropy}  allows us to describe  the evolution of the reconnaissance  as iterative information acquisition. Initialization with no prior information is equivalent to  $\mathcal{V}_0 = \{X\}$, and also  $H(\mathcal{M}|\mathcal{V}_0)=H(\mathcal{M})$.  (Note, that $H(\mathcal{M})$ is unknown at the outset.) Under the reconnaissance, as more information is gathered, $H(\mathcal{M}|\mathcal{V}_k)$ decreases, and assuming that  $H(\mathcal{M}|\mathcal{V}_k)$ approaches zero in the limit as $k\to\infty$, the topology induced partition can be fully recovered from the collected data by merging cells of $\mathcal{V}_k$ as $k$ increases according to \eqref{eq:dataInducedP}.

\subsection{ Guiding reconnaissance through the conditional entropy}

\label{sec:infoGuide}

Since the reconnaissance can be represented as an iterative information acquisition process, an efficient reconnaissance protocol should at each step utilize the acquired information to decide  what motion program to apply next and where in the search domain to apply it. Therefore, our aim is to exploit the structure of the data induced partition as a means to guide the reconnaissance such that it collects relevant topological information. Collecting relevant information in terms of the conditional entropy will correspond to a refinement of $\mathcal{V}_k$ through isoline mapping which yields a strict decrease in the conditional entropy:
\[
H(\mathcal{M}|\mathcal{V}_k)>H(\mathcal{M}|\mathcal{V}_{k+1}).
\] 

As a first step to achieve this objective, the next theorem gives a connection between the structures of the data induced partition and the topology induced partition.

\begin{theorem}
\label{th:cardIneq}
Let $\mathcal{V}_k$ be the data induced partition of the domain, $X$, under the search sequence $B_k$ applied to the scalar  field $f:X\to\real$. Define for every $V_k^j\in\mathcal{V}_k$ the set 
\begin{equation}
\label{eq:subSet}
\mathcal{M}'\left(V_k^j\right) = \left\{M_i\in\mathcal{M}\left(f,X\right): M_i\cap V_k^j \neq\emptyset\right\}.
\end{equation}
Then, it follows that
\begin{equation}
\label{eq:cardIneq}
\left|\mathcal{M}'\left(V_k^j\right)\right| \geq |-2\chi( V_k^j) + 1|.
\end{equation}
\end{theorem}
\begin{proof}
We start the proof by noting that $S(B_k)$ contains the discovered extremum points, and therefore, according to \eqref{eq:dataInducedP}, there are no mapped extremum points in the cells of the data induced partition. Assume that  $V_k^j\cap \xi^*_i = \emptyset,~\forall \xi^*_i \in \text{Cr}^{0,2}(f,X)$, i.e. all extremum points are known and subtracted from $X$. Then, by  applying the same argument as in the proof of Theorem \ref{th:TIPcardinality}, we can show that:
\[
\left|\mathcal{M}'\left(V_k^j\right)\right| = |-2\chi( V_k^j) + 1|. 
\]
If, however, this assumption is not valid, i.e. there are extremum points that lie within the given cell (i.e. extrema that have not been discovered), the cardinality of $|\mathcal{M}'(V_k^j)|$ will increase. 
\end{proof}

Having information about the number and locations of extremum points allows many conclusions to be drawn regarding the topology induced partition. In what follows, we describe corollaries of Theorem \ref{th:cardIneq} that establish properties of $\mathcal{M}$ under the assumption that all extremum points have been discovered. This assumption can be expressed in terms of $\text{Cr}^{0,2}(f,X)$ and $S(B_k)$  as:
\begin{equation}
\label{eq:allDiscovered}
\text{Cr}^{0,2}(f,X) \subset S(B_k),
\end{equation}
since $S(B_k)$, in addition to the mapped level contours, contains the mapped extremum points. (See the definition in \eqref{eq:mappedIsoDefined}.) Moreover, taking into account \eqref{eq:dataInducedP}, this assumption will imply that all extremum points are reflected in the topology of the cells in the data induced partition. Specifically, given a mapped isoline contour $\xi$, and given that there are exactly $K$ extremum points in its interior, the sum of the Euler characteristics of all cells of $\mathcal{V}_k$ that also lie in this contour's interior is given by $1-K$.

\begin{corollary}
\label{thm:uncertainty}
For each $V_k^j\in \mathcal{V}_k$ consider the Euler characteristic $\chi(V^j_k)$.  Assume that all extremum points have been discovered \eqref{eq:allDiscovered}, and define the set 
\begin{equation}
\label{eq:eulerSet}
\mathcal{V}'_k = \left\{V_k^j\in\mathcal{V}_k: \chi\left(V_k^j\right)\leq -1\right\}.
\end{equation}
Then the following holds:
\begin{equation}
H(\mathcal{M}|\mathcal{V}_k) =- \sum_{V^j_k \in \mathcal{V}'_k} \sum_{M_i \in \mathcal{M}} \frac{\mu\left(M_i\cap V^j_k\right)}{\mu(X)}\log_2\frac{\mu\left(M_i\cap V^j_k\right)}{\mu\left(V^j_k\right)}.\end{equation}
\end{corollary}
\begin{proof}
If the reconnaissance has discovered all extremum points in $X$ and they are included in $S(B_k)$, the partition $\mathcal{V}_k$ will not contain cells with Euler characteristic $1$. Moreover, by Theorem \ref{th:cardIneq}, if for a particular cell $V^j_k\in\mathcal{V}_k$, $\chi(V^j_k) = 0$, then $\left|\mathcal{M}'\left(V_k^j\right)\right| = 1$. This implies that there exist $M_i\in\mathcal{M}(f,X)$, s.t. $V^j_k \subseteq M_i$, which  implies that
\[
\frac{\mu(M_j\cap V^i_k)}{\mu(X)}\log_2\frac{\mu(M_j\cap V^i_k)}{\mu\left(V^i_k\right)} = 0.
\]
Therefore, only elements of $\mathcal{V}_k$ with  Euler characteristic less than zero will have a contribution to the conditional entropy. 
\end{proof}

\begin{corollary}
\label{cor:simpleRule}
Let $\mathcal{V}_{k-1}$ be a data induced partition with the set $\mathcal{V}'_{k-1}$  defined by \eqref{eq:eulerSet}. Assume that the isoline $\xi$ is mapped such that   $\xi\in V_{k-1}^j \in \mathcal{V}'_{k-1}$, and it induces the update
\[
\mathcal{V}_k=\left(\mathcal{V}_{k-1}\setminus \{V^j_{k-1}\}\right)\cup\left\{V_k^a,V_k^b\right\}.
\]
Then assuming that both $\mu(V_k^a),\mu(V_k^b)>0$, it follows that
\[
H(\mathcal{M}|\mathcal{V}_{k-1})>H(\mathcal{M}|\mathcal{V}_{k}).
\]
\end{corollary}
\begin{proof}
Note that
\begin{eqnarray*}
H(\mathcal{M}|\mathcal{V}_{k-1})-H(\mathcal{M}|\mathcal{V}_{k}) &=& -\sum_{M_i\in\mathcal{M}} \frac{\mu(M_i\cap V^j_{k-1})}{\mu(X)}\log_2\frac{\mu(M_i\cap V^j_{k-1})}{\mu(V^j_{k-1})}\\
&+& \sum_{M_i\in\mathcal{M}} \sum_{l\in\{a,b\}}\frac{\mu(M_i\cap V^l_{k})}{\mu(X)}\log_2\frac{\mu(M_i\cap V^l_{k})}{\mu(V^l_{k})}.
\end{eqnarray*}
From the log-sum inequality this expression is equal to zero if and only if:
\begin{equation}
\label{eq:logSumEq}
\frac{\mu(M_i\cap V^a_{k})}{\mu(V^a_{k})} = \frac{\mu(M_i\cap V^b_{k})}{\mu(V^b_{k})}, \forall M_i\in\mathcal{M},
\end{equation}
and otherwise is strictly larger than zero. 

Note that i) according to Theorem \ref{th:cardIneq} $|\mathcal{M}'(V_{k-1}^j)|>1$; and ii) every isoline lies within a single cell of the topology induced partition. 
Therefore,  there is at least one cell $M_i$ such that $M_i\cap V^b_{k}=\emptyset$ and $M_i\cap V^a_{k}\neq\emptyset$, or equivalently $M_i\cap V^a_{k}=\emptyset$ and $M_i\cap V^b_{k}\neq\emptyset$. This contradicts \eqref{eq:logSumEq}, which implies that:
\[
H(\mathcal{M}|\mathcal{V}_{k-1})>H(\mathcal{M}|\mathcal{V}_{k}).
\]
\end{proof}

Assuming that  all extremum points are known, the information about the topology induced partition is localized in the cells of the  {\it data induced partition} that have   Euler characteristic $\chi(V^j_k)<0$. Therefore, a parsimonious reconnaissance procedure that does not acquire redundant information is to concentrate the isoline mapping in these cells, which, according to Corollary \ref{cor:simpleRule}, yields strictly decreasing conditional entropy.   In this way, given that all extremum points have been discovered, i.e. all areas that contribute to the conditional entropy have been identified, this procedure will continuously yield information about the topology induced partition of the function. However, if this is not the case, i.e. if there are undiscovered extremum points, isoline mapping carried out in accordance with the protocol of Fig. \ref{fig:searchArch} will be insufficient to reconstruct the topology induced partition, and the agent will have to also apply gradient following to identify the remaining extremum points. 

This type of reasoning reveals a connection between reconnaissance and the paradigms of exploration vs. exploitation. In various fields of study, such as {\it ecology} \cite{Krebs:1978}, {\it behavioral economics} \cite{James:1991jr} and {\it cognitive psychology} \cite{Cohen:2007fk}, interrelated concepts of {\it exploration} and {\it exploitation} have proven to be useful of organizing discussions of human decision dynamics.  The distinction between these phases of information acquisition have been well described by Cohen et al. \cite{Cohen:2007fk}: ``Decisions require an exploration of alternatives before committing to exploiting the benefits of a particular choice. Furthermore, many decisions require reevaluation, and further exploration of alternatives, in the face of changing needs or circumstances.'' All decision makers face the choice of ``whether to exploit well-known but possibly suboptimal alternatives or to explore risky but potentially more profitable ones.''

In this sense, the analogy between the exploration vs. exploitation paradigm and the current set-up of unknown  field reconnaissance can be revealed by considering  the  partition cells in $\mathcal{V}'_k$.  {\it The isoline mapping that subdivides these cells is  equivalent to exploitation}. On the other hand, the discovery of new critical points and thereby the addition of elements to $\mathcal{V}'_k$ can be equated to the finding of new alternatives to be exploited. In this way,  {\it the mapping of gradient lines becomes equivalent to exploration}. 

We further formalize this notion through defining a bound on the conditional entropy that can be evaluated from the perspective of the reconnaissance agent.

\begin{corollary}
Assume that all extremum points (maxima and minima of the field) are known and define
\begin{equation}
\label{eq:relEntSup}
\overline{H}\left(\mathcal{V}_k\right) =  \sum_{V^j_k\in\mathcal{V}_k}\frac{\mu\left(V^j_k\right)}{\mu(X)} \log_2\left|-2 \chi\left(V^j_k\right)+1\right|.
\end{equation}
Then, the following inequality  holds 
\begin{eqnarray}
H(\mathcal{M}|\mathcal{V}_{k}) \leq \overline{H}\left(\mathcal{V}_k\right).
\end{eqnarray}
\end{corollary}
\begin{proof}
Noting that the entropy is maximized when the partition is divided into equal parts, we can write
\begin{eqnarray*}
-\sum_{M_i \in \mathcal{M}} \frac{\mu\left(M_i\cap V^j_k\right)}{\mu(X)}\log_2\frac{\mu\left(M_i\cap V^j_k\right)}{\mu\left(V^j_k\right)} &\leq& \sum_{n=1}^{\left|\mathcal{M}'\left(V_k^j\right)\right|}\frac{\mu\left(V^j_k\right)}{\mu\left(X\right)\left|\mathcal{M}'\left(V_k^j\right)\right|}\log_2\frac{\mu\left(V^j_k\right)}{\mu\left(V^j_k\right)\left|\mathcal{M}'\left(V_k^j\right)\right|}\\ &=&-  \frac{\mu\left(V^j_k\right)}{\mu(X)}\log_2\frac{\mu\left(V^j_k\right)}{\mu\left(V^j_k\right)\left|\mathcal{M}'\left(V_k^j\right)\right|}.
\end{eqnarray*}
 On the other hand, with an argument similar to the  proof of Theorem \ref{th:cardIneq}, it can be shown that:
\[
  \frac{\mu\left(V^j_k\right)}{\mu(X)}\log_2\frac{\mu\left(V^j_k\right)}{\mu\left(V^j_k\right)\left|\mathcal{M}'\left(V_k^j\right)\right|}= \frac{\mu\left(V^j_k\right)}{\mu(X)} \log_2\left|-2 \chi\left(V^j_k\right)+1\right|.
 \]
\end{proof}

The metric $\overline{H}\left(\mathcal{V}_k\right)$ purely depends on the data induced partition. Therefore, it represents a subjective quantification of the remaining uncertainty. The next theorem shows how it evolves under a reconnaissance protocol that restricts the isoline mapping motion programs only to the cells of $\mathcal{V}'_k$.

\begin{theorem}
\label{prop:inDe}
Let the search be conducted according to the search protocol shown in Fig. \ref{fig:searchArch}, and suppose that for every $b_k = b^{iso}(\mathbf{r}_o)$, the point $\mathbf{r}_o$ is chosen within a cell  $V_{k-1}^j\in\mathcal{V}'_{k-1}$, where  $\mathcal{V}'_{k-1}$ is defined by \eqref{eq:eulerSet}. Then, if $b_k = b^{iso}$, it follows that 
\begin{equation}
\label{eq:supDecrease}
\overline{H}(\mathcal{V}_{k}) < \overline{H}(\mathcal{V}_{k-1}),
\end{equation}
and if $b_k = b^{grad}$, it follows that
\[
\overline{H}(\mathcal{V}_{k})\geq  \overline{H}(\mathcal{V}_{k-1}),
\]
 where  $\overline{H}(\mathcal{V}_{k})$ is defined by \eqref{eq:relEntSup}. 
\end{theorem}
To prove Theorem \ref{prop:inDe}, we state the following proposition. 
\begin{proposition}
\label{prop:chiSum}
Let $\xi$ be an isoline mapped according to the protocol described by the flow chart of Fig. \ref{fig:searchArch}, such that $\xi\in V_{k-1}^j \in \mathcal{V}_{k-1}$, and which induces the update
\[
\mathcal{V}_k=\left(\mathcal{V}_{k-1}\setminus \{V^j_{k-1}\}\right)\cup\left\{V_k^a,V_k^b\right\}.
\]
 Then, the topology of the cells $V_{k}^a$ and $V_{k}^b$ satisfies the following:
\begin{equation}
\label{eq:noOne}
\chi(V_k^a),\chi(V_k^a) \leq 0
\end{equation}
and 
\begin{equation}
\label{eq:sumEuler}
\chi\left(V_{k}^{a}\right) + \chi\left(V_{k}^{b}\right) = \chi\left(V^{j}_{k-1}\right). 
\end{equation}
\end{proposition}
\begin{proof}
 The search protocol shown in Fig. \ref{fig:searchArch} dictates that all isolines are mapped starting from a previously mapped gradient line. This means that an isoline will always encircle at least a single previously discovered extremum point and as a result form two cells, both of which have Euler characteristics satisfying \eqref{eq:noOne}. 

On the other hand, \eqref{eq:sumEuler} follows directly from the Poincar\'e-Hopf theorem.
\end{proof}

The proof of Theorem \ref{prop:inDe} can now be given.

\begin{proof}[Proof of Theorem \ref{prop:inDe}]
For \eqref{eq:supDecrease} under isoline mapping, the rate of change in $\overline{H}(\mathcal{V}_{k})=\overline{H}_k$ yields
\begin{eqnarray*}
\overline{H}_k- \overline{H}_{k-1} &=& -\frac{\mu(V_{k-1}^j)}{\mu(X)}\log_2\left(-2\chi(V_{k-1}^j)+1\right) \\
&+& \sum_{i\in\{a,b\}}  \frac{\mu(V_{k}^i)}{\mu(X)}\log_2\left(-2\chi(V_{k}^i)+1\right), 
\end{eqnarray*}
where it is assumed that $V_k^a,V_k^b \subset V_{k-1}^j$, and it is taken into account that because of \eqref{eq:noOne}, it follows that
\[
-2\chi(V_{k-1}^j)+1 = |-2\chi(V_{k-1}^j)+1|\geq1 .
\]

From  the log-sum inequality, it follows  that 
\begin{eqnarray*}
&& \sum_{i\in\{a,b\}}  \frac{\mu(V_{k}^i)}{\mu(X)}\log_2\frac{1}{\left(-2\chi(V_{k}^i)+1\right)}=\sum_{i\in\{a,b\}}  \frac{\mu(V_{k}^i)}{\mu(X)}\log_2\frac{\mu(V_{k}^i)}{\mu(V_{k}^i)\left(-2\chi(V_{k}^i)+1\right)}\\
&\geq&   \sum_{i\in\{a,b\}}  \frac{\mu(V_{k}^i)}{\mu(X)}\log_2\frac{ \sum_{i\in\{a,b\}}  \frac{\mu(V_{k}^i)}{\mu(X)}}{ \sum_{i\in\{a,b\}}  \frac{\mu(V_{k}^i)}{\mu(X)}\left(-2\chi(V_{k}^i)+1\right)}\\
&>&\frac{\mu(V_{k-1}^j)}{\mu(X)}\log_2\frac{1}{\left(-2\chi(V_{k-1}^j)+1\right)}.
\end{eqnarray*}

Therefore,
\[
\sum_{i\in\{a,b\}}  \frac{\mu(V_{k}^i)}{\mu(X)}\log_2\left(-2\chi(V_{k}^i)+1\right)<\frac{\mu(V_{k-1}^j)}{\mu(X)}\log_2\left(-2\chi(V_{k-1}^j)+1\right)
\]
and respectively
\[
\overline{H}_k<\overline{H}_{k-1}.
\]
 
Under extremum search, finding a new extremum is equivalent to puncturing  a hole, and respectively, decreasing the Euler characteristic of the cell containing this extremum by $1$. When substituted back, this shows that  $\overline{H}_k$ is nondecreasing under extremum search, which concludes the proof. 
\end{proof}

This result establishes  $\overline{H}_k$ as a useful metric in the context of the exploration vs. exploitation paradigm. That is,  $\overline{H}_k$ is consumed under isoline mapping  (exploitation) and increased under gradient line mapping (exploration). We will utilize this relationship to define reconnaissance protocols that are inspired by foraging.

\section{Machine reconnaissance and topology based feedback}
\label{sec:machineRecon}

In this section,  we present two reconnaissance protocols for identifying the topology induced partition of a particular unknown field. The first one employs the metrics developed above together with  the topology of the cells in the data induced partition as a decision input. Hence, it   utilizes them to decide where and with what motion program to search next. The other strategy under consideration will be called a {\it scanning strategy}, since it will not base next steps on anything inferred about the topology. These two strategies will be compared through Monte-Carlo simulations. From this perspective the  {\it scanning strategy} can be thought of as the baseline solution to the problem against which the benefits of employing topology based feedback can be evaluated. 

\subsection{The topology feedback strategy}
\label{subsec:topGuided}
We will base the topology feedback strategy on the consumption rate,
\[
R_k = \overline{H}_{k-1} -  \overline{H}_k,
\]
which is utilized as a switching  criteria. As established  by Theorem \ref{prop:inDe}, if isoline mapping is restricted to the cells of the data induced partition with Euler characteristic equal or smaller than $-1$, $R_k$ is  strictly positive  in the exploitation mode. In the exploration mode, it is either strictly smaller than zero (if the exploration mode yields the discovery of a new extremum), or equal to zero (if it does not).  The heuristic that we propose is to {\it exploit} (map isolines), when the consumption rate lies within the set 
\begin{equation}
\label{eq:exploitCond}
\{R_k>k^{-T}\}\cup\{R_k\leq 0 \},
\end{equation}
and to {\it explore} otherwise.
The constant $T\geq 0$ determines how aggressive  the reconnaissance strategy is.   The most aggressive  behavior  is equivalent to the rate being such that $R_k<k^{-T}$, $\forall k>0$. In this case, it is easy to verify that the search sequence will consist of alternating $b^{iso}$ and $b^{grad}$ motion programs. At the other extreme, a conservative reconnaissance will correspond to exclusive  exploitation, i.e. for some time $n$, $T$ can be chosen sufficiently large s.t.  $R_k>k^{-T}$, $\forall k>n$. 

The rate $R_k$ provides a choice of a program to search with. The question that remains is at which point to initiate a particularly chosen motion routine. A consistent choice should guarantee that  the conditional entropy, $H(\mathcal{M}|\mathcal{V}_k)$, converges to zero. Under the assumption that all extremum points have been discovered, $\overline{H}_k$ is a bound on the conditional entropy. Therefore, we will show that there exists a suitable choice of an infinite sequence of originating points, which  guarantees that under the exclusive execution of $b^{iso}$ programs, $\overline{H}_k$ goes to zero. 

\begin{theorem}
\label{prop:toZero}
Assume that $B_k$ is generated according to the reconnaissance protocol described by Fig. \ref{fig:searchArch}. Moreover, assume that there exists $n$ such that $b_k=b^{iso}(\mathbf{r}_o),~ \forall k>n $, with the points $\mathbf{r}_o$ chosen at each instance according to the following sequence:
\begin{itemize}
\item Define the segments $l_{ij}=V_k^i\cap \zeta_j$ where $V_k^i \in \mathcal{V}_k$, $\zeta_j \in Q(B_k)$  ($\zeta_j$ a gradient path), and let the constants $A_{ij}$, $a_{ij}$ associated with each segment be defined as:
\begin{eqnarray*}
a_{ij} = \inf _{\mathbf{r}\in l_{ij}} f(\mathbf{r})\\
A_{ij} = \sup _{\mathbf{r}\in l_{ij}} f(\mathbf{r})
\end{eqnarray*}
\item Find a segment $l_{ij}^*$ such that .
\begin{equation}
\label{eq:maxSegment}
l_{ij}^{*} = \arg\max_{l_{ij} = V_k^i\cap \zeta_j} \left( A_{ij} - a_{ij} \right).
\end{equation}
\item Choose the point, $\mathbf{r}_o$ such that $\mathbf{r}_o\in l_{ij}^*$ and
\[
f(\mathbf{r}_o) = \frac{f(A_{ij})+f(a_{ij})}{2}.
\]   
\end{itemize}

Then, the following holds:
\[
\lim_{k\to\infty}\overline{H}_k = 0.
\]
\end{theorem}
\begin{proof} 
Define the metric 
\[
L_k = \sum_{V_k^i\in \mathcal{V}'_k} \sum _{M_j\in\mathcal{M}} \left(\sup_{\mathbf{r}\in V_k^i\cap M_j} f(\mathbf{r}) - \inf_{\mathbf{r}\in V_k^i\cap M_j} f(\mathbf{r})\right),
\]
where $\mathcal{V}'_k$ is the set of cells in $\mathcal{V}_k$ having Euler characteristic $\leq-1$ as defined by \eqref{eq:eulerSet}, and let $N_k$ be the cardinality of the set $\{ V_k^i\cap M_j:V_k^i\in  \mathcal{V}'_k, M_j\in\mathcal{M}\}$. Given that the reconnaissance is executed under the protocol described by Fig. \ref{fig:searchArch}, Prop. \ref{prop:chiSum} dictates that every time a cell is split by isoline mapping, the sum of the Euler characteristics of the two new cells is equal to the Euler characteristic of the original, and moreover, neither of the two resultant cells have Euler characteristic of $1$.  

Assume that one of the new cells formed by the isoline mapping has Euler characteristic of $0$, then this new cell will not be a part of $\mathcal{V}'_{k+1}$, and therefore,  the metric $L_k$ evolves according to:
\begin{equation}
\label{eq:pL}
L_{k+1}  <L_{k}\left( \frac{2 N_k -1}{2 N_k}\right).
\end{equation}
This follows from \eqref{eq:maxSegment}, i.e. the protocol dictates that the segment that is cut in half is the largest one. 

Another possible scenario is that both of the new cells have Euler characteristic smaller than $0$. Again, according to  Prop. \ref{prop:chiSum} this can only occur when the original cell has Euler characteristic smaller than $-1$ (the sum of the Euler characteristics of the resultant cells should equal to the Euler characteristic of the original one).  In this case,  $L_k$ does not change, but it is easy to see that:
\begin{equation}
\label{eq:cardinalityChange}
N_{k+1} = N_{k} + 1,
\end{equation}
since $\mathcal{V}'_{k+1}$ will have one more element than $\mathcal{V}'_{k}$. Also from Proposition \ref{prop:chiSum} it can be concluded that this scenario can occur only a limited number of times until all the cells in $\mathcal{V}'_{k+1}$ have Euler characteristic of $-1$. That is, each cell $V_k^j\in \mathcal{V}'_k$, such that $\chi(V_k^j)=n<-1$, can be divided by a sequence of isoline mappings into at most $n$ cells with $\chi=-1$. Therefore, we can have $L_k = L_{k+1}$ at most $-\chi-1$ times.   By defining $C_k$ as:
\[
C_k = \sum_{V_k^i\in\mathcal{V}'_k} -1 -\chi(V_k^i),
\]
 for every case of $l>k$, we can write:
\[
L_{l} < L_k \left( \frac{2 N_k -1}{2 N_k}\right)^{\max[0,l-k-C_k]}.
\]
Thus,  as $l\to\infty$, $L_l$ converges to zero.  Since $f(\mathbf{r})$ is Morse function,
\[
\sum_{V_k^i\in \mathcal{V}'_k} \sum _{M_j\in\mathcal{M}} \mu ({V_k^i\cap M_j}),
\]
will also converge to $0$ as $k\to\infty$ and consequently so will $\overline{H}_k$. 
\end{proof}

To complete the specification of the topology guided strategy, we designate a random procedure for choosing the originating points for  gradient tracking. That is, the initialization points for the $b^{grad}(\mathbf{r}_o)$ programs will be chosen at random with respect to a uniform probability density on the set of points on the mapped isolines $S(B_k)$. The resulting reconnaissance protocol is summarized in Alg. \ref{alg:searchStrategy}.  Fig. \ref{fig:reconStratRun} shows the isolines mapped during  a particular run of this topology driven reconnaissance strategy. It can be clearly observed that the mapped isolines are concentrated near the critical level sets associated with saddle points.

\begin{algorithm}[h]
\caption{Topology guided reconnaissance  strategy}
\label{alg:searchStrategy}
\begin{algorithmic}[1]
\STATE $k = 0$
\STATE $B_k = \{\}$.
\STATE $\mathcal{V}_k = \{X\}$ \COMMENT{Initialization with no prior information}
\LOOP
\IF { $R_k > k^{-T}$ or $R_k \leq 0$ }
\STATE Choose $\mathbf{r}_o$ according to the sequence described in Thm. \ref{prop:toZero}
\STATE Execute $b_{k+1}=b^{iso}(\mathbf{r}_o)$  to map $\xi$
\STATE $\mathcal{V}_{k+1}= \left(\mathcal{V}_{k}\setminus \left\{V_{k}^j\right\}\right)\cup \text{cc}\left({V}_{k}^j\setminus\xi\right)$
\ELSE 
\STATE Get $\mathbf{r}_o\sim U(S(B_k))$ (That is, choose the point at random from a uniform distribution on the set points on the mapped isolines.)
\STATE Execute $b_{k+1}=b^{grad}(\mathbf{r}_o)$  to map $\zeta$
\FOR {All critical points $\mathbf{r}^*_i$ uncovered by $\zeta$}
\STATE $\mathcal{V}_{k+1}= \left(\mathcal{V}_{k}\setminus \left\{V_{k}^j\right\}\right)\cup \left\{{V}_{k}^j\setminus\mathbf{r}^*_i\right\}$
\ENDFOR
\ENDIF
\STATE $B_{k+1} = B_k \cup \{b_{k+1}\}$
\STATE $k = k + 1$
\ENDLOOP
\end{algorithmic}
\end{algorithm}

\begin{figure}[htbp] 
   \centering
   \includegraphics[width=200pt]{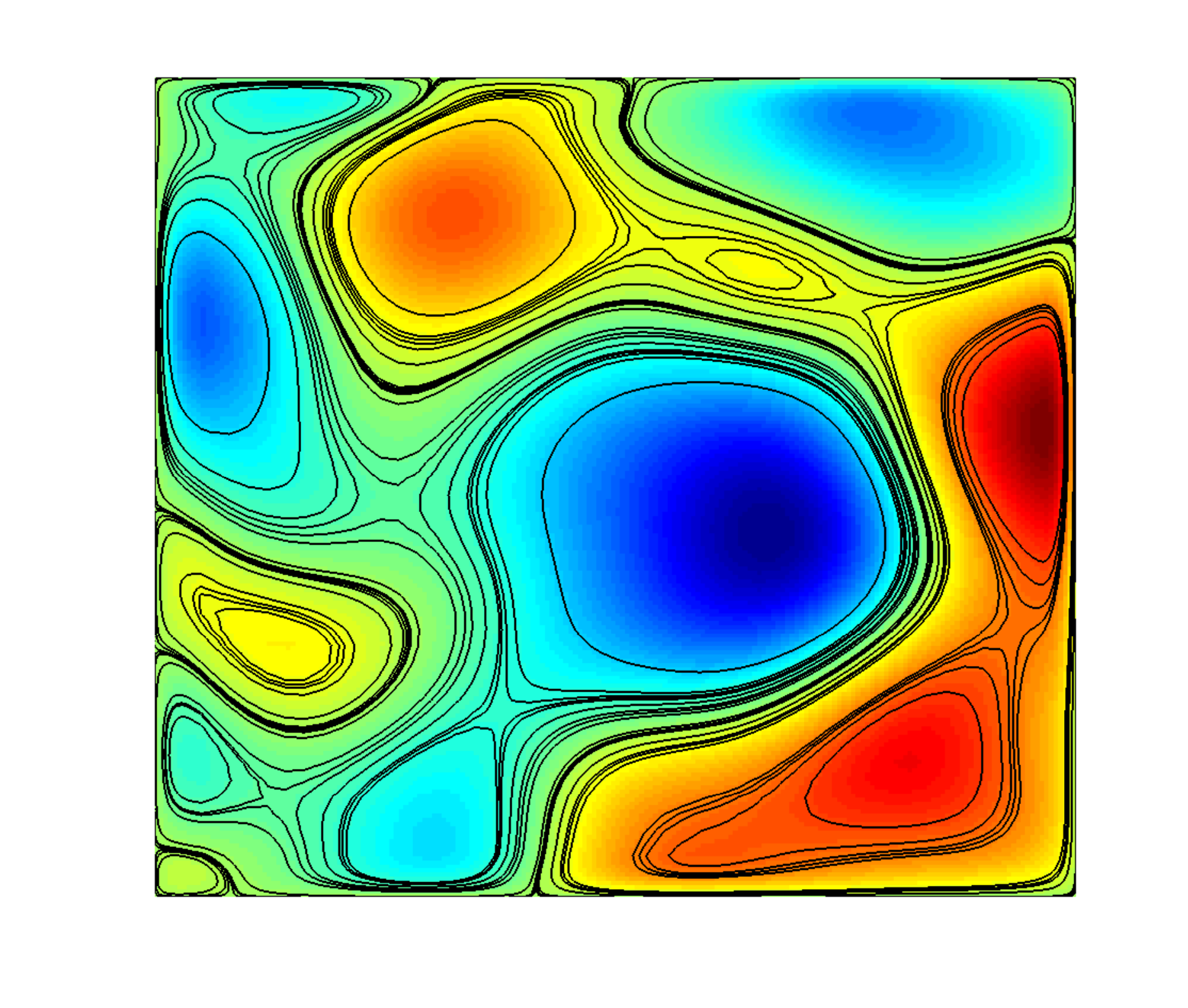} 
   \caption{The isolines mapped under a particular run of the reconnaissance strategy described by Alg. \ref{alg:searchStrategy} for $T=0.5$.}
   \label{fig:reconStratRun}
\end{figure}

\subsection{Open loop scan strategy}
The scanning strategy that will be further utilized as a benchmark  is a simple routine, which divides a mapped gradient line into $n$ parts with equal range drop and maps the isolines that originate from them. It is formally described by Alg. \ref{alg:nscan}.
\begin{algorithm}[htbp]
\caption{Scanning reconnaissance  strategy }
\label{alg:nscan}
\begin{algorithmic}[1]
\STATE $k = 0$
\STATE $B_k = \{\}$.
\STATE $\mathcal{V}_k = \{X\}$ 
\LOOP
\STATE Choose $\mathbf{r}_o\sim U(X)$
\STATE Execute $b^{grad}(\mathbf{r}_o)$ to map $\zeta$
\STATE Find $\{\mathbf{r}_i\in\zeta:1\leq i \leq n\}$ with $f(\mathbf{r}_i)$ corresponding to a uniform partition of the range $f(\zeta)$
\FOR{$i = 1:n$}
\STATE Execute $b^{iso}(\mathbf{r}_i,X)$ to map $\xi$
\STATE $\mathcal{V}_{k+1}= \left(\mathcal{V}_{k}\setminus \left\{V_{k}^j\right\}\right)\cup \text{cc}\left({V}_{k}^j\setminus\xi\right)$
\ENDFOR 
\ENDLOOP
\end{algorithmic}
\end{algorithm}

\begin{figure}[htbp] 
   \centering
  \subfigure[]{ \includegraphics[width=2in]{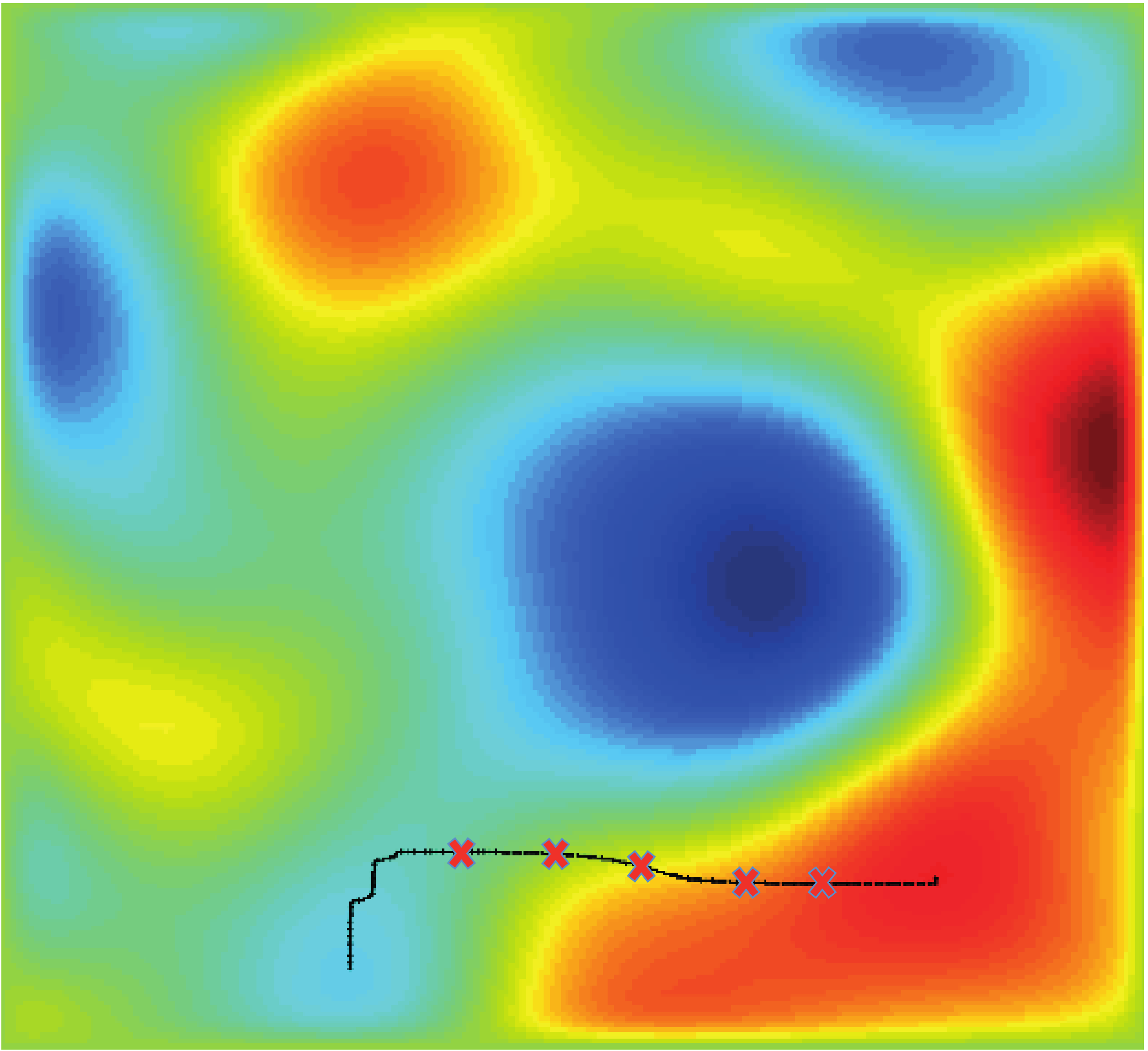}\label{subfig:stage1}}
   \subfigure[]{ \includegraphics[width=2in]{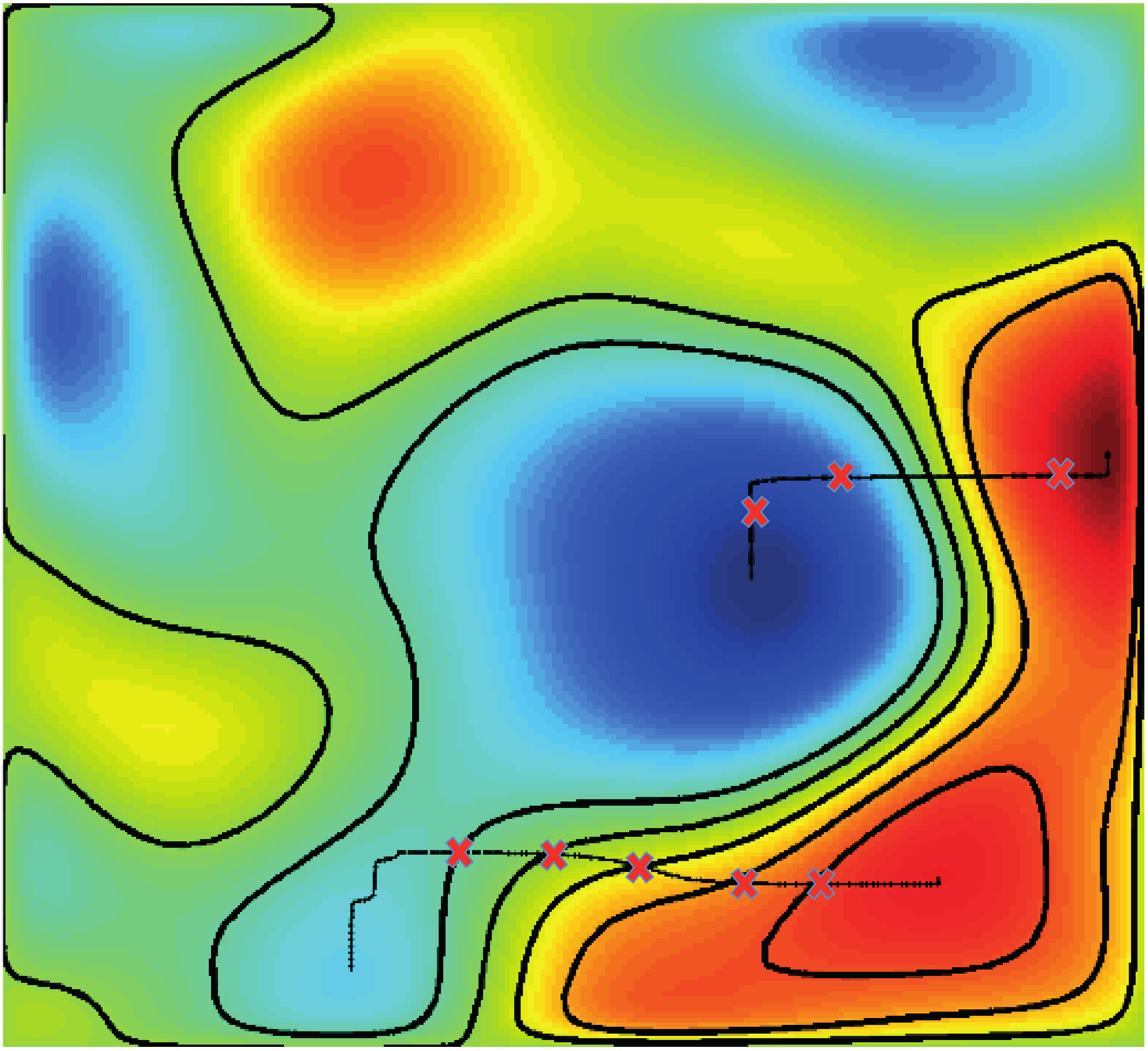}\label{subfig:stage2}} 
   \caption{An illustration of the initial several steps in mapping an unknown  field through the $n-$scan strategy}
   \label{fig:nScan}
\end{figure}
An illustration of the first several steps in executing this strategy is presented in  Fig.~\ref{fig:nScan}. After a gradient line is mapped, its range is partitioned uniformly into $n$ parts (Fig. \ref{subfig:stage1}).  After mapping the isolines corresponding to this partition, the robot moves to a point randomly chosen on an unexplored part of the search domain and again maps a gradient line and then repeats the same isoline mapping  routine (Fig. \ref{subfig:stage2}).

\subsection{Monte Carlo Validation}
To validate and compare these two reconnaissance protocols, we will use scalar fields that are realizations of a spatial gaussian process. Such fields are widely used in atmospheric modeling and data assimilation  \cite{Kalnay:2003vx}, and they have also been employed as underlying models for studies in unknown field reconnaissance \cite{Leonard:2006gh}. 

They are defined as follows: for each point $\mathbf{r}\in X$, the value of the function is considered as a realization of a random process  $f(\mathbf{r})\sim N(0,1)$, which  for each two points $\mathbf{r}_1,\mathbf{r}_2\in X$ has a correlation
\begin{equation}
\label{eq:potField}
E\left[f(\mathbf{r}_1),f(\mathbf{r}_2)\right] = \exp\left(-\frac{\|\mathbf{r}_1-\mathbf{r}_2\|^2}{2 d^2}\right).
\end{equation}

Reference  \cite{Adler2} provides conditions under which  the function  $f$, as a  realization of a random field, will be a Morse function with probability $1$. In practical terms, we can calculate the realization of the  field at discrete grid points, and then estimate the value of the field at each point of the domain through linear interpolation.

\begin{figure}[htbp] 
   \centering
    \subfigure[]{\label{subfig:stratD}\includegraphics[width=200pt]{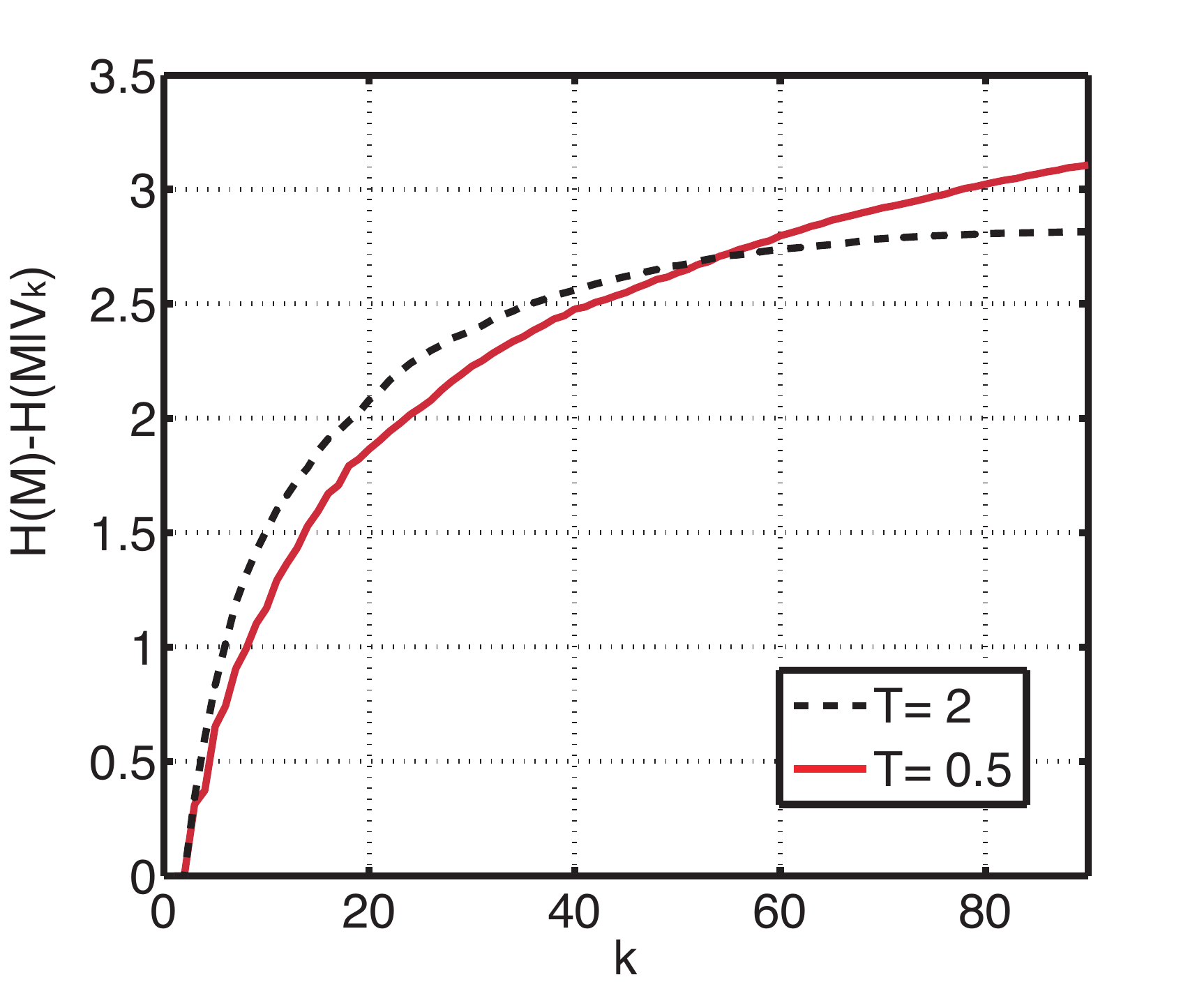}}
    \subfigure[]{\label{subfig:stratSTD}\includegraphics[width=200pt]{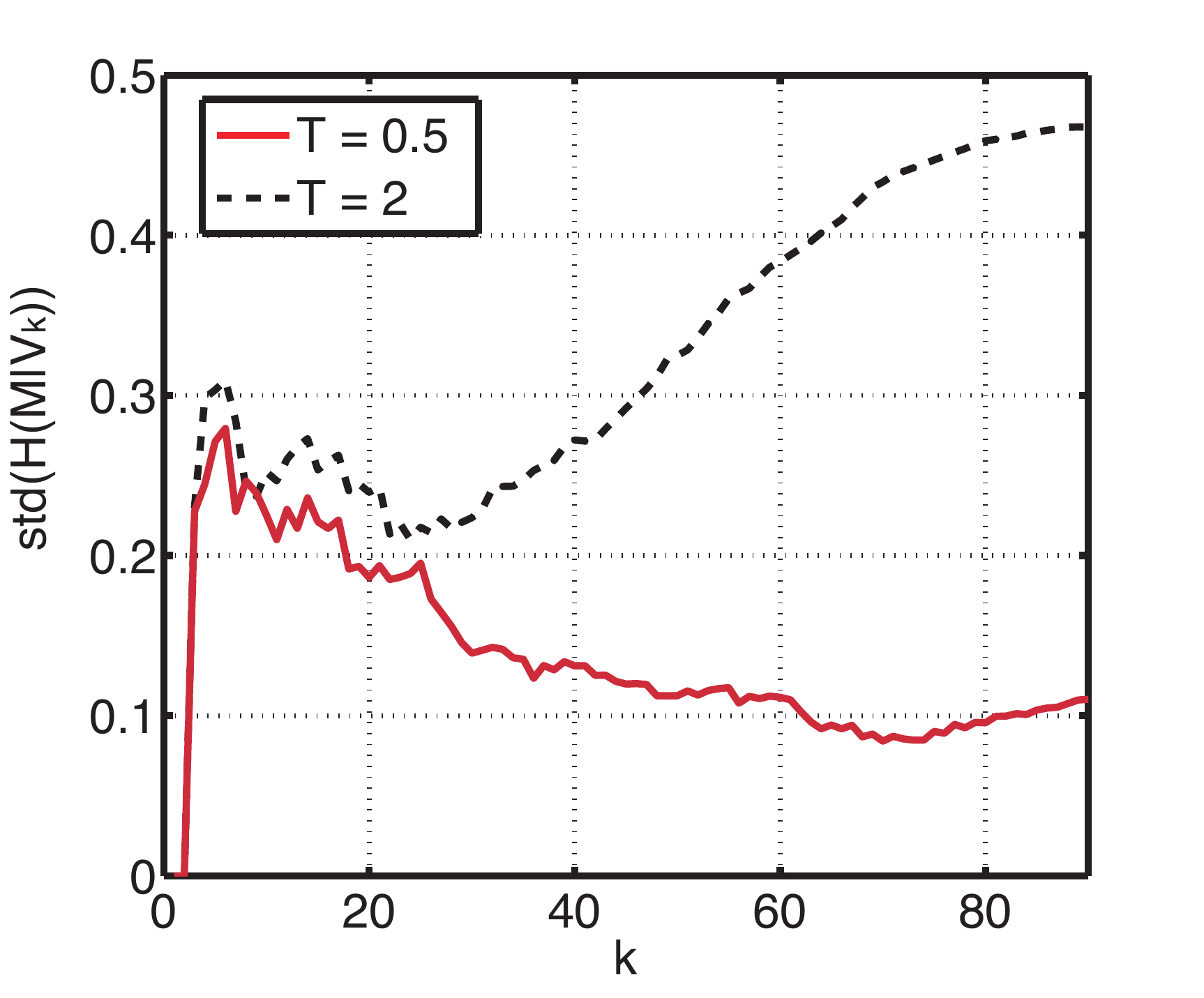}}
   \caption{The evolution of the conditional entropy for two different values of $T$: an aggressive strategy ($T=0.5$) and a more conservative strategy ($T=2$).}
   \label{fig:strat}
\end{figure}

In Fig. \ref{fig:strat}, we show the result of Monte-Carlo simulations of  applying the topology guided strategy with two different values of $T$: an aggressive strategy corresponding to $T=0.5$, and a more conservative strategy corresponding to a larger value, $T = 2$. Fig. \ref{subfig:stratD} plots the average of
\[
H(\mathcal{M}) - H(\mathcal{M}|\mathcal{V}_k),
\]
for each instance $k$ over  $100$ runs of the strategy, each on a different randomly generated field. Both strategies are initialized with $\mathcal{V}_0=\{X\}$, which implies that $ H(\mathcal{M}|\mathcal{V}_0)= H(\mathcal{M})$. An important observation is that the strategy with the large $T$ is more efficient in the short term. This is due to the fact that the short term exploitation of the available resources gives an initial gain. However, the strategy with the large $T$ is overwhelmed by the strategy with the small $T$ in the long run, as the early investment in finding extremum points pays off and  provides more area to be exploited as the reconnaissance progresses. Fig. \ref{subfig:stratSTD}  gives the standard deviation of the conditional complexities of the two strategies. This shows that the more aggressive strategy in this case is also the more consistent in the long term. 

\begin{figure}[htbp] 
   \centering
    \subfigure[]{\includegraphics[width=182pt]{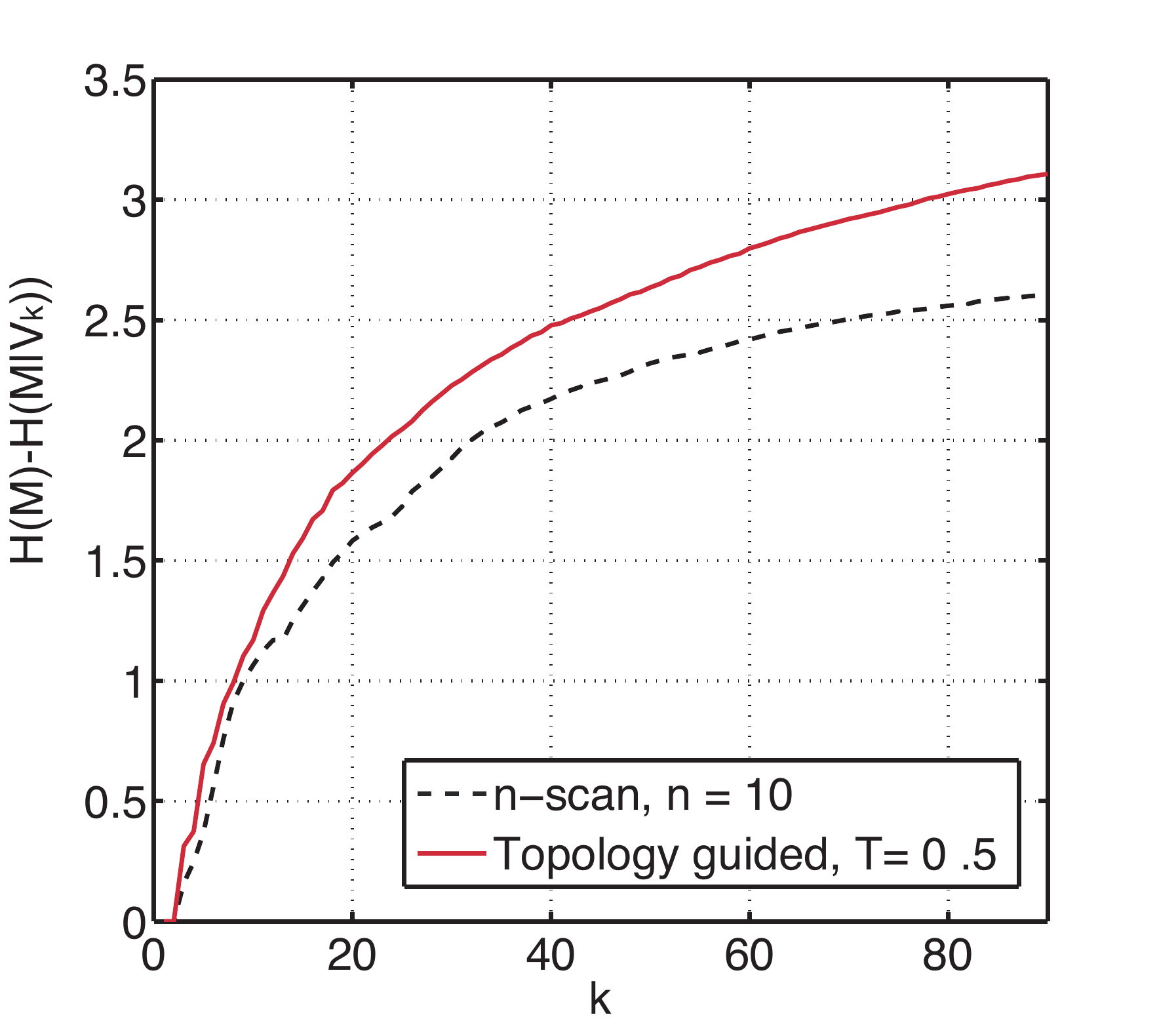}}
    \subfigure[]{\includegraphics[width=200pt]{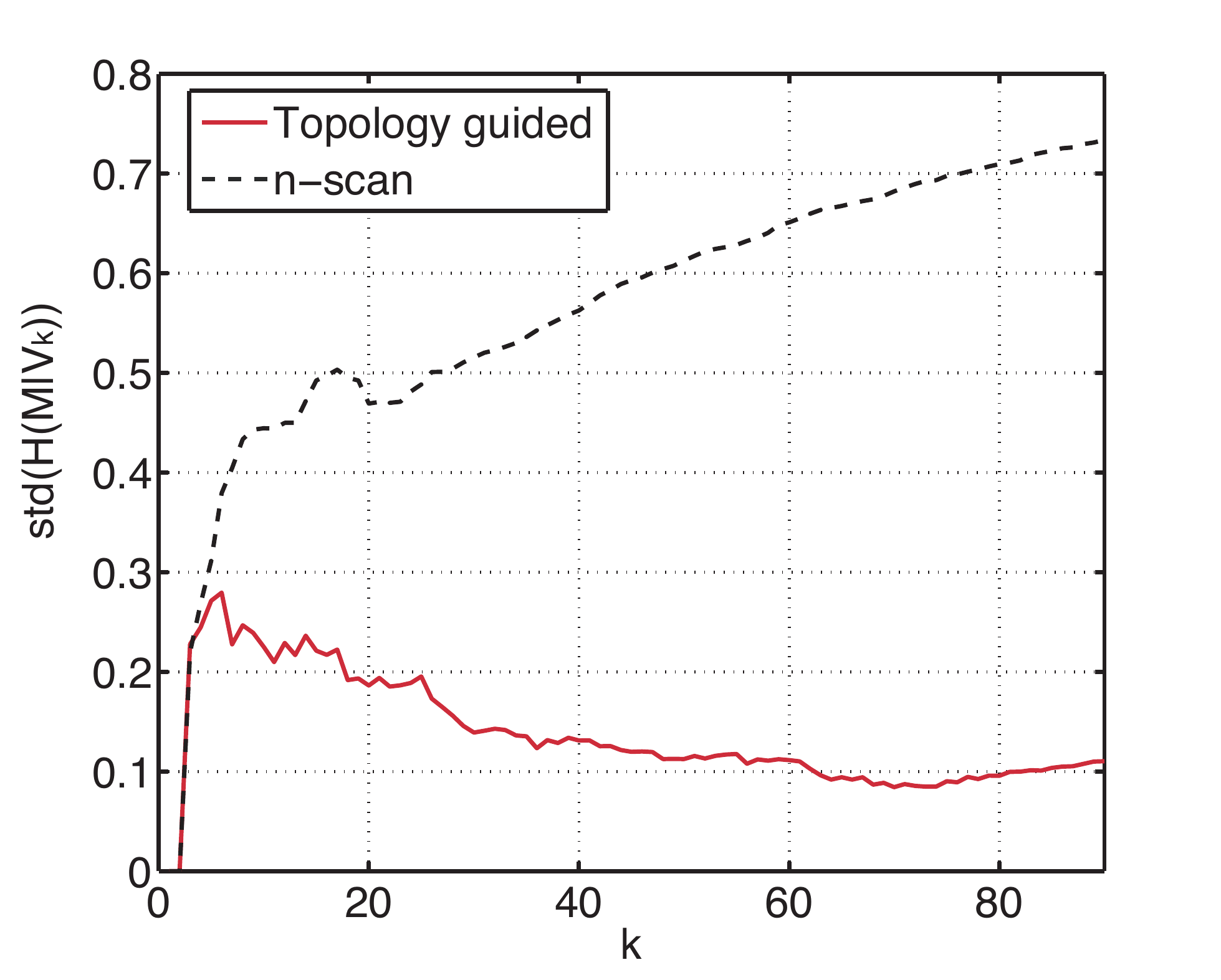}}
   \caption{A comparison of $H(\mathcal{M})- H(\mathcal{M}|\mathcal{V}_k)$ comparing  a topology guided strategy that uses feedback of information about the topological characteristics of cells currently in the data induced partition and an open-loop, $n$-scan strategy.}
   \label{fig:comp}
\end{figure}

Fig. \ref{fig:comp} presents the topology guided strategy compared with the benchmark $n$-scan strategy. The topology guided strategy can be clearly identified as the more efficient one. It consistently collects more information and achieves this with smaller dispersion.

\section{Human guided  reconnaissance}
\label{sec:humanRecon}
In the previous section, we have proposed a protocol  for unknown field reconnaissance, whose  goal is to uncover the topological structure of the field. This design objective has been motivated by the need to efficiently identify a minimum set of features,  whose mapping  can yield a good qualitative and quantitative understanding of the underlying phenomenon. In this section, the focus shifts from the design of algorithms to the assessment of human decision making. In particular, the guiding conjecture will be that humans intuitively recognize the importance of the field's topology and therefore manifest a preference for discovering the topology induced partition.

\subsection{Reconnaissance game for testing human decision making}

To evaluate human-in-the-loop reconnaissance, we have designed a  game, which simulates the mapping of a unknown field by a robotic vehicle that is under the supervisory control of a human subject. The subject is  instructed that the hypothetical reconnaissance  agent under her  control can perform two types of tasks---the mapping of  gradient lines and the  mapping of  isolines. Hence, the subject can utilize the same two search programs as the ones that are employed by the machine reconnaissance protocols discussed in the previous section.  

The scalar field is defined on a square search domain and it is generated according to  \eqref{eq:potField}, with two correlation lengths--$d=0.25$ and $d=0.5$ (given that the square has unit length sides). Fig. \ref{fig:exampleFields} illustrates that by varying the correlation length, we effectively control the complexity of the topology, and hence of the induced partition.  Every  subject is presented with multiple fields to map. The fields are chosen from a large archive with each field having a unique identifying number. Every even numbered field has high topological complexity ($d=0.25$), and every odd numbered field has low topological complexity ($d=0.5$).  

\begin{figure}[htbp] 
   \centering
  \subfigure[Low topology complexity area ($d = 0.50$).]{\label{subfig:lowTop}\includegraphics[width=120pt]{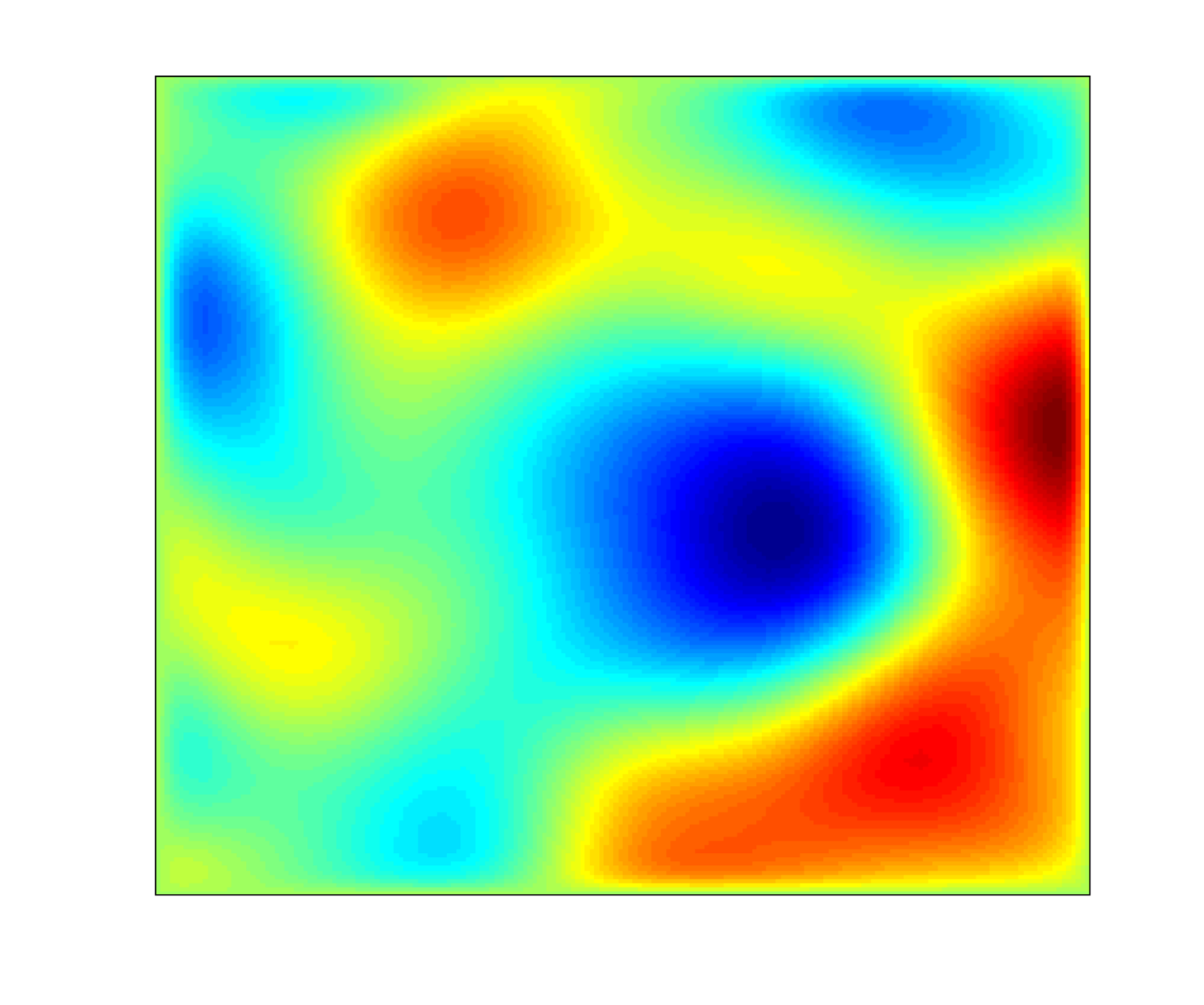}}
    \subfigure[Hight topology complexity area ($d = 0.25$).]{\label{subfig:highTop}\includegraphics[width=120pt]{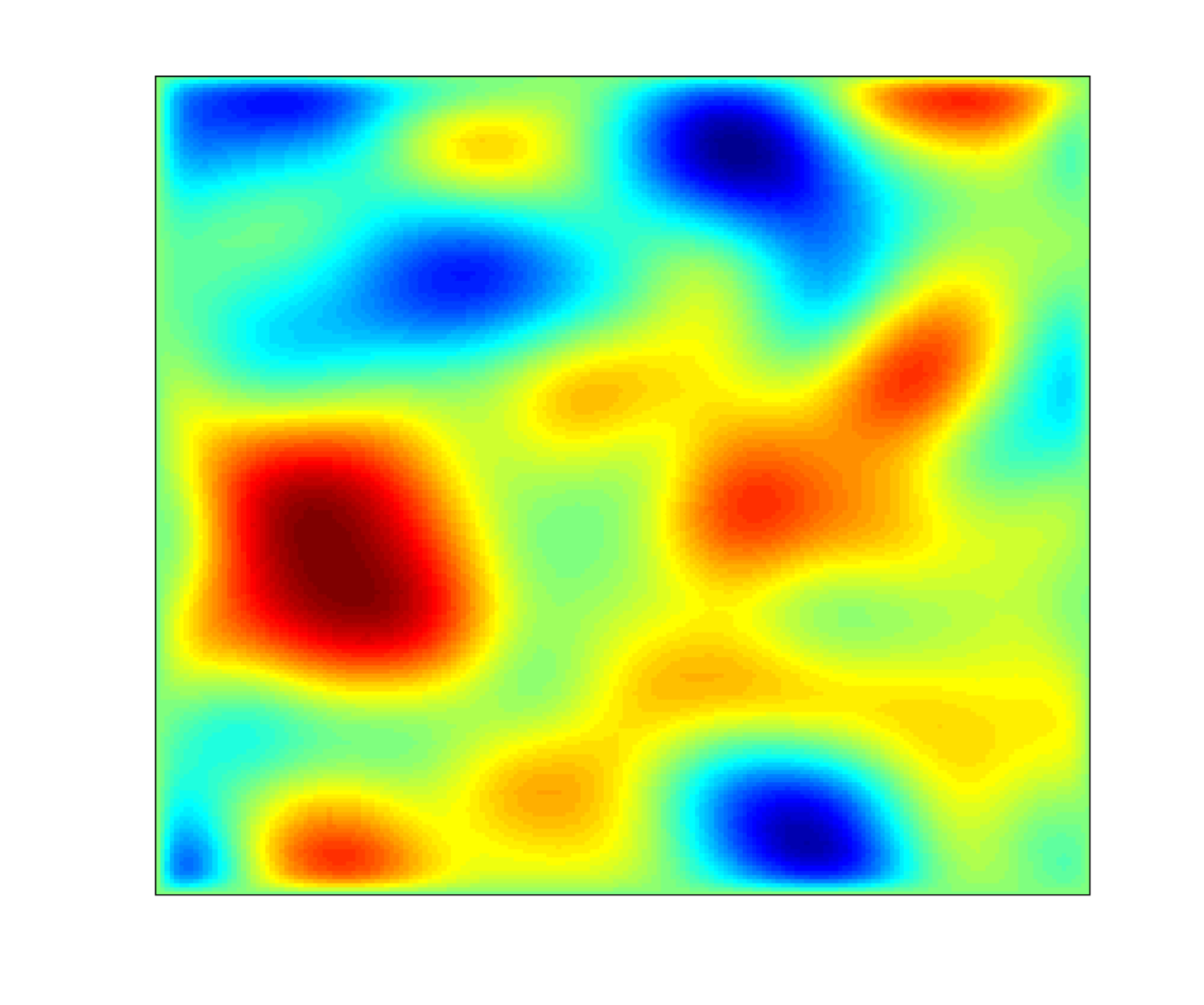}}
   \caption{Two examples of high and low topology complexity fields.}
   \label{fig:exampleFields}
\end{figure}
Every  subject is instructed by a video \cite{gameInstr:2010} before the experiment, which assures the uniformity of the instructions. The basic points that are conveyed by  the video are:
\begin{itemize}
\item The game involves robotic reconnaissance  of an unknown terrain.
\item The player  teams up with a  virtual robot that has been programmed to perform simple motions---mapping gradient lines and mapping isolines.
\item The responsibility of the player is to instruct the robot how to explore a specific area and,  when she believes the acquired map of the given area is of sufficient quality, to command the robot to another one.
\item  The subject is told that there is no wrong choice of what constitutes a sufficient mapping of a terrain.
\end{itemize}

The game interface  is shown in Fig. \ref{fig:gameInterface}. The subjects are given $12$ minutes for mapping. When the clock runs down to zero, the game automatically stops.  Every time the subject clicks on the empty space, she automatically maps a gradient line (Fig. \ref{subfig:click1}). On the other hand, when a subject clicks on a gradient path, the result is the mapping of an isoline (Fig. \ref{subfig:click2}).  Whenever a subject  believes that the map is of  sufficient quality,  she clicks on  the ``Go to another area'' button and starts mapping a new square domain (Fig. \ref{subfig:click3}).  

\begin{figure}[htbp] 
   \centering
  \subfigure[]{\label{subfig:click1}\includegraphics[width=150pt]{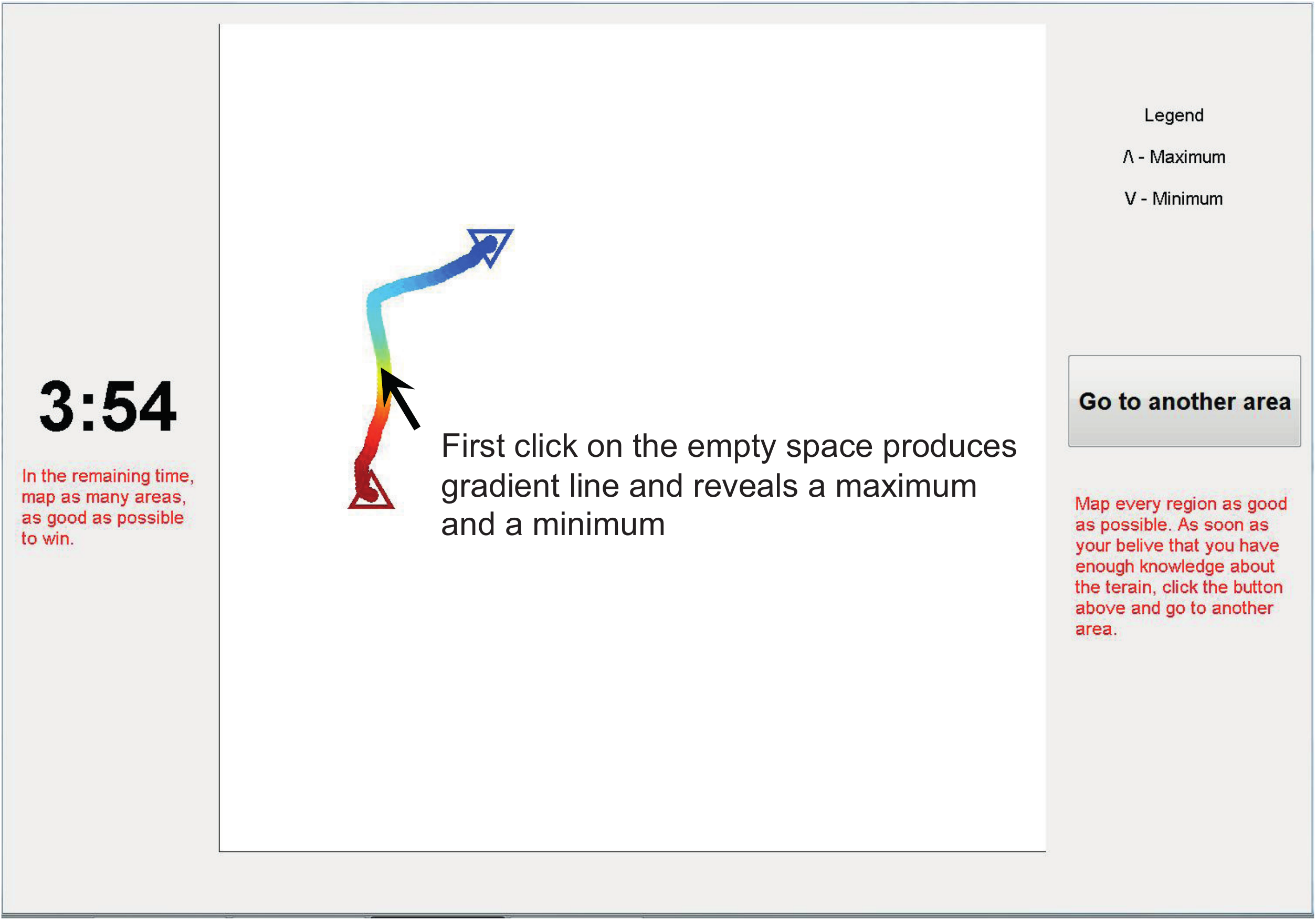}}
    \subfigure[]{\label{subfig:click2}\includegraphics[width=150pt]{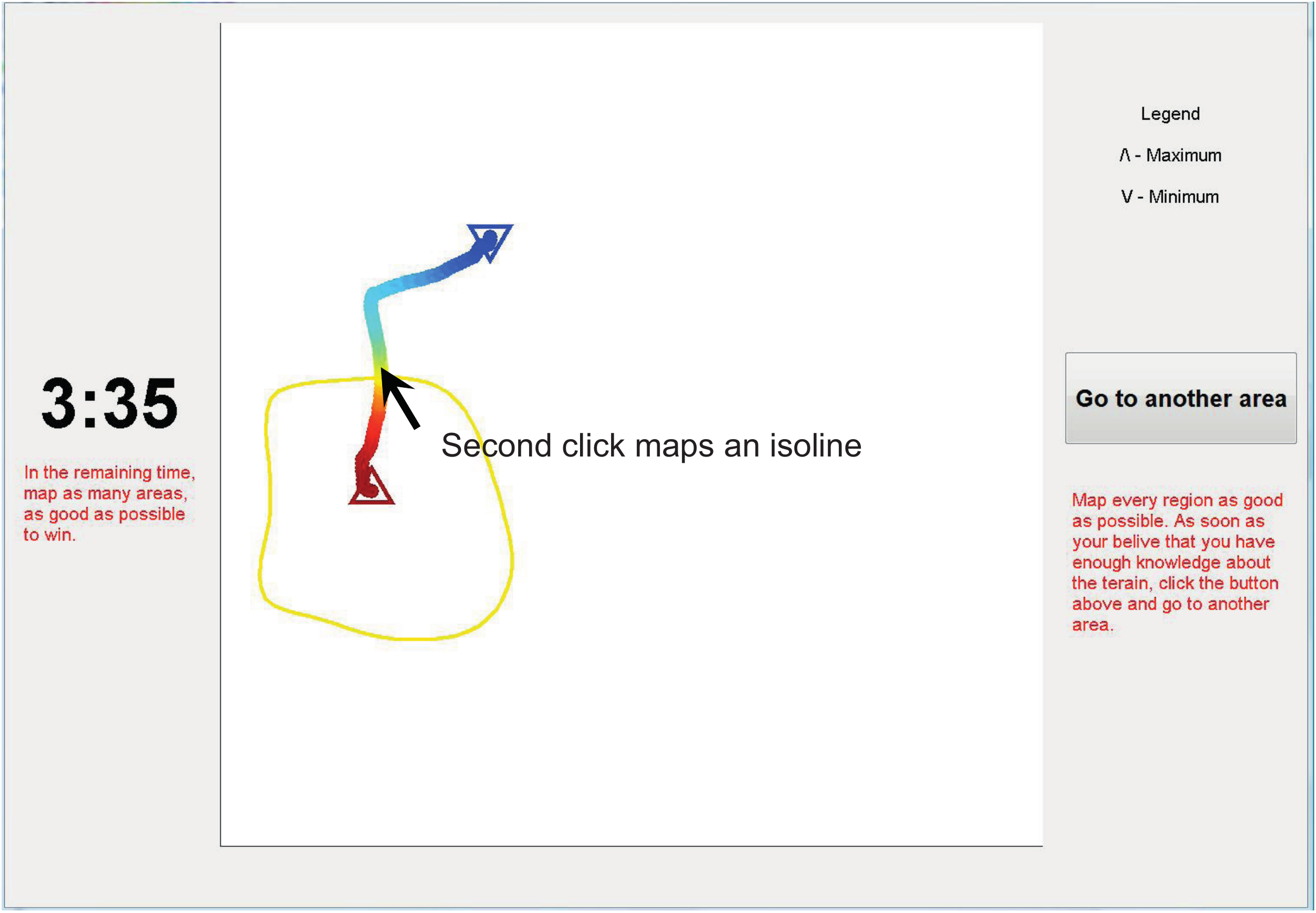}}
      \subfigure[]{\label{subfig:click3}\includegraphics[width=150pt]{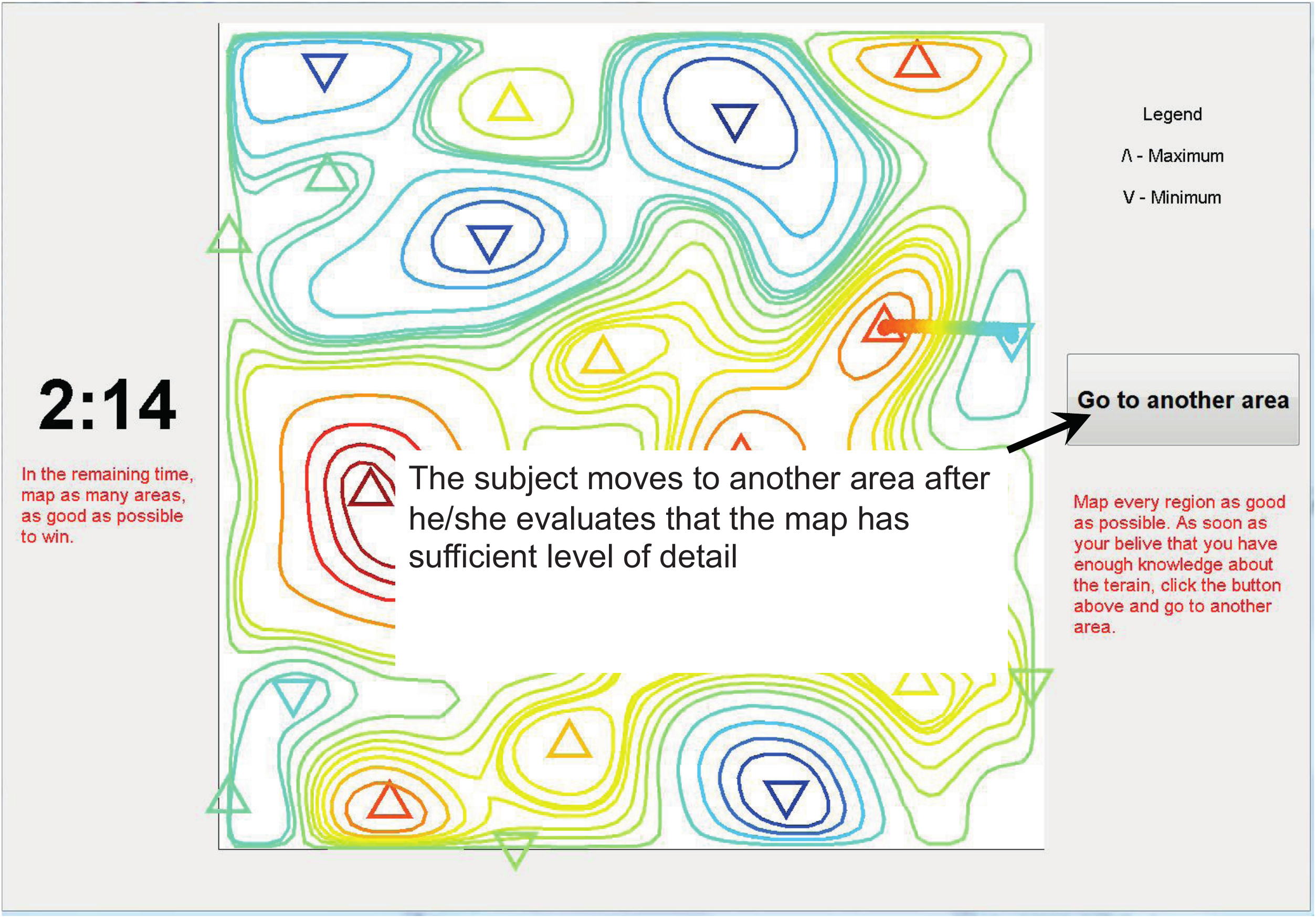}}
   \caption{A basic depiction of the game interface as described in the text.}
   \label{fig:gameInterface}
\end{figure}

\subsection{Topology feedback in human subjects}

As we have seen in Section \ref{sec:information}, acquiring topology relevant information is equivalent to continuously decreasing the conditional entropy. Therefore, if a subject is utilizing topology feedback with the goal of uncovering the structure  induced by the critical points, she will have a preference for mapping isolines that reduce the conditional entropy.

Corollary \ref{thm:uncertainty} established that the uncertainty about the structure of the topology induced partition is concentrated within cells of the data induced partition that--assuming all extremum points are removed from them--have Euler characteristic $\chi\leq-1$. This observation leads to a natural procedure for determining  the existence of a preference for identifying the topology of the field. The data induced partition can be divided into two parts, $\mathcal{V}^{o}_{k}$, and its complement $\mathcal{V}_k\setminus\mathcal{V}^{o}_{k}$. In this division, the set $\mathcal{V}^{o}_{k}$ contains all cells of the data induced partition that have Euler characteristic less than or equal to $-1$, together with the cells that would have such characteristic under the discovery of new extremum points. In other words, $\mathcal{V}^{o}_{k}\subseteq \mathcal{V}'_{k}$, where $\mathcal{V}'_{k}$ is defined by \eqref{eq:eulerSet}. If all extremum points have been discovered, then $\mathcal{V}^{o}_{k}= \mathcal{V}'_{k}$. It can be easily verified that every time an isoline is mapped within $\mathcal{V}^{o}_{k}$ the conditional entropy is decreased, while if it is mapped within $\mathcal{V}_k\setminus\mathcal{V}^{o}_{k}$, it stays the same. Therefore, the preference of a subject for mapping isolines within $\mathcal{V}^{o}_{k}$, as opposed to mapping isolines within its complement, will also imply a preference for identifying the cells of $\mathcal{M}$. 

To quantify this propensity, we state the hypothesis that the probability of a subject clicking on a point $\mathbf{r}_{i}$ in a cell $V$ belonging to $\mathcal{V}^{o}_{k}$ to map an  isoline is given by:
\[
P^{\beta}(\mathbf{r}_i\in V\in \mathcal{V}^{o}_{k}) = \left(\frac{\mu(\mathcal{V}^o_k)}{\mu(X)}\right)^\beta,
\]
where $\beta>0$ is a positive constant  for each subject. 

If $\beta=1$, the probability of the subject  mapping an isoline within a cell belonging to $\mathcal{V}^{o}_{k}$ is equal to its normalized area. In other words, the behavior of the subject is random, and she does not utilize any topology feedback. However, if $\beta=0$, the subject exclusively maps isolines within $\mathcal{V}^{o}_{k}$, and in this way achieves continuous reduction of the conditional entropy. (The topology guided protocol described in Sec. \ref{subsec:topGuided} has $\beta=0$.)  Therefore, the preference of the subject towards identifying the topology can be directly quantified by the value of $\beta$, with smaller values corresponding to higher preference.   

Assuming that the clicks are independent of each other, the probability of a particular sequence of isoline mapping clicks,
\[
\{\mathbf{r}_1,\mathbf{r}_2,\cdots,\mathbf{r}_m\},
\]
where $\mathbf{r}_i$ corresponds to the position of the $i^{th}$ click,  will be given by:
\begin{equation}
\label{eq:probability}
P(\mathbf{r}_1,\mathbf{r}_2,\cdots,\mathbf{r}_m) = \prod_{i=1}^m \left\{   \begin{array}{cc}
      P^{\beta}(\mathbf{r}_i) & \text{if~}\mathbf{r}_i\in V\in \mathcal{V}^o_{k_i}\\
      1 - P^{\beta}(\mathbf{r}_i)  & \text{else} \\
   \end{array}\right..
\end{equation}
Here, we have used the indexing $n_i$ to denote the state of $\mathcal{V}^o_{k_i}$ immediately before the $i^{th}$ isoline is mapped.

The parameter $\beta$ can be estimated through a maximum likelihood estimator as
\[
\beta = \arg\max_{\beta>0} P(\mathbf{r}_1,\mathbf{r}_2,\cdots,\mathbf{r}_m). 
\] 

\begin{figure}[htbp] 
   \centering
\includegraphics[width=200pt]{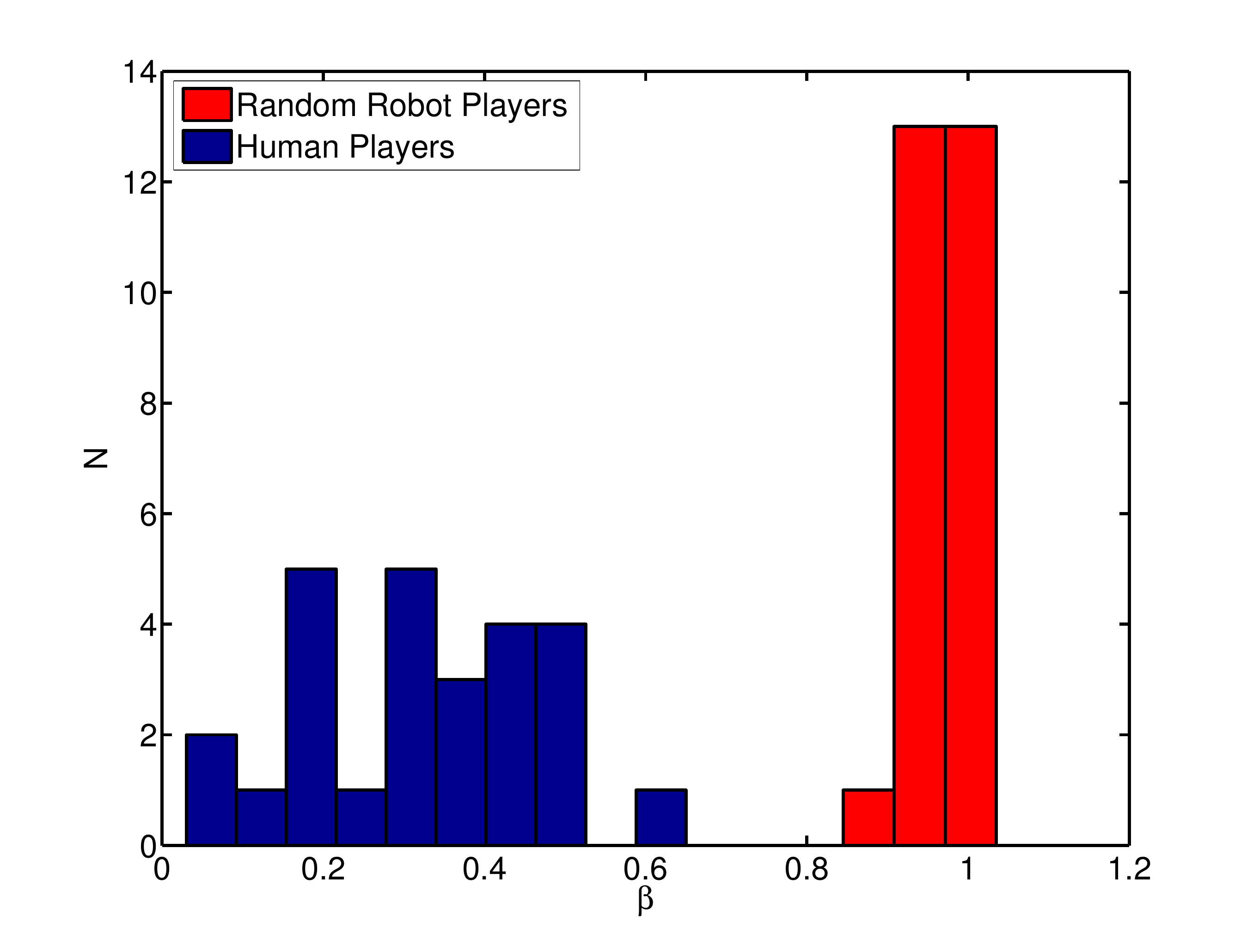} 
   \caption{Histograms of the estimated parameter $\beta$ for two populations--- a human player population and a random computer player control population.}
   \label{fig:betaHist}
\end{figure}

Fig. \ref{fig:betaHist} shows the histograms of $\beta$ for two populations of players. The first population consists of $26$ human players, who were subjects of the above described experiment. The control population consists of $26$ ``computer players'', who conduct reconnaissance on the same  fields that were mapped by the human players by mapping the same number of isolines. The difference between the control  population and the human population is that the control population is programmed to map the isolines randomly by choosing the originating points from a uniform random distribution on the search domain.   

The statistics of the two samples, whose histograms are shown in Fig. \ref{fig:betaHist},  are as follows: the mean of the human population is $\bar{\beta}_{human} = 0.32$ and the standard deviation is $\sigma_{human} = 0.15$; the mean of the computer players population is $\bar{\beta}_{computer} = 0.97$ with standard deviation of $\sigma_{computer} = 0.04$. As expected, the computer population has a $\beta$ characteristic close to $1$, which corresponds to their random game strategy. At the same time, the small value of  $\beta$ for the human population is an evidence that the humans have a strong preference for discovering the topology. 

To confirm the difference between the random style of the computer players, and the  style of the humans, we  performed statistical significance testing. {\it Kolmogorov-Smirnov statistics} (K-S) was used to test the hypothesis that the two populations were sampled from the same random process. The K-S test rejects or confirms such a hypothesis based on the distance between the empirical cumulative distribution functions (CDF) of  two samples. For the two populations under investigation, this distance was calculated to be $1$. (The same can be also concluded from Fig. \ref{fig:betaHist}). This distance corresponds to $p-value=6\times10^{-13}$, which is the probability of achieving such extreme test statistic by chance, i.e. the probability that the two populations have been sampled from the same random distribution.  To calculate the test statistic and its corresponding $p-value$, we used the {\verb kstest2 } function of Matlab R2010a, which in turn has been developed according to \cite{Marsaglia:2003}.

\begin{figure}[htbp] 
   \centering
    \includegraphics[width=200pt]{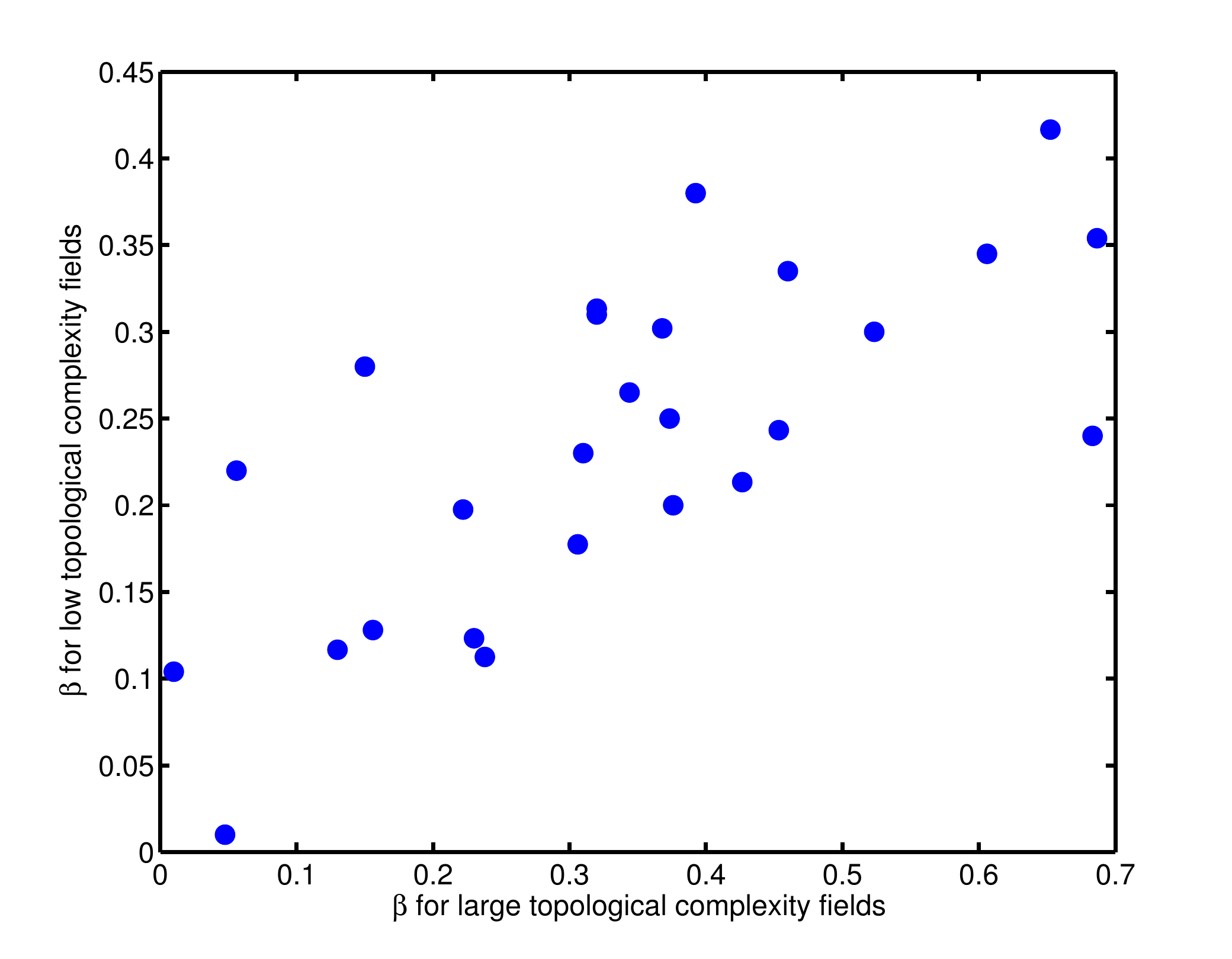} 
   \caption{The parameter $\beta$ for each subject with respect to the topological complexity of the environment.}
   \label{fig:lowHighBeta}
\end{figure}

From the statistical analysis of $\beta$, it is clear that humans have a manifested preference for mapping the topology of the  field, i.e. $\beta<1$. This preference, however, is stronger in some subjects and weaker in other. Since an alternating sequence of low topology complexity maps and high topology complexity maps is used in the experiment, we can verify the consistency of the preference for each subject through computing $\beta$ independently for the two types of fields. Fig. \ref{fig:lowHighBeta} shows the result of this analysis (the correlation between these two data sets is $0.78$).  

We use this consistency to divide the subjects into three equal-sized groups based on their $\beta$ characteristic: bottom third (containing the third of the subjects with the lowest values of $\beta$), middle third and a top third. Each group can be viewed as having a different reconnaissance style, and we will analyze how the differences in these styles manifest through the defined information metrics.

\begin{figure}[htbp] 
   \centering
   \includegraphics[width=400pt]{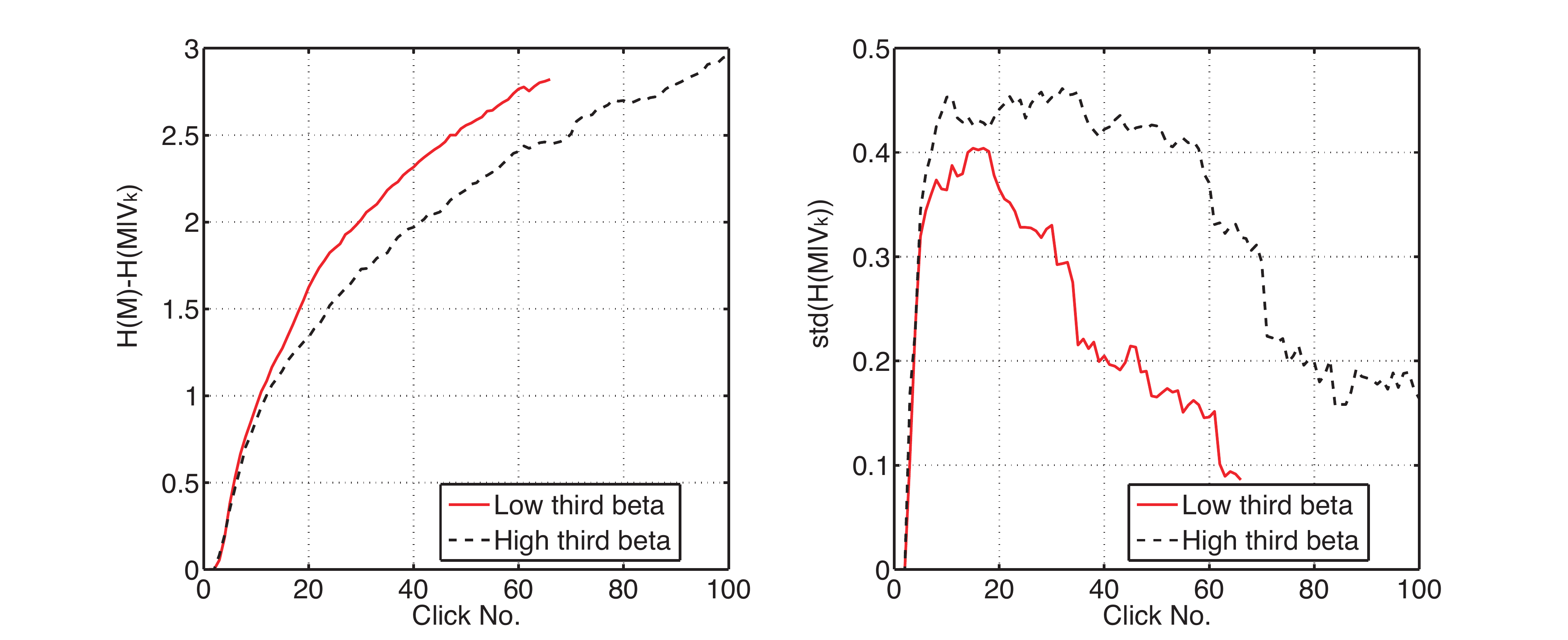}
    \caption{The evolution of \eqref{eq:normEntropy} and its standard deviation as function of the number of clicks.}
   \label{fig:REbeta}
\end{figure}

Fig.~\ref{fig:REbeta} shows 
\begin{eqnarray}
\label{eq:normEntropy}
H(\mathcal{M}) - H(\mathcal{M}|\mathcal{V}_k),
\end{eqnarray}
averaged over separate games played by the subjects in  the bottom third and the top third group for each click $k$. In this case, we have used $H(\mathcal{M})=H(\mathcal{M}|\mathcal{V}_0)$ ($\mathcal{V}_0=\{X\}$)  as a normalization factor to account for the fact that  the different  fields explored by the subjects  have varying topological complexities.  Also, to be able to account for the different lengths of the games, we have adopted standard game duration for each group. That is, the length over which the separate games are averaged in each group corresponds to the duration above which only $10\%$ of the games continue.

There are three observations that can be made: i) the subjects in the bottom third perform more efficient reconnaissance of the topology; ii) the subjects in the top third spend more clicks per area; iii) the standard deviation of the conditional entropy  as a function of the number of  clicks is higher for the subjects that use less topology based feedback. 

Fig.~\ref{fig:entropy} compares the same groups in terms of the entropy of the data induced partition. One can clearly observe that at the expense of spending more clicks per a field, the subjects in the top third  achieve higher data induced partition entropy.

Another interesting observation can be made through Fig. \ref{fig:uniformity}. The graph corresponds to the distance between the maximum possible partition entropy, $\log_2|\mathcal{V}_k|$, and the actual one $H(\mathcal{V}_k)$:
\[
\log_2|\mathcal{V}_k| - H(\mathcal{V}_k).
\]
It also represents a metric of the uniformity of the size of the cells within the data induced partition. This metric initially grows for subjects in both of the groups. This corresponds to a concentrated mapping of isolines at the beginning (isolines associated with a single gradient line), which leads to few small-size cells and a large cell corresponding to the unexplored space.

The initial growth gradually tapers off for the subjects in the bottom third. This can be explained by the fact that they pursue an understanding of the field's topology as opposed to a detailed map of its shape. In other words, they parsimoniously try to acquire a  good picture of the underlying structure induced by the critical points and, as a result, distribute the mapped isolines more uniformly around the search domain. On the other hand, the subjects  in the top third acquire maps with a higher level of detail. They achieve this by methodically mapping higher numbers of isolines per each gradient path.

\begin{figure}[htbp] 
   \centering
      \subfigure[The evolution of the  entropy for the subjects in the bottom third and the high third $\beta$.\label{fig:entropy}]{\includegraphics[width=170pt,height=135pt]{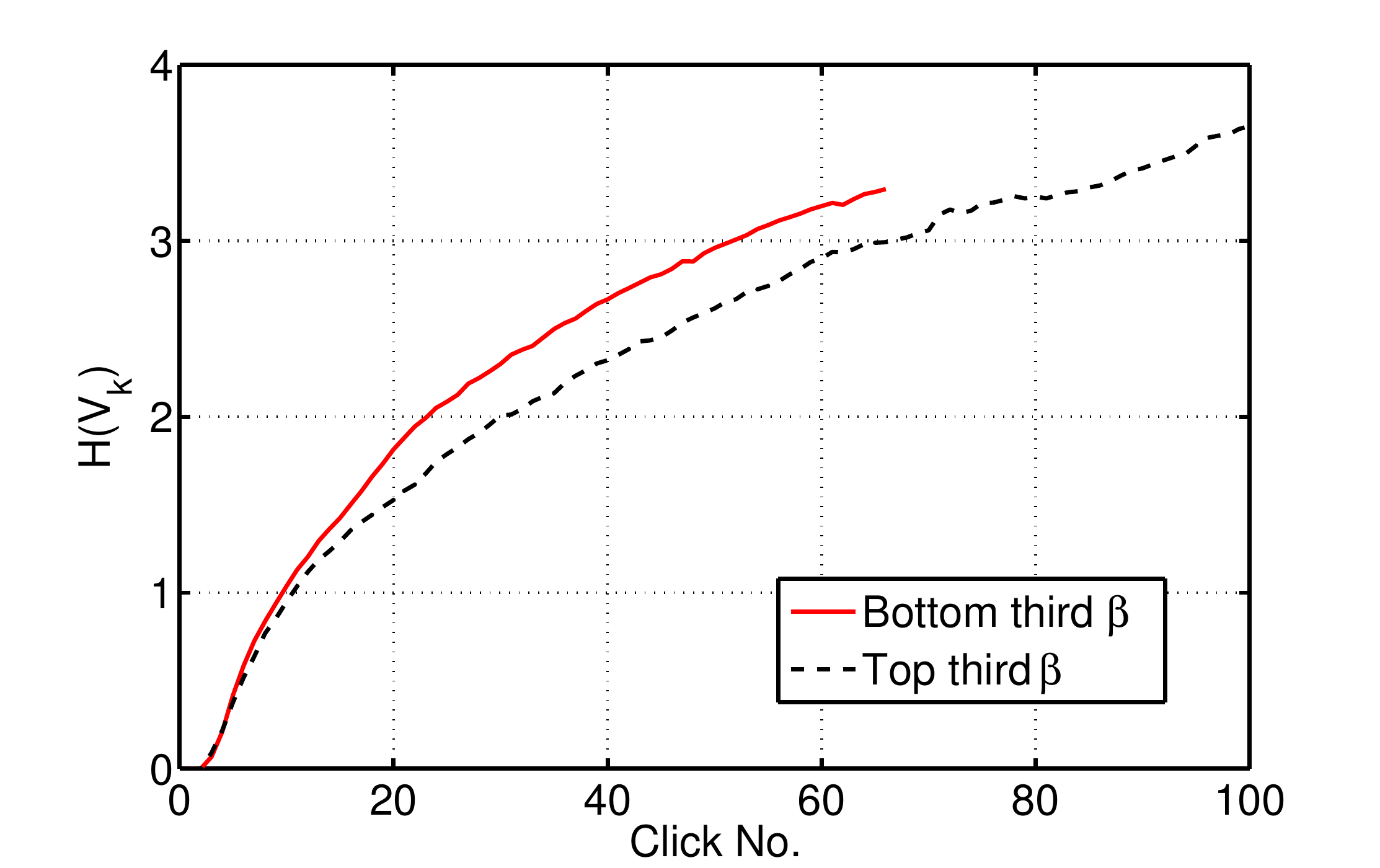} }
     \subfigure[The distance $\log_2|\mathcal{V}_k|-H(\mathcal{V}_k)$ as function of the number of clicks $k$.\label{fig:uniformity}]{\label{subfig:click1}\includegraphics[width=180pt]{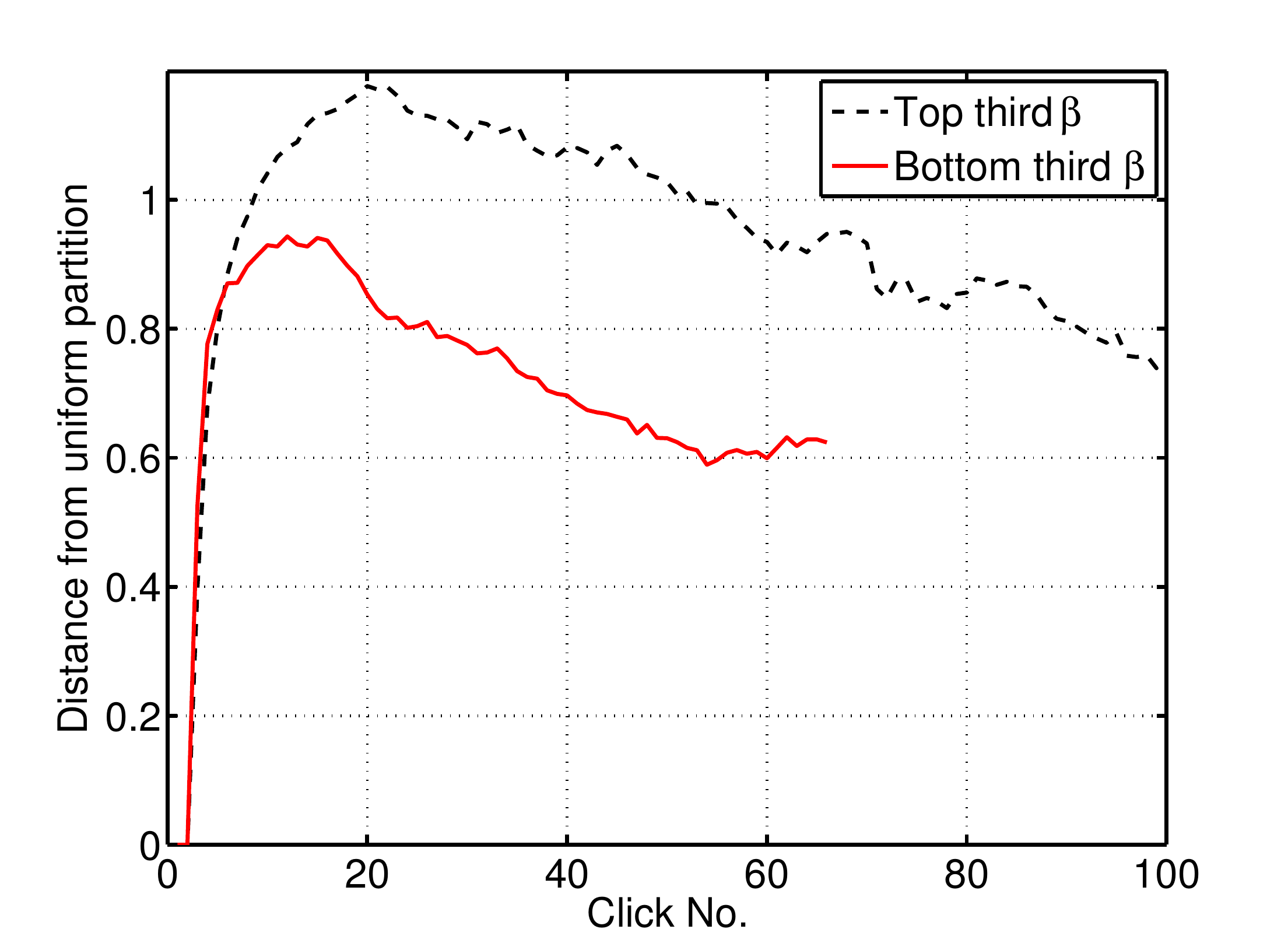} }
   \caption{The evolution of the entropy as a function of the number of clicks.}
   
\end{figure}

\section{Conclusions and Open Problems}
\label{sec:conclusion}
The research described in this paper began as an attempt to understand decision dynamics in the execution of tasks carried out by mixed teams of humans and robotic agents.   The specific focus has been the decisions that are required in carrying out reconnaissance missions--and more specifically in the exploration of unknown scalar fields. The paper has introduced formal models of information content in the kinds of spatially varying unknown fields.   A Shannon-like entropy metric of information content of an unknown field was proposed, and it was shown how this metric could be related to the differential topology---and more specifically to the critical level sets of the field.  In terms of this information metric, we were able to describe the way in which a reconnaissance mission acquired knowledge of the unknown field.  By mapping gradient lines and isolines using the reactive control laws of our previous work (\cite{Baronov:uq,Baronov:2008fk,Baronov:2010aa}), an increasingly fine partition of the reconnaissance domain was obtained.  It is precisely the partition entropy of these partitions whose rate of increase can be used to define how effectively the reconnaissance is progressing.  By defining a distinguished partition associated with the {\it a priori} unknown critical level sets of the field, we have been able to make use of certain conditional entropies to describe how much topological information is being  acquired by the reconnaissance.

In trying to understand the information needed to describe a unknown field defined on a compact domain, the critical level sets have turned out to be surprisingly important.  While they clearly comprise a minimal set of essential qualitative features, the connection established in Theorem 3 with the {\em entropy} of the field came as something of a surprise.  In view of this, the results of Section IV and the algorithm of Section V have been aimed at showing how a reconnaissance mission can be structured so as to climb an {\em information gradient} aimed at the goal of estimating  all critical level sets.

The final part of the paper reports results of an experiment in which human subjects (undergraduate engineering students) acted as mission directors in  a simulated robotic reconnaissance of an unknown scalar field.  A measure of bias toward learning topological features of the unknown field is described, and in terms of this measure, we have observed that all our subjects showed some tendency toward trying to discover topological information---i.e.\ toward trying to learn the location of all critical level sets.  Some subjects exhibited this tendency much more than others, however.  Analysis of the game play showed that subjects who were strongly biased toward accumulating information about the topological characteristics of the field tended to be more parsimonious in terms of the numbers of isolines they mapped.  They also tended to be more consistent from one game to another.  The subjects who were biased toward seeking topological information exhibited a more even balance between numbers of isolines and number of gradient lines mapped.  Subjects who were less biased toward topological information tended to acquire more general detail about the search domain by mapping large numbers of isolines.

 The results we have reported must be viewed as preliminary, but they point to a number of questions that seem worthy of further study. In terms of modeling information acquisition in problems of search and reconnaissance, there are many potentially important extensions that could be pursued. Perhaps the most obvious is to consider fields that vary in either time or space. All the work reported above has treated spatially stationary gaussian fields. In order for the models to capture the key features of real world reconnaissance problems, the correlation lengths of \eqref{eq:potField}  should be allowed to vary over the domain. When this occurs, there are a variety of questions about which little is known. The trade off between metric and topological information will need to be studied carefully in such settings. In particular, if spatially non-stationary random fields are used in experiments with human beings, it will be of  interest to learn what features of the field have greatest attraction for people directing simulated reconnaissance missions. Is there greater interest in finding and mapping the highest peaks or are regions of high variability in the field strength of greater interest?
 
 In any setting in which human decisions in the exploration of unknown fields are studied, it will be important to understand how acquired information provides cues regarding what to explore next. We have explored the way in which knowledge of the field's topological characteristics seems to inform next steps in the reconnaissance. Decision makers may be sensitive to other types of features as well, and research is needed to understand which characteristics of field variability will be most important in guiding a human mission director. 
 
 A significant area that is open for future research is temporally varying fields---such as would be used to model a chemical cloud undergoing both diffusion and convection in the atmosphere. Nothing is currently known about the time variation of critical level sets or the partition entropies that we have studied. Moreover, in terms of decision making, time variation raises an entirely new set of questions for future study. 
 
 Yet another question for the future is to understand how the reconnaissance of an known field should be carried out most effectively when multiple sensor-equipped robotic agents are available. The question is of special interest in the case in which the field is either spatially or temporally varying, but it is of interest in the stationary case as will. 
 
 We conclude by saying that the research presented in this paper has raised more questions than it has answered. At this point our only definitive conclusion is that the process of information acquisition in the exploration of unknown fields is worthy of further study.

\end{document}